\documentclass[11pt]{article}
\usepackage{graphicx} 

\usepackage[ruled,vlined]{algorithm2e}
\usepackage{fullpage}
\usepackage{float}
\usepackage{bbm}
\usepackage{graphicx}
\usepackage{upref}
\usepackage{enumerate}
\usepackage{latexsym}
\usepackage{color,graphics}
\usepackage{comment} 
\usepackage{relsize}
\usepackage{url} 
\usepackage{bm} 
\usepackage{marvosym}
\usepackage{booktabs}
\usepackage{hyperref}
\usepackage{amsmath,amsthm,amssymb,epsfig}
\usepackage{thmtools}
\usepackage{tikz}
\usepackage{xcolor}
\usepackage{authblk}
\usetikzlibrary{arrows.meta,positioning}
\definecolor{qoptdone}{RGB}{38,116,178}

\usepackage{caption}      

\usetikzlibrary{decorations.pathreplacing}
\usepackage{amsfonts}
\usepackage{varioref}

\newtheorem{theorem}{Theorem}[section]
\newtheorem{Lemma}[theorem]{Lemma}

\newtheorem{proposition}[theorem]{Proposition}
\newtheorem{corollary}[theorem]{Corollary}
\newtheorem{fact}[theorem]{Fact}

\newtheorem{definition}[theorem]{Definition}

\newtheorem{remark}[theorem]{Remark}

\newtheorem{result}{Result}

\newcommand{\R}{\mathbb{R}}
\newcommand{\Z}{\mathbb{Z}}
\newcommand{\C}{\mathbb{C}}
\newcommand{\N}{\mathbb{N}}

\newcommand{\remarkref}[1]{\hyperref[#1]{{Remark~\ref*{#1}}}}
\newcommand{\propref}[1]{\hyperref[#1]{{Proposition~\ref*{#1}}}}

\def\bra#1{\mathinner{\langle{#1}|}}
\def\ket#1{\mathinner{|{#1}\rangle}}
\newcommand{\braket}[2]{\langle #1|#2\rangle}

\renewcommand{\part}[2]{\frac{\partial #1}{\partial #2}}

\newcommand{\all}[2]{\begin{align}\label{#2} #1\end{align}}
\newcommand{\al}[1]{\begin{align} #1\end{align}}

\newcount\Comments
\newcommand{\kibitz}[2]{\ifnum\Comments=1\textcolor{#1}{#2}\fi}

\newcommand{\notagbreak}{\notag \\}

\newcommand{\norm}[1]{\lVert#1\rVert}

\newcommand{\thmref}[1]{\hyperref[#1]{{Theorem~\ref*{#1}}}}
\newcommand{\lemref}[1]{\hyperref[#1]{{Lemma ~\ref*{#1}}}}
\newcommand{\remref}[1]{\hyperref[#1]{{Remark~\ref*{#1}}}}
\newcommand{\corref}[1]{\hyperref[#1]{{Corollary~\ref*{#1}}}}
\newcommand{\eqnref}[1]{\hyperref[#1]{{Equation~(\ref*{#1})}}}
\newcommand{\claimref}[1]{\hyperref[#1]{{Claim~\ref*{#1}}}}
\newcommand{\factref}[1]{\hyperref[#1]{{Fact~\ref*{#1}}}}
\newcommand{\defref}[1]{\hyperref[#1]{{Definition~\ref*{#1}}}}
\newcommand{\exampleref}[1]{\hyperref[#1]{{Example~\ref*{#1}}}}
\newcommand{\hypref}[1]{\hyperref[#1]{{Hypothesis~\ref*{#1}}}}
\newcommand{\secref}[1]{\hyperref[#1]{{Section~\ref*{#1}}}}
\newcommand{\chapref}[1]{\hyperref[#1]{{Chapter~\ref*{#1}}}}
\newcommand{\apref}[1]{\hyperref[#1]{{Appendix~\ref*{#1}}}}
\newcommand{\tabref}[1]{\hyperref[#1]{{Table~\ref*{#1}}}}

\title{Quantum algorithms for orthogonal polynomial transforms}
\author[1]{Anupam Prakash\thanks{\texttt{anupam.prakash@jpmchase.com}}}
\author[1]{Shree Hari Sureshbabu}
\author[1]{Dylan Herman}
\author[1]{Shouvanik Chakrabarti\thanks{\texttt{shouvanik.chakrabarti@jpmchase.com}}}
\affil[1]{Global Technology Applied Research, JPMorganChase, New York, NY 10001, USA}
\date{\today}

\begin{document}

\maketitle

\begin{abstract}
    A quantum orthogonal polynomial transform (QOPT) is an algorithmic primitive that maps a superposition of standard basis states $\sum_k \alpha_{k} \ket{k}$ coherently to a basis of normalized univariate polynomials orthogonal with respect to a probability measure $\mu(x)$. We provide new discrete and continuous QOPTs for polynomial families in the Askey scheme extending the results for the quantum Hermite transform (Jain et al., STOC'26). Our efficient QOPT algorithms require time $O(\text{polylog}(N, 1/\epsilon))$, where $N$ is the grid size and $\epsilon$ is the error, for the discrete Charlier, Meixner and Krawtchouk transforms and for the integer-order Laguerre transform in the continuous setting. The efficient QOPTs are obtained by uncovering the links between Askey scheme polynomials and Gaussian quantum optical gates and developing a compilation framework for $SU(2)$ and $SU(1,1)$ optical gates on Cartesian grids extending the framework developed by Iyer et al. (arXiv:2602.15180).
    Further, we reduce the continuous Jacobi transform to the discrete Hahn transform and provide an $O\!\left(N\operatorname{polylog}((N+\alpha+\beta+1)/\epsilon)\right)$-time Hahn transform, a quadratic speedup over the naive implementation. This is based on a more efficient compilation of the Clebsch--Gordan transform for coupling $\mathrm{SU}(2)$ representations with spins $(j_{1}, j_{2})$.
    Finally, we develop a new framework for Laguerre transforms for all orders $\nu>0$ by fast-forwarding the corresponding radial oscillator using a 3-term chirp decomposition and an efficient algorithm for the Quantum Hankel Transform on a logarithmic grid.

\end{abstract}

\tableofcontents

\section{Introduction}
Transforms into a basis of orthogonal functions are among the most versatile
primitives in computational mathematics. The discrete Fourier transform
reshaped signal processing, numerical linear algebra, and algorithm design once
the fast Fourier transform reduced its cost from $O(N^2)$ to $O(N\log N)$
\cite{cooleytukey1965}. Its polynomial counterparts, which are the expansions in Chebyshev, Legendre, Hermite, and Laguerre bases, are the computational backbone of spectral methods for differential equations, Gaussian quadrature, and function approximation \cite{trefethen2013, shen2011spectral}. The development of fast orthogonal polynomial transforms (OPTs) has consequently been a sustained program in numerical analysis. Examples include the near-linear algorithm for Legendre expansions by Alpert and Rokhlin \cite{alpertrokhlin1991}, spherical-harmonic transforms on the sphere by Driscoll and Healy \cite{driscollhealy1994}, and algorithms for general Vandermonde-like matrix families obeying a three-term recurrence by Potts, Steidl, and Tasche \cite{pottssteidltasche1998}.
These algorithms achieve $O(N\log N)$ or $O(N\log^2 N)$ complexity precisely
because the underlying polynomials share structural ties with the Fourier
transform, exploited through fast cosine transforms and Toeplitz-Hankel
factorizations \cite{townsendwebbolver2018}.

Quantum computation promises to sharpen this speedup from polynomial to
exponential: an orthogonal basis change encoded as a unitary can, in favorable
cases, be realized by a circuit whose size is poly-logarithmic in the
dimension. The quantum Fourier transform is the archetype---$O(\log^2 N)$ gates
against the classical $O(N\log N)$ \cite{nielsenchuang2010}---and it underwrites
a large fraction of the exponential quantum speedups we know, from phase
estimation to Shor's algorithm. It is therefore natural to ask which orthogonal
polynomial transforms admit comparably efficient quantum realizations, and what
algorithmic capabilities such primitives would unlock. The answer is not a
matter of transcription. The efficiency of the quantum Fourier transform rests
on the recursive factorization of the Fourier matrix, and the unitaries
underlying generic orthogonal polynomial transforms possess no such structure to
exploit; identifying the structure that does enable a fast quantum transform is
the central question of this work.

Let us examine what additional structure is needed for an orthogonal family to have a fast transform. 
Every positive measure with finite moments yields, by Gram--Schmidt on the
monomials, a unique orthonormal family, and Favard's theorem places these
families in bijection with three-term recurrences
\begin{align}
  x\,p_k(x) = b_k\,p_{k+1}(x) + a_k\,p_k(x) + b_{k-1}\,p_{k-1}(x),
  \qquad b_k>0,
\end{align}
equivalently with tridiagonal Jacobi matrices whose spectral measure is $\mu$
\cite{ismail2005}. The recurrence is a \emph{complete}
characterization of orthogonality---and precisely for that reason it cannot, on
its own, accelerate the transform: for a generic measure the resulting change of
basis is a dense matrix, and applying it costs $\Theta(N^2)$ operations. A fast
transform, classical or quantum, must therefore exploit structure beyond
orthogonality itself.

The organizing principle that we uncover is that the orthogonal polynomial families
admitting fast quantum transforms are, to a large extent, those of the
\emph{Askey scheme} of hypergeometric orthogonal polynomials
\cite{koekoekleskyswarttouw2010} (Figure~\ref{fig:askey}). Beyond the three-term recurrence shared by all univariate
orthogonal polynomial families, the orthogonal polynomial families in the Askey scheme are distinguished by the existence of a second-order
differential or difference operator such that
\al{
  \mathcal{L}\,p_{k}(x) = \lambda_{k}\,p_{k}(x).
}
The polynomial states $\ket{\mathbf{k}}$ are in fact the eigenstates of a single second-order differential or difference operator $\mathcal{L}$  \cite{askeywilson1985, koekoekleskyswarttouw2010}. The 
operator $\mathcal{L}$ specializes to a differential operator for the continuous families and a difference operator for the discrete families
\cite{nikiforovsuslovuvarov1991}.

The relevance of the
Askey scheme for our purposes is representation-theoretic rather than
combinatorial: in the classical harmonic analysis developed by Askey,
Koornwinder, Vilenkin, Klimyk, and others, these families arise as the matrix
elements and spherical functions of the rank-one Lie groups---the
Heisenberg--Weyl group, $SU(2)$, and $SU(1,1)$---with the second-order operator
$\mathcal{L}$ distinguishing each family realized as the Casimir invariant and
the three-term recurrence encoding the action of the raising and lowering
generators \cite{vilenkinklimyk1991, koornwinder1984, nikiforovsuslovuvarov1991}.
This is precisely the structure that quantum computation is equipped to exploit.
Seen in this light, the search for quantum speedups has largely been a search
for efficiently implementable, group-adapted changes of basis: the quantum
Fourier transform is harmonic analysis on abelian and, in the hidden-subgroup
program, nonabelian finite groups \cite{childsvandam2010}, and the Schur--Weyl
transform is harmonic analysis on the pair $(S_n, SU(d))$
\cite{bacon2006efficient}. The rank-one groups underlying the Askey scheme have,
with isolated exceptions, been absent from this program. The reason is
structural: their natural carrier is the harmonic oscillator rather than the
qubit, and the transforms they generate are \emph{continuous}, modeled on
bosonic modes and Gaussian optical operations rather than on discrete gates. It
is exactly these transforms that we bring within reach of gate-based quantum
computation, extending the harmonic-analytic view of quantum algorithms from the
finite groups to the rank-one continuous groups.

An efficient quantum Hermite transform was introduced recently in \cite{jain2025efficient}, where an efficient change of basis from the standard basis to a basis of Hermite polynomials defined on a discrete grid is carried out with complexity similar to the quantum Fourier transform. Unlike the classical algorithm, the efficient QOPT relies on the fast forwarding of the quantum harmonic oscillator \cite{wei1963lie}. 
This motivates the following definition of a \emph{quantum orthogonal polynomial transform} (QOPT).
\begin{definition} \label{def:qopt}
 (QOPT) Let $(\mu(x), p_{k}(x))$ be a base measure and a family of orthogonal polynomials where $\mu(x)$ is supported on $\R$ in the continuous and on a lattice $\Z$ in the discrete case. The QOPT is a quantum algorithm for generating a superposition of orthonormal polynomial states starting with a superposition in a known basis,
 \al{
  \sum_{k} \alpha_{k} \ket{k} \leftrightarrow  \sum_{k} \alpha_{k} \ket{\bf{k}}.
 }
The state $\ket{\bf{k}} = \frac{1}{\sqrt{Z}}\sum_{x} \sqrt{\mu(x)}p_{k}(x)\ket{x}$ has amplitudes proportional to the corresponding orthonormal polynomials.
\end{definition}
\noindent An implementation of a continuous QOPT in the circuit model must in addition specify the grid for $x$. Further, the transform implemented in the circuit model does not preserve exact orthogonality and is applicable up to an energy cutoff specified in the implementations. These features are inherited due to representing continuous variables on a discrete grid and are unavoidable \cite{jain2025efficient}.

\definecolor{qoptexp}{RGB}{38,116,178}    
\definecolor{qoptsub}{RGB}{200,120,20}    

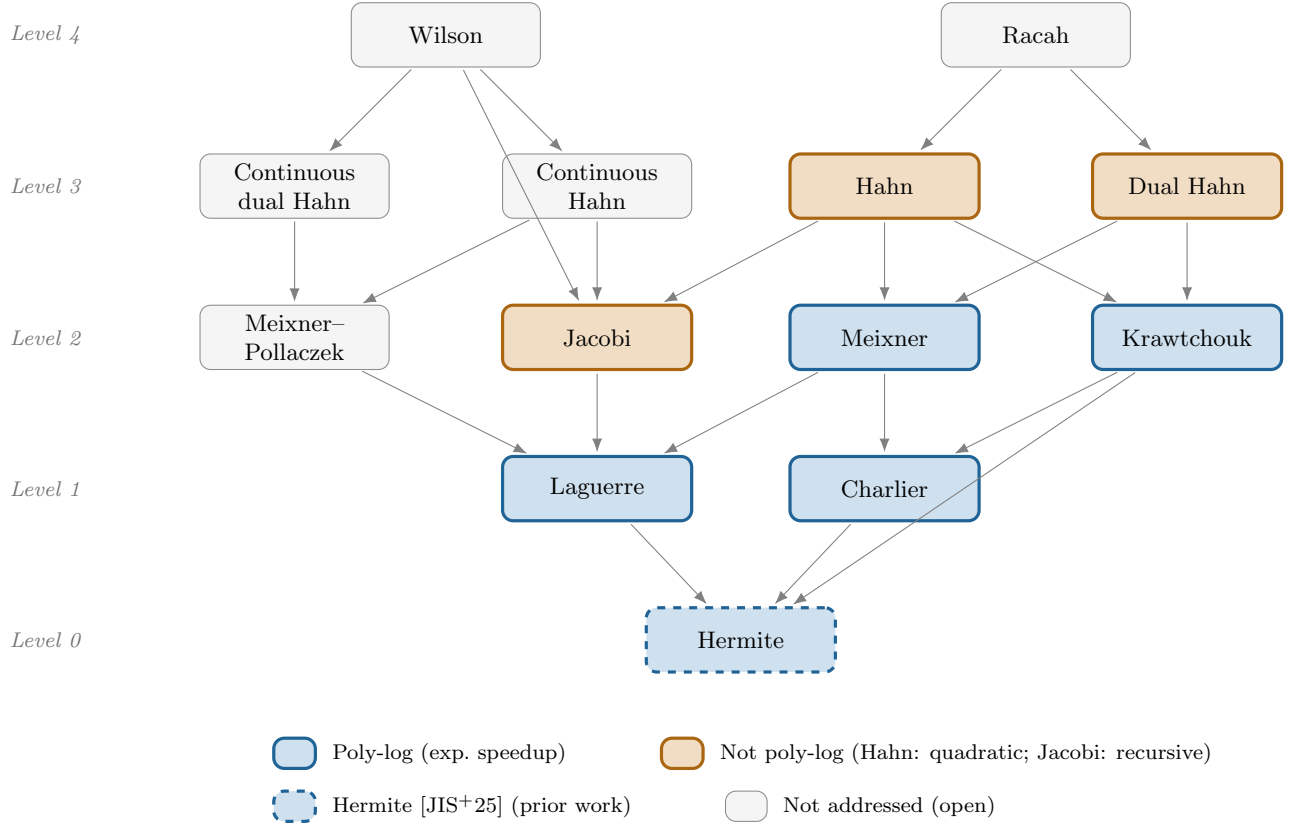
\begin{figure}[htp]
\centering
\begin{tikzpicture}[
  fam/.style   = {draw, rounded corners, align=center, font=\footnotesize,
                  minimum width=2.5cm, minimum height=0.85cm, inner sep=2pt},
  exp/.style   = {fam, fill=qoptexp!22, draw=qoptexp!85!black, very thick},
  sub/.style   = {fam, fill=qoptsub!25, draw=qoptsub!85!black, very thick},
  prior/.style = {fam, fill=qoptexp!22, draw=qoptexp!85!black, very thick, dashed},
  open/.style  = {fam, fill=black!4, draw=black!45},
  lvl/.style   = {font=\scriptsize\itshape, text=black!55},
  leg/.style   = {font=\scriptsize},
  lim/.style   = {-{Latex[length=1.8mm]}, draw=black!50,
                  shorten >=1pt, shorten <=1pt},
  x=1cm, y=1cm]

  \node[lvl] at (-2.9,8) {Level 4};
  \node[lvl] at (-2.9,6) {Level 3};
  \node[lvl] at (-2.9,4) {Level 2};
  \node[lvl] at (-2.9,2) {Level 1};
  \node[lvl] at (-2.9,0) {Level 0};

  \node[open]  (wil) at (2.4,8)  {Wilson};
  \node[open]  (rac) at (10.2,8) {Racah};

  \node[open]  (cdh) at (0.4,6)  {Continuous\\dual Hahn};
  \node[open]  (ch)  at (4.4,6)  {Continuous\\Hahn};
  \node[sub]   (hah) at (8.2,6)  {Hahn};
  \node[sub]  (dh)  at (12.2,6) {Dual Hahn};

  \node[open]  (mp)  at (0.4,4)  {Meixner--\\Pollaczek};
  \node[sub]   (jac) at (4.4,4)  {Jacobi};
  \node[exp]   (mei) at (8.2,4)  {Meixner};
  \node[exp]   (kra) at (12.2,4) {Krawtchouk};

  \node[exp]   (lag) at (4.4,2)  {Laguerre};
  \node[exp]   (cha) at (8.2,2)  {Charlier};

  \node[prior] (her) at (6.3,0)  {Hermite};

  \draw[lim] (wil) -- (cdh);
  \draw[lim] (wil) -- (ch);
  \draw[lim] (wil) -- (jac);
  \draw[lim] (rac) -- (hah);
  \draw[lim] (rac) -- (dh);

  \draw[lim] (cdh) -- (mp);
  \draw[lim] (ch)  -- (mp);
  \draw[lim] (ch)  -- (jac);
  \draw[lim] (hah) -- (jac);
  \draw[lim] (hah) -- (mei);
  \draw[lim] (hah) -- (kra);
  \draw[lim] (dh)  -- (mei);
  \draw[lim] (dh)  -- (kra);

  \draw[lim] (mp)  -- (lag);
  \draw[lim] (jac) -- (lag);
  \draw[lim] (mei) -- (lag);
  \draw[lim] (mei) -- (cha);
  \draw[lim] (kra) -- (cha);
  \draw[lim] (kra) -- (her);
  \draw[lim] (lag) -- (her);
  \draw[lim] (cha) -- (her);

  \begin{scope}[shift={(0.4,-1.5)}]
    \node[exp,   minimum width=0.55cm, minimum height=0.4cm] (l1) at (0,0) {};
    \node[leg, right=2pt of l1] (t1) {Poly-log (exp.\ speedup)};
    \node[sub,   minimum width=0.55cm, minimum height=0.4cm, right=1.1cm of t1] (l2) {};
    \node[leg, right=2pt of l2] (t2) {Not poly-log (Hahn: quadratic; Jacobi: recursive)};

    \node[prior, minimum width=0.55cm, minimum height=0.4cm] (l3) at (0,-0.7) {};
    \node[leg, right=2pt of l3] (t3) {Hermite \cite{jain2025efficient} (prior work)};
    \node[open,  minimum width=0.55cm, minimum height=0.4cm, right=1.1cm of t3] (l4) {};
    \node[leg, right=2pt of l4] {Not addressed (open)};
  \end{scope}

\end{tikzpicture}
\caption{The Askey scheme of hypergeometric orthogonal polynomials, organized
top-to-bottom by decreasing level; arrows denote limit transitions between
families. Blue families admit poly-logarithmic QOPTs with exponential speedups
(Charlier, Krawtchouk, Meixner, Laguerre). Orange families are realized here but
are \emph{not} poly-logarithmic: the Hahn transform attains a quadratic speedup
over the naive $O(N^2)$ implementation, and the Jacobi transform can be implemented using the Hahn in a restricted setting. The Hermite transform at the base is the
foundational QOPT of \cite{jain2025efficient} on which the digital
implementations are built.}
\label{fig:askey}
\end{figure}
In this paper we realize efficient QOPTs for the Charlier, Meixner, and Krawtchouk polynomials in the discrete setting and for Laguerre polynomials in the continuous setting. We identify an efficient Hahn transform as a bottleneck for an efficient Jacobi transform for integer-valued parameters $\alpha, \beta$. We obtain a quadratic speedup over the naive implementation of the Hahn transform; however, this falls short of providing exponentially efficient Hahn and Jacobi transforms. The levels of the Askey scheme are organized by the number of free parameters in the hypergeometric function defining the corresponding polynomials as illustrated in Figure \ref{fig:askey}. The Wilson, Racah, continuous Hahn and continuous dual Hahn families associated with continuous Jacobi transforms and scattering-type problems are not addressed by the techniques developed in this work. The Meixner–Pollaczek family lies within reach of $SU(1,1)$ constructions and is a natural extension that we do not pursue here.

The distinction between the efficient QOPTs and the non-efficient ones can be explained in terms of the underlying algorithm. In the efficient case, the evolution $e^{-i\mathcal{L}t}$ can be \emph{fast-forwarded}: a Wei--Norman factorization into a fixed product of exponentials of the group generators reduces the simulation to $O(\mathrm{poly}\log N)$ gates, independent of the evolution time
\cite{wei1963lie}. For the Charlier,
Krawtchouk, Meixner, and Laguerre families the transform is fast-forwardable
directly, and we obtain QOPTs with gate complexity $O(\mathrm{poly}\log N)$,
matching the quantum Fourier transform up to polylogarithmic factors and giving
an exponential speedup over the $O(N\log^{2} N)$ cost of the corresponding
classical fast transforms. For the Hahn and univariate
Jacobi families no such direct fast-forwarding is known for the underlying P\"oschl-Teller Hamiltonian; there we obtain a
recursive speedup, quadratic in the discrete case. However, even with the qudratic speedup for the Hahn transforms, speedups for mltivariate Jacobi transforms analogous to the recursive speedps Schur--Weyl transform inspired speedups are plausible. \cite{bacon2006efficient}. 

The base measures for the orthogonal polynomial families for which we have efficient QOPTs with an exponential speedup over classical are commonly arising probability distributions, namely the Poisson, negative binomial, binomial, and gamma distributions. The Hahn and Jacobi transforms are orthogonal with respect to the discrete hypergeometric and beta distributions. Thus all the QOPTs discussed in this paper have a probabilistic interpretation.

We discuss next the conceptual links between the Askey scheme and efficient QOPTs, these links are primarily representation-theoretic. The Askey scheme organizes the harmonic analysis
of the rank-one groups: the Heisenberg--Weyl group, $SU(2)$, and $SU(1,1)$. The orthogonal polynomial families arise either as the matrix elements of the generators of these groups, as second-order operators $\mathcal{L}$ describing the Casimir invariant/coupling structure of the group or recoupling operators at the higher levels of the scheme. The three-term recurrence encodes the action of the raising and lowering
generators in the appropriate setting. The link to quantum computing follows as the rank-one groups are generated by the elementary
Gaussian gates in continuous-variable quantum computation.

These connections are particularly clear in the discrete OPT setting. 
Displacements
generate the Heisenberg--Weyl group, beam splitters generate $SU(2)$, and
single-mode squeezers generate $SU(1,1)$. This can be made precise using Perelomov's coherent state formalism \cite{perelomov1986}. We use this observation to obtain discrete QOPTs for the Charlier, Krawtchouk, and Meixner polynomials. The QOPTs for the Krawtchouk transform have appeared recently in the quantum computing literature \cite{guan2026quantumKrawtchouktransformusing}; the Krawtchouk transform is also implicit in the analysis of quantum algorithms like permanent sampling \cite{aaronson2011computational} and the decoded quantum interferometry (DQI) algorithm \cite{marwaha2026complexitydecodedquantuminterferometry}. It remains an open question to obtain quantum applications for the Charlier and Meixner transforms that we demonstrated to be closely related to the Krawtchouk transform. The Krawtchouk transform is adapted to $SU(2)$ and thus to qubits; the other transforms are more natural in a continuous-variable setting.

The remaining discrete transform Hahn transform can be viewed as a special case of the Schur--Weyl transform when restricted to Young diagrams with at most two rows. We provide an $O(N \text{polylog} N)$ quantum algorithm for the Hahn transform, improving upon the naive $O(N^{2})$ algorithm. The algorithm uses the identification between the Hahn polynomials and Clebsch--Gordan coefficients for $SU(2)$ that has so far not been used for the design of quantum algorithms. 

The continuous Laguerre and Jacobi orthogonal polynomial families, on the other hand, are related to coupled $SU(2)$ and $SU(1,1)$ representations, which can be equivalently described in terms of multi-mode isotropic oscillators. Unlike a single harmonic oscillator, which has integer eigenvalues and one-dimensional eigenspaces, a multivariate quantum harmonic oscillator exhibits degeneracy in its spectrum. However, the multi-oscillator system is superintegrable: there exist commuting
Hamiltonians whose simultaneous diagonalization resolves the Hilbert space into
one-dimensional eigenspaces \cite{miller2013superintegrability}.
For the two-mode harmonic oscillator, the Hamiltonian and the angular momentum operator $L_{z}$ provide one such system of commuting operators. The eigenfunctions split as Laguerre polynomials in the radial coordinate and
Fourier modes in the angular coordinates. We leverage this decomposition to obtain an efficient QOPT for the one-parameter family of Laguerre polynomials with integer parameters. The additional ingredient for the efficient Laguerre transform is the efficent change of basis from cartesian to circular modes implemented via a circular beam splitter, this arises due to the Schwinger $SU(1,1)$ representations being realized as fixed number difference sectors of the isotropic oscillator. 

If we try to extend the same approach to the two-parameter Jacobi polynomial family, we find that the 4-dimensional super-integrable Smorodinsky--Winternitz oscillator system consisting of coupled radial oscillators has an uncoupled joint radial and angular momentum basis for which the Jacobi polynomials appear as eigenfunctions. The Hahn transform appears as the overlap between the uncoupled and the coupled bases in this setting. An efficient Hahn transform and an efficient procedure to measure the angular quantum number $l$ remain the bottlenecks for obtaining an efficient Jacobi transform.

\subsection{Discrete QOPT Results}
The main contribution of this work is the development of a single unifying framework that implements efficient quantum orthogonal polynomial transforms for a broad set of orthogonal polynomial families in the Askey scheme. However, it is helpful to consider discrete and continuous orthogonal polynomial transforms separately in order to describe the main results as the algorithmic mechanisms arising in the two cases are distinct. 

The discrete QOPTs that we obtain in this work are those for the Charlier, Meixner, Krawtchouk and Hahn polynomials. These polynomials are supported on a discrete lattice and are orthogonal with respect to the Poisson, negative binomial, binomial and hypergeometric weights respectively. A unified perspective on these polynomial families is obtained by viewing them as solutions to the discrete hypergeometric equation as described in greater detail in Section \ref{sec:dqopt}. 

The Charlier, Meixner and Krawtchouk polynomial families are degree $1$ solutions of the discrete hypergeometric equation and further can be identified as the matrix elements of Gaussian linear optical gates. The identification is made using Perelomov's coherent state framework for the Heisenberg group, $SU(1,1)$ and $SU(2)$ respectively \cite{perelomov1986} and is explained further in section \ref{pmorov}. The QOPTs are then reduced to the implementation of Gaussian and linear optical gates with spatial degrees of freedom represented by a discrete grid ~\cite{jain2025efficient}. The choice of the quantum optical operation, such as Displacement, Squeezing, Beam-Splitter action, phase-space Rotation, selects the polynomial family. 

The above three discrete transforms can each be obtained using a single Gaussian gate. However, it is helpful to have a framework for simulating compositions of Gaussian gates, as this also allows phase estimation and resolution into the corresponding eigenbasis of orthogonal polynomials. Such a framework was given in \cite{iyer2026efficientquantumcircuitshighdimensional}. Here we generalize it to include compositions of gates from a Gaussian gate set that includes squeezing gates and provide a better analysis in terms of some of the parameters.

\begin{result}[Informal Theorem \ref{thm:digital_implementation_gaussian_optical_gates}]
    Let $\ket{\psi}$ be an
arbitrary normalized state on $K$ optical modes each with energy at most $N$. Let
$\widehat{U} = \widehat{W}_R \cdots \widehat{W}_1$ be a product of $R$
elementary Gaussian gates, where each $W_j$ can be broken into $M$ bounded evolution time gates. Then, for any target error $\epsilon < 1$, and taking $N = \Theta\!\left( \log(1/\epsilon) + \log K + RM \right)$, there is a digital quantum circuit $\overline{U}$ acting on
$\mathcal{O}(\textup{poly}(\log(K), \log(N), \log(1/\epsilon), RM))$ qubits and using
$\mathcal{O}\!\left( (RM \log(KN/\epsilon))^3 \right)$ one- and two-qubit gates
that $\epsilon$-approximates $\widehat{U}$ on inputs restricted to the truncated subspace.
\end{result}

The improvement over the analysis in \cite{iyer2026efficientquantumcircuitshighdimensional} concerns the dependence on $K$. It allows the number of modes to be non-constant because it suppresses the doubly exponential error term $L^K\exp(-\gamma L)$ in the error analysis. Our analysis may therefore provide improvements for large $K$.

The informal statement for the quantum Charlier transform 
is given below. The theorem statement serves as a template for QOPTs that are poly-logarithmic in the grid size for both discrete and continuous families. The underlying asymptotics for $L, N$ are common to all the cases and are obtained from the reduction to the quantum Hermite transform \cite{jain2025efficient}, arising from the Plancherel-Rotach/WKB asymptotics that determine the energy cutoffs in terms of grid size.

\begin{theorem}[Quantum Charlier transform] \label{intro:qct}
Fix $\mu>0$, target error $\epsilon<1$. Let
\begin{align}
    \ket{\psi}=\sum_{m\in\{0,\ldots,N-1\}}a_{m}
    \ket{m},
    \qquad \sum_{m}|a_{m}|^2=1,
\end{align}
be an arbitrary state in the effective qubit space $\mathcal{H}_{N}$. Set the grid size to
\begin{align*}
L=\Theta\left((2^{\sqrt{\mu}}N)^{9/4}/\epsilon^{13/4} \right).
\end{align*}
Then there is a digital quantum circuit $\overline{U}$ on $\mathcal{O}(\log(L))$ qubits. Its gate complexity is
\begin{align*}
\mathcal{O}\left(\textup{poly}(\mu, \log(1/\epsilon))\right),
\end{align*}
and it satisfies
\begin{align*}
    \lVert (U_{\mathrm{Ch}}(\mu) - \overline{U})P_{N}\rVert < \epsilon,
\end{align*}
where $U_{Ch}(\mu)$ is the coherent quantum Charlier transform and $P_N$ is the orthogonal projection of  $\mathcal{H}_L$ onto  $\mathcal{H}_N$. 
 \end{theorem} 
\noindent The statement for the quantum Charlier transform holds verbatim for the Krawtchouk and Meixner transforms that are efficient QOPTs in the same sense with the displacement gates replaced by beam splitter and squeeze gates respectively. 

The Hahn transform is the discrete QOPT that corresponds to degree 2 solutions of the discrete hypergeometric equation.
It is specified by two distinct integer parameters $(\alpha, \beta)$. 
The matrix entries of the Hahn transform are related to the Clebsch-Gordan (CG) coefficients. This allows us to implement the Hahn transform using a CG and inverse CG cascade like that used for the Schur transform \cite{bacon2006efficient}, along with some Racah recoupling steps. We obtain a linear time algorithm for the two parameter Hahn transform, which attains a quadratic speedup over the naive $O(N^{2})$ implementation and remains poly-logarithmic in $(\alpha + \beta)$ for the allowed values of these parameters. 

\begin{theorem}[Quantum Hahn transform]
\label{thm:quantum_hahn_informal}
Let $\alpha,\beta,N\in\Z_{\geq0}$ and $0<\epsilon<1$.  Algorithm
\ref{alg:hahn_cg_qubit_registers} implements the Hahn state-synthesis map, and
its adjoint implements the forward analysis map, to operator-norm error at most
$\epsilon$.
The circuit uses $\mathcal{O}\!\left(
            N\operatorname{polylog}\!\left(\frac{N+\alpha+\beta+1}{\epsilon}\right)
        \right)$ one and two-qubit gates along with $\mathcal{O}\!\left(
            N+\operatorname{polylog}\!\left(\frac{N+\alpha+\beta+1}{\epsilon}\right)
        \right)$ qubits.

\end{theorem}

The construction of the quantum circuit for the Hahn transform uses the
fact that the ordered fixed-$M$ block of the Hahn transform differs from the CG
transform only by two known diagonal sign gauges. The circuit realizes the first 
$SU(2)$ representation $V_{j_2}$ as the symmetric (Dicke) subspace of $T=N$ spin-$\tfrac12$
constituents, and couples it to $V_{j_1}$
through a cascade of elementary two-level CG rotations $V_{K,\nu}$, and
then applies Racah recoupling to move the bracketing onto the tree in which
$V_{j_1}$ couples to the single symmetric irrep $j_2$. Every rotation angle is
computed reversibly from the running labels, so large $\alpha+\beta$ enlarges
only arithmetic bit-width, not the constituent register.

In particular, the quantum Hahn transform has no additional polynomial dependence on $|\alpha-\beta|$. The parameter magnitudes affect the reversible arithmetic only through their logarithmic bit width. The quantum Hahn transform by itself does not have any consequences for the more general Schur transform. It corresponds to the $d=2$ case and matches the resources for the optimal known algorithm, which is $O(d^{4} n)$ \cite{burchardt2025high}. This concludes our discussion of the main results for discrete QOPTs. To summarize, we obtain poly-logarithmic dependence on the grid size for the degree 1 transforms using the generalized coherent state formalism and a quadratic speedup over the naive implementation for the two-parameter Hahn transform using a cascade similar to the Schur transform.

\subsection{Continuous QOPT Results}
Let us now switch to a continuous setting. By a classical result of Bochner, the orthogonal polynomial solutions of the continuous hypergeometric equation are precisely the Hermite, Laguerre and Jacobi polynomials. The Laguerre and Jacobi transforms are the degree 1 and 2 analogs of the discrete polynomial transforms discussed above. The continuous QOPTs have an additional feature: the second-order hypergeometric equation reduces to a Schr\"odinger operator after application of a suitable ground state transform. The theory is discussed in further detail in Section \ref{sec:cqopt}.

The Laguerre and Jacobi polynomials are closely related to the separation of variables in higher-dimensional isotropic oscillators. The eigenfunctions of the two-dimensional isotropic quantum harmonic oscillator have a tensor product decomposition in Cartesian as well as polar eigenbases. The Laguerre polynomials arise as the radial part of the polar eigenfunctions, with the angular parts organized by the Fourier transform. That is, the Laguerre polynomial $L_{n}^{|m|}(r^{2})$ embeds in the angular momentum sector $m$ for the 2D isotropic harmonic oscillator and can be viewed as the eigenfunction of a radial oscillator. The Jacobi polynomials, being the continuous counterparts of the Hahn polynomials, arise as eigenfunctions of coupled radial oscillators. If the radial oscillators being coupled have identical $m$ parameters, the Jacobi polynomials reduce to spherical harmonics in dimension 3.

As the Laguerre polynomials are associated with a two-dimensional polar grid, there are two possible implementations of the Laguerre transform.
The first method works for integer-valued $m$ parameters and reduces to fast-forwarded oscillator evolution on the Cartesian grid used by the quantum Hermite transform. The informal theorem for the quantum Laguerre transform using Cartesian grids is given below.

\begin{theorem}[Cartesian-grid Laguerre transform from the QHT]
\label{thm:laguerre-cartesian-qht-informal}
Let $m\in\Z$, $N\geq1$, and $0<\epsilon<1$.  Put
\begin{align}
    N_{\rm Cart}=2N+|m|,
    \qquad
    N_{\rm F}=\Theta\!\left(N_{\rm Cart}+\log(1/\epsilon)+1\right),
\end{align}
where the implicit constant is chosen as in
Theorem~\ref{thm:digital_implementation_gaussian_optical_gates}.
There is a power-of-two Cartesian grid with
\begin{align}
    L_{\rm C}=O\!\left(N_{\rm F}^{9/4}\epsilon^{-13/4}\right)
\end{align}
points per coordinate and a digital circuit
$\mathsf{LagCart}_{m,N,\epsilon}$ such that, for every normalized
$\sum_{n<N}a_n\ket n$,
\begin{align}
 \left\|
 \mathsf{LagCart}_{m,N,\epsilon}
 \left(\sum_{n<N}a_n\ket n\ket0\ket0_{\rm work}\right)
 -\left(\sum_{n<N}a_n
 \ket{\Phi_{n,m}^{\mathrm{Cart},(L_{\rm C})}}\right)
 \ket0_{\rm work}
 \right\|_2\leq\epsilon.
 \label{eq:laguerre-cartesian-qht-guarantee-informal}
\end{align}
The gate count and live workspace are polynomial in
$\log(N_{\rm F}+1)$ and $\log(1/\epsilon)$.
\end{theorem}
Further, in continuous OPTs we present a framework for implementing Laguerre transforms on a logarithmic grid and identify the bottlenecks for implementing an efficient Jacobi transform.

\subsection{Summary and related work}
Applying this identification to the optical realization of each family yields the quantum orthogonal polynomial transforms of \tabref{tab:transforms}. We organize
them by the support of the underlying orthogonality measure into discrete
transforms, whose measure is supported on a lattice (Charlier, Krawtchouk,
Meixner, Hahn), and continuous transforms, whose measure is supported on an interval or half-line (Laguerre, Hankel, Jacobi), together with their
multidimensional extensions.

\begin{table}[htp]
    \centering
    \begin{tabular}{l l l}
        \toprule[0.5mm]
        \textbf{Transform} & \textbf{Optical Operation} & \textbf{Support} \\
        \midrule
        Hermite~\cite{jain2025efficient} &  Quadrature basis of $\ket{n}$ & $\R$ \\
        Charlier &  Displacement $D(\alpha)$ & $\Z_{\ge 0}$ \\
        Krawtchouk~\cite{guan2026quantumKrawtchouktransformusing} &  Beamsplitter on $\ket{N,0}$ & $\{0,\ldots,N\}$  \\
        Meixner & Single-mode Squeeze $S(r)$ &  $\Z_{\ge 0}$ \\
        Hahn & Clebsch-Gordan circuits &  $\{0,\ldots,N\}$ \\
        Laguerre & Cartesian to radial mode change &  $[0,\infty)$ \\
        \bottomrule[.5mm]
    \end{tabular}
    \caption{The quantum orthogonal polynomial transforms established by viewing each as matrix elements of a quantum optical gate.}
    \label{tab:transforms}
\end{table}

\noindent The following table summarizes the different quantum orthogonal polynomial transforms along with the algorithmic and computational mechanisms used to obtain them. We also comment on the numerical analysis on discrete grids used for the error analysis.
\begin{enumerate}
    \item \emph{Charlier transform}: The displacement operator
    $D(\alpha)$ with $\mu=|\alpha|^2$ realizes the Charlier transform on
    $\Z_{\ge 0}$, with the intertwining identity $D(\alpha)a^\dagger D(\alpha)^\dagger = a^\dagger-\bar\alpha$ and the orthogonality measure being the Poisson distribution.

    The transform itself can be implemented arithmetically in the grid representation or using the compositional framework for Gaussian gates. 
    
    \item \emph{Krawtchouk transform}: Beamsplitters at fixed total photon number realize the Krawtchouk transform on $\{0,\ldots,N\}$.

    In the grid representation itself, a single beam splitter can be realized as an arithmetic rotation on two registers or by using the compositional framework for Gaussian gates. Different representations of $SU(2)$ allow different implementations of the beam splitter. For example, on Dicke states, beam splitters can be represented as qubit rotations.

    \item \emph{Meixner transform}: Single-mode Squeezing on the $SU(1,1)$ number basis realizes the Meixner transform, whose orthogonality measure is the negative binomial (thermal) distribution of squeezed light.

    The Meixner transform can again be implemented arithmetically using hyperbolic rotations and the compositional framework for Gaussian gates.

    \item \emph{Hahn transform}: An $SU(2)$ Clebsch-Gordan coupling in a fixed spin-$j$ sector realizes the Hahn transform on $\{0,\ldots,N\}$. The implementation given goes through the identification of the Hahn polynomials with the CG coefficients and is similar to the cascade implementation of the Schur transform. More efficient poly-logarithmic implementations are left open. Possible paths to a poly-logarithmic Hahn transform are discussed in the conclusion.

    \item \emph{Laguerre transform}: The Laguerre transform requires an efficient procedure for grid conversion between Cartesian and polar grids and a circular beam splitter. The transform depends on a single integer parameter.

    An implementation of the Laguerre transform on the logarithmic grid is also proposed. The logarithmic grid can be used to implement the transformfor non integer Laguerre parameters.

\end{enumerate}

\textbf{Related work:} We discuss briefly the relation between our work and other work on QOPTs that have recently appeared in the quantum computing literature. The first QOPT in a quantum circuit was the quantum Hermite transform introduced by \cite{jain2025efficient}. The result is based on a digital implementation of the fast-forwarding of the quantum harmonic oscillator, which is derived from Lie algebraic considerations and has been well known in the physics literature under Wei--Norman type transformations \cite{wei1963lie}. However, an exponentially efficient implementation in terms of quantum circuits on a discrete grid requires the careful error analysis carried out in \cite{jain2025efficient}.

The follow-up work \cite{iyer2026efficientquantumcircuitshighdimensional} generalizes the framework of the quantum Hermite transform to fast-forwarding other operators quadratic in position and momentum operators, which are effectively quantum optical gates corresponding to $SU(2)$. In this work, we show that the simulation extends to the
$SU(1,1)$ optical gates, and we sharpen the analysis in parameter regimes that
were treated as constant in
\cite{iyer2026efficientquantumcircuitshighdimensional}. The digital
implementations of all our QOPTs are built on this extended simulation
framework.

A quantum Krawtchouk transform realizing the QOPT definition has recently been given in \cite{guan2026quantumKrawtchouktransformusing}. The Krawtchouk transform has also been implicit in the analysis of the decoded quantum interferometry algorithm \cite{marwaha2026complexitydecodedquantuminterferometry} and in permanent sampling \cite{aaronson2011computational}. These are the places where the beam splitter gates acting on $SU(2)$ representations have been useful in the design of quantum algorithms. Our work provides both a conceptual framework that unifies these seemingly disparate appearances of the Krawtchouk transform and a family of new QOPTs that will potentially have further algorithmic implications.

As we noted previously, classical OPTs are heavily used for the solution of PDEs and in spectral methods more generally. It would be interesting to come up with applications of efficient QOPTs to solving PDEs and quadrature problems. In this direction, we note the work \cite{Childs_2021} that utilizes a quantum Chebyshev transform for the solution of PDEs.
The Chebyshev transform is a special case of the Jacobi transform with parameters $\alpha=\beta=-1/2$. Investigating applications of the QOPTs designed in this paper towards solution of PDEs and quadrature remains an interesting direction for future work.

The optical gates underlying our proposed orthogonal transforms are native operations for continuous-variable (CV) hardware, and our constructions are particularly well matched to the emerging class of hybrid oscillator-qubit quantum processors. The recent review~\cite{liu2026hybrid} discusses in detail full-stack architectures for such
hybrid continuous-variable--discrete-variable (CV-DV) devices, in which bosonic
modes and qubits coexist and are programmed through
instruction set architectures whose primitive gates include precisely the Gaussian bosonic operations that generate the Heisenberg--Weyl, $SU(2)$, and $SU(1,1)$ groups. On such hardware the gates realizing our discrete QOPTs are available natively, and a single displacement, beam splitter, or squeeze gate is a single physical instruction rather than a compiled subroutine. The discrete QOPTs embed directly into the instruction set of a hybrid processor. 
This offers a natural and hardware-efficient route to implementing the Askey scheme transforms, and it complements the algorithmic
viewpoint we adopt here. However, in the present work, we intentionally follow
the digital implementation route by expressing each optical gate on a truncated,
qubit-encoded Fock space and rendering it through the efficient quantum Hermite
transform~\cite{jain2025efficient}, so that our transforms run on standard gate-based quantum systems
without assuming access to native bosonic hardware.

\section{Preliminaries}
We briefly introduce the notation for the harmonic oscillator and the two-mode Schwinger representations of $SU(2)$ and $SU(1,1)$. These are used directly for the Charlier, Meixner and Krawtchouk QOPTs. The Schwinger representations also appear naturally in the Laguerre transform because they are embedded in the 2D isotropic oscillator.

\subsubsection{Quantum harmonic oscillator}

The quantum harmonic oscillator represents a single bosonic mode. Let $\hat{x}$ and $\hat{p}$ denote the position and momentum operators acting on $L^2(\R)$ by
\begin{align}
    (\hat{x}f)(x)=x f(x), \qquad (\hat{p}f)(x)=-i\frac{d}{dx}f(x),
\end{align}
so that $[\hat{x},\hat{p}]=iI$. The one-dimensional quantum harmonic oscillator Hamiltonian is
\begin{align}
    \hat{H}=\frac{1}{2}\left(\hat{x}^2+\hat{p}^2\right).
\end{align}
It is often convenient to use the annihilation and creation operators
\begin{align}
    a=\frac{1}{\sqrt{2}}(\hat{x}+i\hat{p}), \qquad
    a^{\dagger}=\frac{1}{\sqrt{2}}(\hat{x}-i\hat{p}),
\end{align}
which satisfy $[a,a^{\dagger}]=I$. Equivalently,
\begin{align} 
    \hat{x}=\frac{a+a^{\dagger}}{\sqrt{2}}, \qquad
    \hat{p}=\frac{a-a^{\dagger}}{i\sqrt{2}}.
\end{align}
The number operator $\hat{N}=a^{\dagger}a$ has eigenstates $\ket{n}$, called Fock states, with
\begin{align}
    \hat{N}\ket{n}=n\ket{n}, \qquad
    a\ket{n}=\sqrt{n}\ket{n-1}, \qquad
    a^{\dagger}\ket{n}=\sqrt{n+1}\ket{n+1}.
\end{align}
In this basis, $\hat{H}=\hat{N}+1/2$, so the Fock states are precisely the energy eigenstates of the oscillator. In the position representation they are the Hermite functions
\begin{align}
    \psi_n(x)=\frac{1}{\sqrt{2^n n!\sqrt{\pi}}}e^{-x^2/2}H_n(x),
\end{align}
up to a harmless phase convention.

For $m$ modes we write $a_j,a_j^{\dagger}$ and $\hat{x}_j,\hat{p}_j$ for the corresponding operators on the $j$-th tensor factor, with nontrivial commutators $[a_j,a_k^{\dagger}]=\delta_{j k}I$ or, equivalently, $[\hat{x}_j,\hat{p}_k]=i\delta_{j k}I$. The multi-mode Fock basis is
\begin{align}
    \ket{n_1,\ldots,n_m}:=\ket{n_1}\otimes\cdots\otimes\ket{n_m}.
\end{align}
The Gaussian quantum optical gates are generated by low-degree polynomials in the mode operators. Displacements are generated by linear combinations of $a$ and $a^{\dagger}$, beamsplitters by terms such as $a_j a_k^{\dagger}+a_k a_j^{\dagger}$, and squeezers by quadratic terms such as $a_j^2+(a_j^{\dagger})^2$.

\subsubsection{Schwinger $SU(2)$ representations}

We recall the Schwinger realization of $SU(2)$ on two oscillator
modes on sectors with fixed total photon number.  Given two modes with annihilation operators $a_1,a_2$, define
\begin{align}
    J_+ = a_1^{\dagger}a_2, \qquad
    J_- = a_2^{\dagger}a_1, \qquad
    J_z = \frac{1}{2}(a_1^{\dagger}a_1-a_2^{\dagger}a_2).
\end{align}
Equivalently, $J_x=(J_++J_-)/2$ and $J_y=(J_+-J_-)/(2i)$.  These operators
satisfy the commutation relations for $SU(2)$, 
\begin{align}
    [J_z,J_+]=J_+, \qquad [J_z,J_-]=-J_-, \qquad [J_+,J_-]=2J_z.
\end{align}
They commute with the total number operator
\begin{align}
    \hat{N}_{12}=a_1^{\dagger}a_1+a_2^{\dagger}a_2.
\end{align}
The two-mode Fock space decomposes into Schwinger SU(2) representations organized by the total-number, 
\begin{align}
    \mathcal{H}^{(M)}_{12}=\textup{span}\{\ket{x,M-x}:x=0,\ldots,M\}.
\end{align}
On this sector the Schwinger operators give the spin-$j$ irreducible
representation with $j=M/2$.  The dictionary between Fock labels and spin labels
is
\begin{align}
    \ket{j,m}\quad \longleftrightarrow \quad \ket{j+m,\,j-m},
    \qquad j=\frac{M}{2},\quad m=x-\frac{M}{2}.
\end{align}
Thus fixing the total number of two oscillator modes is the same as fixing a
finite-dimensional $SU(2)$ irrep.  Number-preserving optical operations generated
by $J_x,J_y,J_z$ act inside this sector as $SU(2)$ rotations.  For example,
beamsplitters act as rotations in this Schwinger representation.

For later use, if we have two independent two-mode pairs, then fixing their
total photon numbers to $2j_1$ and $2j_2$ realizes a tensor product of spin
representations $V_{j_1}\otimes V_{j_2}$.  The Clebsch--Gordan transform is the
basis change from the product of the two fixed-sector Fock bases to the basis of
total spin.  
this fixed-sector Clebsch--Gordan transform to realize the Hahn transform.

\subsubsection{Schwinger \texorpdfstring{$SU(1,1)$}{SU(1,1)} representations}

There is an analogous oscillator realization of $SU(1,1)$ obtained from
quadratic two-mode operators.  Given two modes, define
\begin{align}
    K_+ = a_1^{\dagger}a_2^{\dagger}, \qquad
    K_- = a_1a_2, \qquad
    K_0 = \frac{1}{2}(a_1^{\dagger}a_1+a_2^{\dagger}a_2+1).
\end{align}
These operators satisfy
\begin{align}
    [K_0,K_+]=K_+, \qquad [K_0,K_-]=-K_-, \qquad [K_+,K_-]=-2K_0.
\end{align}
Unlike the $SU(2)$ generators above, the operators $K_+$ and $K_-$ do not
preserve the total photon number.  Instead they create or annihilate photons in
pairs and therefore commute with the number-difference operator
\begin{align}
    \Delta_{12}=a_1^{\dagger}a_1-a_2^{\dagger}a_2.
\end{align}
Thus the two-mode Fock space decomposes into fixed number-difference sectors.
For $q\geq0$ one such sector is
\begin{align}
    \mathcal{H}^{(q)}_{12}=\textup{span}\{\ket{n+q,n}:n\geq0\},
\end{align}
and for $q<0$ the two modes are exchanged.  On this sector the above generators
realize a positive discrete-series representation of $SU(1,1)$ with Bargmann
index
\begin{align}
    k=\frac{|q|+1}{2}.
\end{align}
Using the basis
\begin{align}
    \ket{n,k}\quad \longleftrightarrow \quad
    \begin{cases}
        \ket{n+q,n}, & q\geq0,\\
        \ket{n,n+|q|}, & q<0,
    \end{cases}
\end{align}
the generator $K_0$ has eigenvalue $n+k$.  In this sense, fixing a photon-number
difference is the noncompact analogue of fixing the total photon number in the
Schwinger $SU(2)$ construction.

The two-mode squeezing gate is the basic optical operation generated by this
algebra.  In the above notation it has the form
\begin{align}
    S_2(z)=\exp(zK_+-\overline{z}K_-),
\end{align}
up to the harmless sign convention for $z$.  Since $S_2(z)$ preserves
$\Delta_{12}$, it acts independently on each fixed-$q$ sector, equivalently on
each $SU(1,1)$ irrep labelled by $k=(|q|+1)/2$.  This is the
$SU(1,1)$ counterpart of the way beamsplitters act as $SU(2)$ rotations on fixed
total-number sectors.

The same positive discrete-series label also appears in the radial oscillator
notation.  A fixed angular momentum sector with angular momentum $m$ has
$k=(|m|+1)/2$, matching the fixed number-difference label above after passing to
circular modes.  Thus the Laguerre and Meixner constructions both use the same
underlying $SU(1,1)$ representation theory, but in different oscillator bases.

\subsection{Quantum optical gates and their digital implementations}

\label{eqn:relevant_gate_set}

We now recall the quantum optical operations used below. These operations are unitary gates generated by linear or quadratic expressions in the oscillator operators. We list each operation first in oscillator notation and then in position notation.

\paragraph{Single-mode operations.}
The single-mode displacement operator is generated by translations in position and momentum. Writing $\alpha=x_0+ip_0$, we have
\begin{align}
    D(\alpha)=\exp(\alpha a^{\dagger}-\overline{\alpha}a), \qquad D(\alpha)=\exp(ip_0\hat{x}-ix_0\hat{p}),
\end{align}
up to the usual global phase convention.

The single-mode squeezing operator is equivalently generated by the dilation operator in position and momentum variables,
\begin{align}
    S(r)=\exp\left(\frac{r}{2}\left(a^2-(a^{\dagger})^2\right)\right), \qquad S(r)=\exp\left(-\frac{ir}{2}(\hat{x}\hat{p}+\hat{p}\hat{x})\right).
\end{align}

The phase rotation of a single mode is
\begin{align}
    R(\theta)=\exp(i\theta a^{\dagger}a), \qquad R(\theta)=\exp\left(\frac{i\theta}{2}(\hat{x}^2+\hat{p}^2)\right),
\end{align}
where the second equality follows from $\hat{H}=a^{\dagger}a+1/2$, up to an irrelevant global phase.

\paragraph{Two-mode operations.}
For real squeezing parameter $r$, the two-mode squeezing operator has the form
\begin{align}
    S_2(r)=\exp\left(r(a_1a_2-a_1^{\dagger}a_2^{\dagger})\right), \qquad S_2(r)=\exp\left(ir(\hat{x}_1\hat{p}_2+\hat{p}_1\hat{x}_2)\right).
\end{align}
The complex-parameter version is obtained by taking the corresponding complex linear combination of the same two-mode quadratic generators.

The beamsplitter operations mix two oscillators while preserving the total oscillator degree. We use two closely related conventions. The real beamsplitter is
\begin{align}
    \textup{BS}_{Y}(\theta)=\exp\left(\theta(a_1a_2^{\dagger}-a_2a_1^{\dagger})\right), \qquad
    \textup{BS}_{Y}(\theta)=\exp\left(i\theta(\hat{x}_1\hat{p}_2-\hat{p}_1\hat{x}_2)\right).
\end{align}
The complex beamsplitter is
\begin{align}
    \textup{BS}_{X}(\theta)=\exp\left(i\theta(a_1a_2^{\dagger}+a_2a_1^{\dagger})\right), \qquad
    \textup{BS}_{X}(\theta)=\exp\left(i\theta(\hat{x}_1\hat{x}_2+\hat{p}_1\hat{p}_2)\right).
\end{align}
\label{subsec:digital_implementation_optics}
\subsection{Digital implementation of optical operations}
The focus will be on the Gaussian optical gates because in the digital model the non-Gaussian operations are very easy to implement, as we use qubits to encode the Fock basis. The Kerr gate $K(\gamma) = e^{i\gamma\hat{n}^2}$ is diagonal in the Fock basis. Thus, using our main result below, Theorem \ref{thm:digital_implementation_gaussian_optical_gates}, we can perform universal optical quantum computation: an optical unitary on $K$ modes that is decomposed into $n$ elements of the Gaussian and Kerr gate set can be simulated to error $\epsilon < 1$ using $\mathcal{O}(\textup{poly}(n, K, \log(1/\epsilon)))$ one- and two-qubit gates and ${\mathcal{O}}(\textup{poly}(\log(K), \log(n)))$ qubits. This is mostly a byproduct of our results, as the full gate set will not be needed for implementing the QOPTs.

We now recall the finite-dimensional oscillator used to implement the gates digitally and explain how it interfaces with the truncated Fock space. For a target cutoff $N$, the logical oscillator register stores the truncated number basis (which we refer to as the \emph{effective qubit space})
\begin{align}
    \mathcal{H}_{N}=\textup{span}\{\ket{n}:0\leq n<N\},
\end{align}
using $\mathcal{O}(\log N)$ qubits. The quadrature evolutions are not applied directly on this register. Instead, we embed the low-energy oscillator into a larger grid space of dimension $L=2^{\ell}$. Let the grid be indexed by $j=-L/2,\ldots,L/2-1$, with
\begin{align}
        x_j=j\sqrt{2\pi/L}.
\end{align}
The $L$-dimensional discretized Hermite state associated with the $n$-th oscillator eigenfunction is
\begin{align}
        \ket{\psi_n^{(L)}}=\left(\frac{2\pi}{L}\right)^{1/4}\sum_{j=-L/2}^{L/2-1}\psi_n(x_j)\ket{j},
\end{align}
where $\psi_n$ is the Hermite function from the oscillator notation section. These states discretize and truncate the continuum oscillator eigenfunctions; for $L$ sufficiently large compared with the cutoff $N$, the first $N$ such states are nearly orthonormal and approximate the low-energy eigenstates of the discrete oscillator. The state $|\psi_n^{(L)}\rangle$ lives in the \emph{operational qubit space} $\mathcal{H}_{L}$, where of course $\mathcal{H}_N \subset \mathcal{H}_{L}$.

The discrete quantum Hermite transform is the unitary basis change
\begin{align}
    \mathsf{QHT}:\ket{n}\mapsto \ket{\psi_n^{(L)}}, \qquad 0\leq n<N,
\end{align}
with inverse $\textup{QHT}^{\dagger}$ mapping discretized oscillator states back to the number basis. We also use the discretized quadrature operators
\begin{align}
        \overline{x}=\sum_{j=-L/2}^{L/2-1}x_j\ket{j}\bra{j}, \qquad \overline{p}=F_L^{\dagger}\overline{x}F_L,
\end{align}
where $F_L$ is the centered finite Fourier transform. For several modes, these definitions are applied tensor factor by tensor factor, giving operators $\overline{x}_j$ and $\overline{p}_j$ on the $j$-th discretized oscillator.

We will use the following efficient implementation of this basis change.
\begin{restatable}[quantum Hermite transform, \cite{jain2025efficient} adapted]{theorem}{quantumHermiteTransform}
\label{thm:qht}
Let $\epsilon > 0$. If $L = \mathcal{O}(N^{9/4}/\epsilon^{13/4})$, then there exists a quantum circuit $\mathsf{QHT}$ of gate complexity $\mathcal{O}((\log N + \log(1/\epsilon))^3 \times \log(1/\epsilon))$ that can perform the following map from $\mathcal{H}_{N} \rightarrow \mathcal{H}_{L}$ with error $\epsilon$:
\begin{equation}
\sum_{m=0}^{N-1} \alpha_m |m\rangle \mapsto \sum_{m=0}^{N-1} \alpha_m |\psi_m\rangle.
\end{equation}
The coefficients $\alpha_m \in \mathbb{C}$ are arbitrary and normalized, i.e., $\sum_m |\alpha_m|^2 = 1$.
\end{restatable}

This gives a direct way to implement Gaussian optical gates on truncated Fock states. Suppose $|\psi\rangle\in\mathcal{H}_N^{\otimes K}$ and $U$ is a single- or two-mode Gaussian gate generated by the quadratures $\hat{x}_j,\hat{p}_j$ and, when two modes are involved, $\hat{x}_k,\hat{p}_k$. In the oscillator basis the same gate is described by the mode operators $a_j,a_j^{\dagger}$ and $a_k,a_k^{\dagger}$, but for digital implementation we move to the grid representation, replace the continuum quadratures by $\overline{x}_j,\overline{p}_j,\overline{x}_k,\overline{p}_k$, apply the corresponding grid evolution, and then return to the number basis:
\begin{align*}
    (\mathsf{QHT}_j^{\dagger} \otimes \mathsf{QHT}_k^{\dagger})U(\overline{x}_j, \overline{p}_j, \overline{x}_k, \overline{p}_k)(\mathsf{QHT}_j \otimes \mathsf{QHT}_k)|\psi\rangle,
\end{align*}
with the evident one-mode version when $U$ acts on only one oscillator. This expression should be read as an approximate implementation: the input is supported on the first $N$ Fock levels, while the intermediate quadrature evolution takes place in the larger $L$-dimensional grid space and may temporarily populate levels above the input cutoff. The corollary below, proved in the detailed digital implementation section, states that choosing $N$ and $L$ as indicated makes both this leakage and the discretization error negligible while keeping the gate complexity polynomial in the number of modes, the number of gates, and the target precision parameter.

Additionally, we decompose each $U(\overline{x}_j, \overline{p}_j, \overline{x}_k, \overline{p}_k)$ into simpler evolutions that are at most quadratic in the quadrature operators (Section \ref{subsec:phasespace}). Each component has a fixed upper bound on its evolution time, so a longer evolution is implemented by repeating evolutions with a smaller angle. Thus, if $\widehat{U}=\widehat {W}_R\cdots \widehat{W}_1$ is a product of $R$ elementary Gaussian factors (Section \ref{eqn:relevant_gate_set}), this decomposition multiplies the number of factors by a quantity that we call the \emph{amplification factor} $M$. The amplification factor accounts for the decomposition into basic Gaussian gates with evolution time bounded by a constant. Under this assumption, $R\times M$ is an effective measure of the number of basic Gaussian gates.

\begin{restatable}
[Digital implementation of Gaussian optical gates]{theorem}{digitalImplTheorem}
\label{thm:digital_implementation_gaussian_optical_gates}
Let
\begin{align}
    \ket{\psi}=\sum_{\mathbf m\in\{0,\ldots,N-1\}^K}a_{\mathbf m}
    \ket{\mathbf m},
    \qquad \sum_{\mathbf m}|a_{\mathbf m}|^2=1,
\end{align}
be an arbitrary state in the effective qubit space $\mathcal{H}_{N}^{\otimes K}$.  Let
$\widehat{U}=\widehat {W}_R\cdots \widehat{W}_1$ be a product of $R$ elementary Gaussian factors (Section \ref{eqn:relevant_gate_set}). Suppose the maximum amplification factor of any component $\widehat{W}_j$ is at most $M$.
For target
error $\epsilon<1$, assume
\begin{align*}
    N=\Omega\!\left(\log(1/\epsilon)\right),
\end{align*}
with a sufficiently large implicit constant, and use an operational grid
space $\mathcal{H}_{L}^{\otimes K}$ with
\begin{align*}
    L=\Theta\!\left(
       K^{13/4}(5^{RM}N)^{9/4}\epsilon^{-13/4}+\log K
    \right).
\end{align*}
Then there is a digital quantum circuit $\overline{U}$ on
$\mathcal{O}(K\log L)$ qubits and
\begin{align*}
\mathcal{O}\!\left(
K\bigl(\log(5^{RM}N)+\log(K/\epsilon)\bigr)^3\log(K/\epsilon)
+RM\log^2L\right)
\end{align*}
one- and two-qubit gates
that satisfies
\begin{align*}
    \lVert (\widehat{U} - \overline{U})P_{N}\rVert < \epsilon,
\end{align*}
where $P_N$ is the orthogonal projection onto
$\mathcal H_N^{\otimes K}$ in the number-basis encoding.

\end{restatable}

This is proven in Section \ref{sec:eff_implement}. Thus one can interpret the above as saying that the digital implementation requires at most a cubic overhead in circuit complexity.  The gates listed in Section~\ref{eqn:relevant_gate_set} generate the Gaussian unitaries on finitely many modes: displacements supply the affine translations, while phase rotations, squeezers, and beamsplitters generate the relevant symplectic transformations.  Adding any non-quadratic gate such as the Kerr gate $K(\gamma)=e^{i\gamma\hat n^2}$ gives a universal continuous-variable gate set in the standard finite-energy approximation sense, and in our digital Fock encoding the Kerr gate is diagonal.  Thus the above implies a digital encoding of a universal optical gate set. Thus we have the following simple corollary that may be of independent interest.

\begin{corollary}[Digital encoding of universal quantum optics]
Suppose we use the digital encoding spaces $\mathcal{H}_{N}, \mathcal{H}_L$ as specified in Theorem \ref{thm:digital_implementation_gaussian_optical_gates}. Let $\widehat{U}$ be an optical unitary on $K$ modes that decomposes into $R$ one- and two-mode Gaussian or Kerr gates. Suppose the amplification factor of the Gaussian components is bounded by $M$. Then there is a digital quantum circuit $\overline{U}$ on $\mathcal{O}(K\log(L))$ qubits and 
\begin{align*}
\mathcal{O}\!\left(
K\bigl(\log(5^{RM}N)+\log(K/\epsilon)\bigr)^3\log(K/\epsilon)
+RM\log^2L\right)
\end{align*}
one- and two-qubit gates
that satisfies
\begin{align*}
    \lVert (\widehat{U} - \overline{U})P_{N}\rVert < \epsilon,
\end{align*}
where $P_N$ is the orthogonal projection onto
$\mathcal H_N^{\otimes K}$ in the number-basis encoding.

\end{corollary}

\section{Discrete orthogonal quantum polynomial transforms} \label{sec:dqopt}
The discrete QOPTs that we obtain in this work are those for the Charlier, Meixner, Krawtchouk and Hahn polynomials. These polynomials are supported on a discrete lattice and are orthogonal with respect to the Poisson, negative binomial, binomial and hypergeometric weights respectively. The Askey scheme provides a unified perspective on these polynomial families by viewing them as solutions to the discrete hypergeometric equation.

Let $\Delta f = f(x+1) - f(x)$ and $\nabla f = f(x) - f(x-1)$ be the forward
and backward finite-difference operators on the integer lattice. The discrete
hypergeometric equation is the second-order difference equation
\all{
H_{\Delta}\,y(x) = \Big[\,P_{2}(x)\,\Delta\nabla + P_{1}(x)\,\Delta + \lambda\,\Big]y(x) = 0,
}{Hdelta}
with the degree constraints $\deg P_{2}(x) \leq 2$ and $\deg P_{1}(x) \leq 1$.
The Charlier, Meixner and Krawtchouk polynomials are special cases of this discrete framework where $\deg P_{2} \leq 1$. The Hahn polynomials are the discrete analogs of the Jacobi polynomials and correspond to $\deg P_{2} = 2$.

The hypergeometric equation transforms to a self-adjoint Sturm--Liouville form when multiplied by a weight $\rho(x)$ satisfying the constraint $\rho(x) P_{2}(x) \Delta \nabla + \rho(x) P_{1}(x) \Delta + \rho(x)\lambda=
\Delta(\rho(x) P_{2}(x) \nabla) + \lambda \rho(x)$. The constraint 
implies that multiplication by $\rho(x)$ reduces the
hypergeometric equation to self-adjoint form,
$$
\Delta\!\left[\,\rho(x)P_{2}(x)\,\nabla y(x)\,\right] + \lambda\,\rho(x)\,y(x) \;=\; 0.
$$
Matching the coefficients of $y(x+1)$ on both sides yields the discrete
Pearson equation,
\all{
\frac{\rho(x+1)}{\rho(x)} = \frac{P_{2}(x)+P_{1}(x)}{P_{2}(x+1)}
} {pearson}
which provides an explicit expression for $\rho(x)$. The $y(x-1)$
coefficients match identically and impose no further constraint.

The hypergeometric equation and the weights for converting the four discrete polynomial families Charlier, Meixner, Krawtchouk and Hahn to the self-adjoint form are given in Table~\ref{tab:3}. The weights correspond to the Poisson, negative binomial, binomial and hypergeometric distributions respectively.

\begin{table}[H]
\renewcommand{\arraystretch}{1.8}
\centering
\resizebox{\textwidth}{!}{%
\begin{tabular}{|c|c|c|c|c|}
\hline
\textbf{Family} & \textbf{Support} & $P_{2}(x)$ & $P_{1}(x)$ & \textbf{Weight} $\rho(x)$ \\
\hline
Charlier $(\mu)$ & $\Z_{\ge0}$ & $x$ & $\mu - x$ & $\dfrac{\mu^{x}}{x!}\,e^{-\mu}$ \\
\hline
Meixner $(\beta,c)$ & $\Z_{\ge0}$ & $x$ & $\beta c-(1-c)x$ & $\dfrac{(\beta)_{x}\,c^{x}}{x!}$ \\
\hline
Krawtchouk $(p,\mathcal{N})$ & $\{0,\dots,\mathcal{N}\}$ & $x$ & $\dfrac{\mathcal{N}p-x}{1-p}$ & $\dbinom{\mathcal{N}}{x}p^{x}(1-p)^{\mathcal{N}-x}$ \\
\hline
Hahn $(\mathcal{N}, \alpha, \beta )$& $\{0,\dots,\mathcal{N}-1\}$ & $x(\mathcal{N}+\alpha-x)$ & $(\beta+1)(\mathcal{N}-1)-(\alpha+\beta+2)x$ & $\dbinom{\alpha+x}{x}\dbinom{\beta+\mathcal{N}-1-x}{\mathcal{N}-1-x}$ \\
\hline
\end{tabular}
}
\caption{Coefficients $P_2,P_1$ and orthogonality weights $\rho$ for the four discrete QOPT families. Krawtchouk is supported on $\mathcal{N}+1$ points and Hahn on $\mathcal{N}$ points; the weights are the Poisson, negative binomial, binomial and hypergeometric distributions respectively.}
\label{tab:3}
\end{table}

Similar to the organization of the table, the discrete QOPT constructions are organized by the degree of the corresponding hypergeometric equation. For the Charlier, Meixner and Krawtchouk polynomials, the QOPT is realized by a generalized displacement operator and is poly-logarithmic in the grid dimension. For the discrete Hahn transform, the QOPT is realized by relating the Hahn polynomials to the $SU(2)$ Clebsch--Gordan coefficients and requires resources linear in the lattice dimension.

\subsection{Unified framework for Charlier, Meixner and Krawtchouk Transforms} \label{pmorov}

The three QOPTs for the Charlier, Meixner and Krawtchouk polynomials are the degree $1$ discrete QOPTs according to Table~\ref{tab:3}. Further, they can be grouped together as they share the single algorithmic mechanism of being realized by a generalized displacement operator for one of the rank-one groups.
For the Heisenberg--Weyl group, the operator is an ordinary displacement, while for $SU(2)$ this is a beam-splitter gate and for $SU(1,1)$ it is the squeeze gate.

The states these operators generate from vacuum are the Perelomov coherent states of that group \cite{perelomov1986}. In the notation of the three dynamical algebras used in this work, the operators read:
\begin{align}
    \label{eq:generalized-displacements}
    \text{HW:}&\quad D(\alpha)=e^{\alpha a^\dagger-\bar\alpha a},
    \qquad \alpha\in\C,
    \\
    SU(2):&\quad R(z)=e^{\arctan|z|\,\left(e^{i\arg z}J_+-e^{-i\arg z}J_-\right)},
    \qquad z\in\C,
    \\
    SU(1,1):&\quad S(z)=e^{\operatorname{arctanh}|z|\,\left(e^{i\arg z}K_+-e^{-i\arg z}K_-\right)},
    \qquad |z|<1 .
\end{align} 
Acting on the respective lowest weight states in the representations $\ket{0}$, $\ket{j,-j}$, and $\ket{0,k}$, these operators produce the coherent states:
\begin{align}
    \ket{\alpha}=D(\alpha)\ket 0,
    \qquad
    \ket{z,j}=R(z)\ket{j,-j},
    \qquad
    \ket{z,k}=S(z)\ket{0,k}
\end{align}
The key structural observation is that measuring a coherent state in the number basis of its representation yields the orthogonality measure of the corresponding polynomial family. Measuring each coherent state in the number basis reproduces the Poisson, binomial or negative binomial distribution respectively.

\begin{table}[h!]
    \centering
    \renewcommand{\arraystretch}{1.8}
    \resizebox{\textwidth}{!}{%
    \begin{tabular}{lcc}
    \toprule[0.5mm]
     & \textbf{Coherent-state expansion} & \textbf{Number basis measure} \\
    \midrule
    HW & $\ket{\alpha}=e^{-\mu/2}\sum_{n\ge0}\dfrac{\alpha^n}{\sqrt{n!}}\ket n$ & $e^{-\mu}\dfrac{\mu^n}{n!}$ (Poisson) \\
    $SU(2)$ & $\ket{z,j}=\dfrac{1}{(1+|z|^2)^{j}}\displaystyle\sum_{m=-j}^{j}\binom{2j}{j+m}^{1/2}\!z^{j+m}\ket{j,m}$ & $\binom{2j}{j+m}p^{j+m}(1-p)^{j-m}$ (binomial) \\
    $SU(1,1)$ & $\ket{z,k}=(1-|z|^2)^{k}\displaystyle\sum_{n\ge0}\binom{n+2k-1}{n}^{1/2}\!z^{n}\ket{n,k}$ & $\binom{n+2k-1}{n}(1-|z|^2)^{2k}|z|^{2n}$ (neg.\ binomial) \\
    \bottomrule[.5mm]
    \end{tabular}
    }
    \caption{Measuring a Perelomov coherent state in the number basis of its representation reproduces the orthogonality measure of the polynomial family. Here the Poisson parameter
    $\mu=|\alpha|^2$ in the first row, the binomial parameter $p=|z|^2/(1+|z|^2)$, and the negative binomial parameter is $|z|^{2}$ for $|z|<1$.}
    \label{tab:coherent-measures}
\end{table}

The coherent state provides the ground state for the orthogonal polynomial family. The remaining question concerns the generation of the excited states corresponding to the degree $l$ polynomials. The main result for the Charlier, Meixner and Krawtchouk transforms is that the full transform is obtained by applying a generalized displacement to an
arbitrary number state, rather than just the vacuum. Concretely, the number basis matrix elements of the displacement gate factor as a coherent-state amplitude times the orthonormal polynomial of the family:
\begin{align}
    \label{eq:transform-principle}
    \bra{n}U(z)\ket{l} =\underbrace{\sqrt{w(n)}}_{\text{measure}}
    \;\times\;
    \underbrace{\widehat P_l(n)}_{\text{orthonormal polynomial}}
    \;\times\;(\text{phase})
\end{align}
where $U(z)$ is the relevant displacement of
\eqref{eq:generalized-displacements} and $\widehat P_l$ is the degree-$l$ orthonormal polynomial for the weight $w$. Equivalently, the displacement gate is the polynomial transform, wherein it maps the number basis coherently to the basis of orthonormal polynomial states $\ket{\bf l}$ of
\defref{def:qopt}, up to the diagonal coherent-state envelope and a phase. In the next three subsections, we obtain the corresponding QOPTs according to this template, establish correctness and carry out the error analysis for the grid-based implementations.

\subsubsection{Charlier transform}

The Charlier transform is the simplest instance of the coherent state based QOPT framework as it is generated by the single-mode displacement operator $D(\alpha) = \exp(\alpha a^\dagger-\bar\alpha a)$, which translates the vacuum to the coherent state $\ket{\alpha}$. Measuring $\ket{\alpha}$ in the Fock basis returns the photon number $n$, sampled according to the Poisson distribution $w_\mu(n)=e^{-\mu}\mu^n/n!$, where $\mu=|\alpha|^2$ is the mean photon number, as recorded in the Heisenberg--Weyl row of \tabref{tab:coherent-measures}. The
Charlier polynomials $C_k(n;\mu)$ are the orthogonal polynomials with respect to the Poisson distribution on the discrete integer lattice.

The Charlier transform can be viewed as a prototype for the other discrete QOPT proofs. First, we collect the properties of the Charlier polynomials that are shared with other members of the Askey scheme and then we move to the specific properties of the Charlier polynomials that result in the QOPT being realized by the displacement operator. 

The main underlying claim for the efficient quantum Charlier transform is that the displacement operator not only maps the Fock ground state $\ket{0}$ to the Poisson ground state $\sum_{n} \sqrt{ w_{\mu}(n)} \ket{n}$, but also that $D(\alpha)$ maps the higher number states $ \ket{k}$ coherently to the Charlier polynomial states of degree $k$. The displacement operator implements the Charlier transform in the following sense: applying $D(\sqrt\mu)$ to the number state $\ket{k}$ produces the degree-$k$ Charlier polynomial state. 

In common with other Askey scheme polynomials, the Charlier polynomials can be defined as a generating function as well as a hypergeometric series. 

\begin{definition}
For $\mu>0$ and $k,n\in\Z_{\geq 0}$, the Charlier polynomial $C_k(n;\mu)$ is
\begin{align}
    C_k(n;\mu)=\sum_{s=0}^{\min(k,n)}(-1)^s\frac{k(k-1)\cdots(k-s+1)\,n(n-1)\cdots(n-s+1)}{s!\,\mu^s}= {}_2F_0\!\left(\begin{matrix}-k,\,-n\\ -\end{matrix};\,-\frac1\mu\right),
\end{align}
with the $s=0$ term equal to $1$.
These polynomials are characterized by the generating function
\all{
\sum_{k\ge0}C_k(n;\mu)\frac{t^k}{k!}
=
e^t\left(1-\frac{t}{\mu}\right)^n.
} {gfc}
Here $k$ is the degree of the polynomial.
\end{definition}
Further, the Charlier polynomials are the eigenfunctions of a second order birth-death operator on a discrete lattice. This property is not required for the analysis of the Charlier transform but is recorded here as it is the property 
common to the discrete polynomial transforms in the Askey scheme. 

\begin{proposition}[Charlier birth--death operator] \label{thm:bd}
Let $w_\mu(n)=e^{-\mu}\frac{\mu^n}{n!}$ for $n\in\Z_{\ge 0}$ and define the
birth--death operator
\all{
(\mathcal{L}f)(n)
=
\mu\big(f(n+1)-f(n)\big)
+
n\big(f(n-1)-f(n)\big).
} {Ldef}
Then $\mathcal{L}$ is symmetric on $\ell^2(\mathbb Z_{\ge0},w_\mu)$ and has the
Charlier polynomials as eigenfunctions,
$$
\mathcal{L}\,C_k(\cdot;\mu)
=
-k\,C_k(\cdot;\mu).
$$
\end{proposition}
\noindent The birth--death operator above is the discrete hypergeometric operator $H_\Delta$ of
\eqref{Hdelta} specialized to the Charlier row of Table~\ref{tab:3}
($P_2(x)=x$, $P_1(x)=\mu-x$), with eigenvalue $\lambda_k=-k$. Its symmetry on
$\ell^2(\mathbb Z_{\ge0},w_\mu)$ is exactly the self-adjointness guaranteed by
the Pearson weight $\rho=w_\mu$ obtained from \eqref{pearson}. 

We now turn to the specific properties of the Charlier polynomials. These  polynomials also admit a ladder-type description: each $C_k$ is generated from the constant function by $k$-fold application of a single backward-shift operator.

\begin{restatable}{proposition}{charlierShift} \label{thm:cop}

Let $T_-$ denote the backward shift operator acting on functions $f:\mathbb Z_{\ge0}\to\mathbb C$ by
\[
(T_-f)(n)=f(n-1),
\]
with $f(-1):=0$.
The Charlier polynomials with parameter $\mu>0$ defined in equation \eqref{gfc} have the expression, 

\[
C_k(n;\mu)
=
\left(1-\frac n\mu T_-\right)^k1,
\]
where $1(n)=1$ denotes the constant function on $\mathbb Z_{\ge0}$.

\end{restatable}

\begin{proof}

Let $A:=1-\frac n\mu T_-$ so that the claim is equivalent to the statement $C_k(n;\mu)=A^k1$. Packaging the claim for all $k$ into a single exponential generating function, 
\all{
e^{tA}1 = \sum_{k\ge0}\frac{t^k}{k!}A^k1 = \sum_{k\ge0}C_k(n;\mu)\frac{t^k}{k!}
}{opgen1}
It therefore suffices to show that $e^{tA}1$ agrees with the generating function for the Charlier polynomials. 
Since $A=I-\frac n\mu T_-$, the exponential $e^{tA}
=e^t e^{-(t/\mu)nT_-}$, evaluating the powers of the $nT_-$ operator on the constant function $ 
(nT_-)^r1
=
n(n-1)\cdots(n-r+1)
=
n^{\underline r}$ where $n^{\underline r}$ denotes the falling factorial. The exponential $e^{tA}$ can therefore be expanded as, 
\[
e^{-(t/\mu)nT_-}1
=
\sum_{r\ge0}
\frac{(-t/\mu)^r}{r!}
n^{\underline r}.
\]
Using the binomial expansion $(1-x)^n
=
\sum_{r\ge0}
\binom{n}{r}(-x)^r
=
\sum_{r\ge0}
\frac{n^{\underline r}}{r!}(-x)^r$ with $x=t/\mu$, we conclude that
\[
e^{-(t/\mu)nT_-}1
=
\left(1-\frac t\mu\right)^n.
\]
Hence $e^{tA}1
=
e^t\left(1-\frac t\mu\right)^n$, establishing the equality in equation \eqref{opgen1}.

\end{proof}
Propositions~\ref{thm:bd} and~\ref{thm:cop} give the two important classical properties of the Charlier family. It is the ladder formulation in Proposition~\ref{thm:cop} that is used to prove the main result on the correctness of the quantum Charlier transform. The proof shows that, when intertwined by the displacement operator, the creation operator on Fock space transforms into the ladder operator for the Charlier polynomials on the discrete lattice of Fock-space eigenstates.

\begin{restatable}{theorem}{charlierDisp}\label{thm:charlier-disp}
Let $D(\alpha)=\exp(\alpha a^\dagger-\bar\alpha a)$ and $\mu=|\alpha|^2$. Then $\langle n|D(\alpha)|k\rangle=e^{-\mu/2} \frac{\alpha^{n}} {\sqrt{n!}}\, \frac{(-\bar{\alpha})^k}{\sqrt{k!}}\,C_k(n;\mu)$, where $C_k(n;\mu)$ is the $k$-th Charlier polynomial.
\end{restatable}
\begin{proof}
The proof is sketched here with the full proof included in the Appendix~\ref{App:B} for completeness.

Write $\ket k=(a^\dagger)^k\ket 0/\sqrt{k!}$ and use the intertwining identity
$D(\alpha)a^\dagger D(\alpha)^\dagger=a^\dagger-\bar\alpha$ to reduce the kernel to
$(a^\dagger-\bar \alpha)$ acting on the coherent state $D(\alpha)\ket0$. On
discrete Fock space the creation operator acts as $\sqrt n\,T_-$; a diagonal
similarity $W_\alpha$ then conjugates $a^\dagger-\bar\alpha$ into
$\bar\alpha\big(\tfrac n\mu T_- - I\big)$, which is exactly the Charlier ladder of
Proposition~\ref{thm:cop}. Applying that proposition gives
$\langle n|D(\alpha)|k\rangle
= e^{-\mu/2}\tfrac{\alpha^{n}}{\sqrt{n!}}\tfrac{(-\bar\alpha)^{k}}{\sqrt{k!}}
C_k(n;\mu)$.
\end{proof}

\noindent We next describe the quantum Charlier transform and its algorithm explicitly, then carry out the error analysis for its grid-based implementation.
For $\mu>0$, set
$\alpha=\sqrt{\mu}$ and let
\begin{align}
    \Pi_\pi=e^{i\pi a^\dagger a},\qquad
    U_{\mathrm{Ch}}(\mu)=D(\sqrt{\mu})\Pi_\pi .
\end{align}
Then
\begin{align}
    \langle n|U_{\mathrm{Ch}}(\mu)|k\rangle
    =e^{-\mu/2}\frac{\mu^{(n+k)/2}}{\sqrt{n!k!}}C_k(n;\mu)
    =:\varphi_k(n;\mu).
\end{align}
With $\alpha=\sqrt\mu$ real we have $(-\bar\alpha)^k=(-1)^k\mu^{k/2}$, and the
factor $\Pi_\pi:\ket k\mapsto(-1)^k\ket k$ cancels the $(-1)^k$, leaving the real-valued kernel
$\varphi_k(n;\mu)=\sqrt{w_\mu(n)}\,\widehat C_k(n;\mu)$, where
$\widehat C_k=\mu^{k/2}C_k/\sqrt{k!}$ is the Charlier polynomial normalized to be orthogonal with respect to the Poisson weight.
This is exactly the Charlier instance of the template
\eqref{eq:transform-principle}. The exact Charlier state of degree $k$ is
\begin{align}
    |\mathsf C_k(\mu)\rangle
    :=U_{\mathrm{Ch}}(\mu)|k\rangle
    =\sum_{n\geq 0}\varphi_k(n;\mu)|n\rangle .
\end{align}

\begin{algorithm}[H]
\caption{Quantum Charlier Transform}
\label{alg:charlier_transform}
\textbf{Input:} A Charlier parameter $\mu>0$ and amplitudes $\{a_k\}_{k=0}^{N-1}$ encoding Charlier coefficients in the state
$\ket{\phi}_{\mathrm{Ch}}=\sum_{k=0}^{N-1}a_k\ket{k}$.

\textbf{Output:} The Charlier polynomial state
$\sum_{k=0}^{N-1}a_k\ket{\mathsf C_k(\mu)}=U_{\mathrm{Ch}}(\mu)\ket{\phi}_{\mathrm{Ch}}$, up to the chosen digital approximation error.

\begin{enumerate}
    \item Apply $\Pi_\pi=e^{i\pi a^\dagger a}$, mapping
    $\sum_{k=0}^{N-1}a_k\ket{k}$ to
    $\sum_{k=0}^{N-1}(-1)^k a_k\ket{k}$.
    \item Apply $\mathsf{QHT}$ to move from the number basis to the discretized oscillator grid basis.
    \item Implement $D(\sqrt{\mu})=\exp(-i\sqrt{\mu}\,\hat p)$ on the grid register, splitting the evolution into bounded-time pieces as in the proof of the Charlier transform theorem.
    \item Apply $\mathsf{QHT}^{\dagger}$ on the enlarged cutoff guaranteed by Theorem~\ref{thm:digital_implementation_gaussian_optical_gates}.
\end{enumerate}
\textbf{Return:} $\sum_{k=0}^{N-1}a_k\ket{\mathsf C_k(\mu)}$.
\end{algorithm}

\begin{theorem}[Quantum Charlier transform]
Fix $\mu>0$, target error $\epsilon<1$. Let
\begin{align}
    \ket{\psi}=\sum_{m\in\{0,\ldots,N-1\}}a_{m}
    \ket{m},
    \qquad \sum_{m}|a_{m}|^2=1,
\end{align}
be an arbitrary state in the effective qubit space $\mathcal{H}_{N}$.  Assume $N=\Omega(\log(1/\epsilon))$, set
$r_\mu=\lceil4\sqrt{2}\,e\sqrt\mu\rceil$, and take
$L=\Theta\!\left((5^{r_\mu}N)^{9/4}\epsilon^{-13/4}\right)$. Then there is a digital quantum circuit $\overline{U}$ on $\mathcal{O}(\log(L))$ qubits with gate complexity
\begin{align*}
    \mathcal O\!\left(\bigl(\log(5^{r_\mu}N)+\log(1/\epsilon)\bigr)^3
        \log(1/\epsilon)+r_\mu\log^2L
    \right)
\end{align*}
that satisfies
\begin{align*}
    \lVert (U_{\mathrm{Ch}}(\mu) - \overline{U})P_{N}\rVert < \epsilon,
\end{align*}
where $P_N$ is the orthogonal projection of  $\mathcal{H}_L$ onto  $\mathcal{H}_N$. 
 \end{theorem} 
 \begin{proof} 
 The phase-space form of the displacement operator, for real
$\alpha=\sqrt{\mu}$ we have
\begin{align}
    D(\sqrt{\mu})=\exp(-i\sqrt{\mu}\,\hat p).
\end{align}

The factor $\Pi_\pi=e^{i\pi a^\dagger a}$ is diagonal in the Fock basis (and hence in the qubit encoding) and can
be applied exactly as the phase $|k\rangle\mapsto (-1)^k|k\rangle$. It does not
increase the Fock cutoff and does not contribute to the QHT discretization
error. Thus the only Gaussian evolution that has to be implemented through the
digital optical-gate theorem is the displacement. Splitting it into
$r_\mu=\lceil \sqrt{\mu}/\delta_*\rceil$ pieces gives
\begin{align}
    D(\sqrt{\mu})
    =\prod_{b=1}^{r_\mu}\exp\left(-i\frac{\sqrt{\mu}}{r_\mu}\hat p\right),
    \qquad
    \left|\frac{\sqrt{\mu}}{r_\mu}\right|\leq \frac{1}{4\sqrt{2}\,e}.
\end{align}
We can choose
\begin{align}
    r_\mu:=\left\lceil \frac{\sqrt{\mu}}{\delta_*}\right\rceil
    =\left\lceil 4\sqrt{2}\,e\sqrt{\mu}\right\rceil = \mathcal{O}(\sqrt{\mu}) .
\end{align}
Thus the above, combined with the decomposition of the displacement gate from Section~\ref{subsec:phasespace}, implies that the amplification factor is $\mathcal{O}(\sqrt{\mu})$. Applying Theorem~\ref{thm:digital_implementation_gaussian_optical_gates} implies the result.
\end{proof}

\subsubsection{Meixner transform}
The Meixner transform is the $SU(1,1)$ instance of the coherent state based QOPT framework. We will be working on the two mode Schwinger representation with 
$(K_{+}, K_{-}, K_{0}) = (a_{1}^{\dagger} a_{2}^{\dagger}, a_{1} a_{2}, \frac{1}{2}( n_{1} + n_{2} + 1))$ acting on fixed number difference sector as introduced in the preliminaries. The displacement 
operator $S(z)$ of equation \eqref{eq:generalized-displacements} acting on the two mode vacuum $\ket{0,0}$ generates the two mode squeezed vacuum state corresponding to the $SU(1,1)$ row in Table \ref{tab:coherent-measures}. The squeeze operator $S(z)$ preserves the photon number difference $q$ and hence
acts separately within each fixed-$q$ sector; we write $S(z,k)$ for its
restriction to the sector of number difference $q$, where $k=(|q|+1)/2$ is the
corresponding Bargmann index, so that $S(z)=\bigoplus_{q\in\Z}S(z,k(q))$.

Measuring the squeezed vacuum states samples from the negative binomial distribution $\text{Neg-Bin}(|q|+1, |z|^{2})$. As $|q|+1 \in \Z_{+}$, the Meixner transform is realized for the integer $\beta$ family of Meixner polynomials. We define below the Meixner polynomial family in terms of their hypergeometric representation and generating function, 

\begin{definition} \label{def:meixner}
For $\beta>0$, $0<c<1$, and $k,n\in\Z_{\geq 0}$, the Meixner polynomial $M_k(n;\beta,c)$ is
\begin{align}
    M_k(n;\beta,c)=\sum_{s=0}^{\min(k,n)}\frac{k(k-1)\cdots(k-s+1)n(n-1)\cdots(n-s+1)}{\beta(\beta+1)\cdots(\beta+s-1)s!}\left(1-\frac{1}{c}\right)^s
    ={}_2F_1\!\left(-k,-n;\,\beta;\,1-\frac{1}{c}\right),
\end{align}
with the $s=0$ term equal to $1$. These polynomials are characterized by the generating function
$$
\sum_{k\ge0}\frac{(\beta)_k}{k!}M_k(n;\beta,c)t^k
=
(1-t)^{-\beta-n}
\left(1-\frac{t}{c}\right)^n.
$$
\end{definition}
The quantum Meixner transform rests on the same principle as the Charlier
transform, with one structural difference. For Charlier, the adjoint action
$D(\alpha)a^\dagger D(\alpha)^\dagger=a^\dagger-\bar\alpha$ is affine, which is
what reduces the kernel to the first-order ladder of
Proposition~\ref{thm:cop}. The corresponding $SU(1,1)$ action
$S(z,k)K_+S(z,k)^\dagger$ mixes all three generators $K_+,K_-,K_0$, so no such
reduction is available and the transition kernel
$\bra{n,k}S(z,k)\ket{W,k}$ needs to be computed directly in terms of Meixner polynomials. The computation uses the
normal-ordered form of the squeeze operator,
\begin{align}
    \label{eq:disentangle}
    S(z,k)=e^{zK_+}\,(1-|z|^2)^{K_0}\,e^{-\bar z K_-},
\end{align}
which sends $\ket{W,k}$ down to an intermediate level $\ket{m,k}$ and back up to
$\ket{n,k}$. The kernel is therefore a single sum over $m$, running to
$\min(W,n)$ since the lowering factor cannot pass the ground state and the
raising factor must supply the remaining $n-m$ quanta. The square roots from
$K_-\ket{m,k}=\sqrt{m(m+2k-1)}\ket{m-1,k}$ and its adjoint assemble into the
binomial prefactors and the ratio $(-W)_m(-n)_m/(2k)_m\,m!$, yielding the
${}_2F_1$ in Definition~\ref{def:meixner} of the Meixner polynomials with $\beta=2k$ and $c=|z|^2$. 

\begin{restatable}{theorem}{meixnerSqueeze}
Let $ |z| < 1$, then
the matrix elements of $S(z,k)$ in the $SU(1,1)$ number basis are given by
$$
\bra{M,k}S(z,k)\ket{W,k} = (-1)^W z^{*W} z^M (1-|z|^2)^k
\sqrt{\binom{W+2k-1}{W}\binom{M+2k-1}{M}}\, M_W(M;\,2k,\,|z|^2),
$$
where $M_W(M; \beta, c) = {}_2F_1(-W,-M;\,\beta;\,1-1/c)$ is the $W$-th
Meixner polynomial with $\beta=2k$ and $c=|z|^2$, orthogonal with respect to
the negative binomial distribution $\text{Neg-Bin}(2k, |z|^{2})$.
\end{restatable}
The proof of the Theorem is sketched above and given in Appendix \ref{App:D}. 
We now move over to the implementation of the Meixner transform. 
Fix a number-difference, equivalently angular-momentum, sector $q\in\mathbb Z$
and write
\begin{align}
    \kappa=\frac{|q|+1}{2}.
\end{align}
Let $\{|n,\kappa\rangle:n\geq0\}$ denote the corresponding $SU(1,1)$/radial oscillator number
basis (Section \ref{sec:radial_osc_not}).
For $0<|z|<1$, let
\begin{align}
    \omega_\ell(n;\kappa,z)
    :=(-1)^\ell z^{*\ell} z^n(1-|z|^2)^\kappa
    \sqrt{\binom{\ell+2\kappa-1}{\ell}\binom{n+2\kappa-1}{n}}
    M_\ell(n;2\kappa,|z|^2).
\end{align}
The exact Meixner state of degree $\ell$ in the sector $q$ is
\begin{align}
    |\mathsf M_\ell^{(q)}(z)\rangle
    :=S(z,\kappa)|\ell,\kappa\rangle
    =\sum_{n\geq0}\omega_\ell(n;\kappa,z)|n,\kappa\rangle.
\end{align}

The following optical basis change converts between the Cartesian Fock basis and
the circular, hence radial, basis.

\begin{restatable}[Cartesian to Radial Mode Basis Change]{Lemma}{cardToRad} \label{CBS}
Let \begin{align}
\mathsf{U}_{\textup{C2R}}:=R_1(-\pi/4)R_2(-\pi/4)\textup{BS}_{Y}(-\pi/4)\textup{BS}_{X}(\pi/4),
\end{align}
where $R_j$ denotes a phase/rotation gate on mode $j$. Then
\begin{align}
\mathsf{U}_{\textup{C2R}}\ket{n_1,n_2}=\ket{n_1,n_2}_{+}=\ket{\min(n_1,n_2),\,n_1-n_2}.
\end{align}
\end{restatable}
\noindent A proof is included in  Appendix \ref{App:A} for completeness.

Let $A_q$ be the unitary mapping $|n\rangle|0\rangle \mapsto |n\rangle|n + q\rangle$. We define the Meixner transform to be
\begin{align*}
    U_{\textup{Mx}} = A_q^{\dagger}\mathsf{U}_{\textup{C2R}}^{\dagger}S(z,k)\mathsf{U}_{\textup{C2R}}A_q.
\end{align*}

\begin{algorithm}[H]
\caption{Quantum Meixner Transform}
\label{alg:meixner_transform}
\textbf{Input:} Meixner parameters $(\beta,c)=(2\kappa,|z|^2)$ with
$\kappa=(|q|+1)/2$ and $0<|z|<1$, together with amplitudes $\{a_\ell\}_{\ell=0}^{N-1}$ in the state
$\ket{\phi}_{\mathrm{Mx}}=\sum_{\ell=0}^{N-1}a_\ell\ket{\ell}$.

\textbf{Output:} The Meixner polynomial state
$\sum_{\ell=0}^{N-1}a_\ell U_{\textup{Mx}}\ket{\ell}$, up to the chosen digital approximation error.

\begin{enumerate}
    \item With $q_+=\max(q,0)$ and $q_- = \max(-q,0)$, apply the reversible arithmetic map $A_q:\sum_{\ell=0}^{N-1}a_\ell\ket{\ell}\ket{0}\mapsto\sum_{\ell=0}^{N-1}a_\ell\ket{\ell+q_+}\ket{\ell+q_-}$.
    \item Apply $\mathsf{U}_{\textup{C2R}}$ from Lemma~\ref{CBS} to express the state in the fixed $SU(1,1)$ sector labelled by $\kappa=(|q|+1)/2$.
    \item Apply the sector-wise squeeze $S(z,\kappa)$ by decomposing the two-mode squeezing gate into bounded Gaussian pieces and applying Theorem~\ref{thm:digital_implementation_gaussian_optical_gates}.
    \item Apply $\mathsf{U}_{\textup{C2R}}^{\dagger}$ and then $A_q^{\dagger}$ by reversing the arithmetic embedding from the first step.
\end{enumerate}
\textbf{Return:} $\sum_{\ell=0}^{N-1}a_\ell U_{\textup{Mx}}\ket{\ell}$.
\end{algorithm}

\begin{theorem}[Quantum Meixner transform]
\label{thm:meixner_transform}
Fix $q\in\mathbb Z$, $0<|z|<1$ and target error
$\epsilon<1$. Set
$\rho_z:=\operatorname{arctanh}|z|$.

Let
\begin{align}
    \ket{\psi}=\sum_{m\in\{0,\ldots,N-1\}}a_{m}
    \ket{m},
    \qquad \sum_{m}|a_{m}|^2=1,
\end{align}
be an arbitrary state in the effective qubit space $\mathcal{H}_{N}$.
Set $N_q:=N+|q|$ and assume
$N_q=\Omega(\log(1/\epsilon))$. Let $r_{\mathrm{Mx}}$ be the total number of
bounded-time elementary Gaussian factors in the two Cartesian-to-radial basis
changes and the squeeze implementation; then
$r_{\mathrm{Mx}}=\mathcal O(1+\rho_z)$. Choose
\begin{align*}
    {
    L=\Theta\!\left((5^{r_{\mathrm{Mx}}}N_q)^{9/4}
    \epsilon^{-13/4}\right).}
\end{align*}
Then there is a digital quantum circuit $\overline{U}$ on $\mathcal{O}(\log(L))$ qubits with gate complexity
\begin{align*}
    {
    \mathcal O\!\left(
    \begin{aligned}
        &\bigl(\log(5^{r_{\mathrm{Mx}}}N_q)+\log(2/\epsilon)\bigr)^3
        \log(2/\epsilon)+r_{\mathrm{Mx}}\log^2L
        +\operatorname{polylog}(N_q)
    \end{aligned}
    \right)}
\end{align*}
that satisfies
\begin{align*}
    \lVert (U_{\textup{Mx}} - \overline{U})P_{N}\rVert < \epsilon,
\end{align*}
 where $P_N$ is the orthogonal projection of  $\mathcal{H}_L$ onto  $\mathcal{H}_N$.
\end{theorem}

\begin{proof}
Write $z=e^{i\phi}\tanh\rho_z$.  The sector operator $S(z,\kappa)$ is the
restriction of the two-mode squeezing gate $S_2(z)$ to the fixed
number-difference sector $q$. Let
\begin{align}
    \delta_*:=\frac{1}{4\sqrt{2}e},
    \qquad
    r_z:=\left\lceil \frac{\rho_z}{\delta_*}\right\rceil.
\end{align}
We can perform $A_{q}$ efficiently with quantum arithmetic. Splitting the
squeezing evolution into $r_z=\lceil \rho_z/\delta_*\rceil$ equal pieces
gives steps of squeezing time at most $\delta_*$. Applying the fixed
phase-space decomposition of each small two-mode squeezing step
{(Section~\ref{subsec:phasespace}) gives $\mathcal O(r_z)$ bounded factors.}
{The two fixed basis changes and the phase rotations needed for complex
$z$ contribute only a constant number more, so
$r_{\mathrm{Mx}}=\mathcal O(1+r_z)=\mathcal O(1+\rho_z)$. The sector input is
contained in the two-mode cutoff $N_q=N+|q|$, since $|n,\kappa\rangle$ is
represented by $|n+q_+,n+q_-\rangle$, where $q_+=\max(q,0)$ and
$q_-=\max(-q,0)$. Applying
Theorem~\ref{thm:digital_implementation_gaussian_optical_gates} with $K=2$,
effective cutoff $N_q$, and total bounded-factor count $r_{\mathrm{Mx}}$ gives
the stated grid size and Gaussian gate count; the final
$\operatorname{polylog}(N_q)$ term accounts for $A_q$ and $A_q^\dagger$.}
\end{proof}

\subsubsection{Krawtchouk transform}

The Krawtchouk polynomials are orthogonal polynomials with respect to the
binomial distribution. They arise naturally in the matrix elements of the
beam-splitter (BS) operator, which couples two bosonic modes while preserving
the total photon number. Acting on $\ket{\mathcal{N},0}$ it produces the $SU(2)$
coherent state of the $\mathcal{N}$-photon representation, whose photon number in
the first mode is distributed as $\text{Bin}(\mathcal{N},p)$ with
$p=\sin^2\theta$, corresponding to the $SU(2)$ row in
Table~\ref{tab:coherent-measures}. We define the Krawtchouk polynomials next in
terms of their generating function and hypergeometric representation.

\begin{definition} \label{def:krawtchouk}
The Krawtchouk polynomials $K_k(n;p,\mathcal{N})$ are polynomials of degree $k$ in
the discrete variable $n \in \{0, 1, \ldots, \mathcal{N}\}$ defined by the
generating function
\all{
\sum_{k=0}^{\mathcal{N}} \binom{\mathcal{N}}{k} K_k(n;p,\mathcal{N})\, t^k
= (1+t)^{\mathcal{N}-n}\left(1 - \frac{q}{p}t\right)^n,
}{gfk}
where $0 < p < 1, q=1-p$, and
$K_k(n;p,\mathcal{N}) = {}_2F_1(-k, -n; -\mathcal{N}; 1/p)$.
\end{definition}
\noindent The quantum Krawtchouk transform rests on the same principle as the  Meixner transform, with the same structural difference from the Charlier case. The adjoint action $B(\theta)J_+B(\theta)^\dagger$ mixes all three generators $J_+,J_-,J_0$, so the
kernel $\bra{n,\mathcal{N}-n}B(\theta)\ket{k,\mathcal{N}-k}$ is computed
directly. The normal-ordered form of the beam splitter sends
$\ket{k,\mathcal{N}-k}$ down to an intermediate level and back up to
$\ket{n,\mathcal{N}-n}$, so the kernel is a single sum terminating at
$\min(k,n)$, and the ladder normalizations assemble it into the ${}_2F_1$ of
Definition~\ref{def:krawtchouk}. The $SU(2)$ representation is realized as
the two-mode Schwinger representation, as for the Meixner transform. Here $p=\sin^2\theta$ corresponds to $p=|z|^2/(1+|z|^2)$ for the stereographic
parameter $z=\tan\theta$ of the $SU(2)$ case in
\eqref{eq:generalized-displacements}.

\begin{restatable}{theorem} {kravchuckBs}
Let $\mathcal{N}\in\Z_{\geq 0}$ and let
$B(\theta) = \exp[\theta(a_1^\dagger a_2 - a_1 a_2^\dagger)]$ be the beam-splitter
operator, $p = \sin^2\theta$, and $q = \cos^2\theta$. Then for
$n, k \in \{0, 1, \ldots, \mathcal{N}\}$,
\begin{align}
    \langle n, \mathcal{N}-n | B(\theta) | k, \mathcal{N}-k \rangle
    = (-1)^k \sqrt{\binom{\mathcal{N}}{n}\binom{\mathcal{N}}{k}}\,
    p^{(n+k)/2}\, q^{(\mathcal{N}-n-k)/2}\, K_k(n;p,\mathcal{N}),
\end{align}
where $K_k(n;p,\mathcal{N}) = {}_2F_1(-k, -n; -\mathcal{N}; 1/p)$ is the $k$-th
Krawtchouk polynomial.
\end{restatable}
\noindent The proof of the above Theorem as sketched above is included in Appendix \ref{App:C} for completeness. 

Specializing to the 2-mode Schwinger representation with fixed total photon number $\mathcal{N}$, let
\begin{align}
    A_{\mathcal{N}}\ket{k}=\ket{k,\mathcal{N}-k}, \qquad k=0,\ldots,\mathcal{N},
\end{align}
be the isometric embedding of the degree register into the two-mode fixed-number
sector. For $0<\theta<\pi/2$ set $p=\sin^2\theta$ and define the Krawtchouk
sector transform
\begin{align}
    U_{\mathrm{Kv}}(p,\mathcal{N}):=A_{\mathcal{N}}^{\dagger}B(\theta)A_{\mathcal{N}}.
\end{align}
Then
\begin{align}
    U_{\mathrm{Kv}}(p,\mathcal{N})\ket{k}
    =\sum_{n=0}^{\mathcal{N}}(-1)^k\sqrt{\binom{\mathcal{N}}{n}\binom{\mathcal{N}}{k}}
    p^{(n+k)/2}(1-p)^{(\mathcal{N}-n-k)/2}K_k(n;p,\mathcal{N})\ket{n}.
\end{align}
The diagonal sign $(-1)^k$ can be removed before or after the transform by a
number-basis phase gate, depending on the chosen Krawtchouk phase convention.

\begin{algorithm}[H]
\caption{Quantum Krawtchouk Transform}
\label{alg:Krawtchouk_transform}
\textbf{Input:} Krawtchouk parameters $(p,\mathcal{N})$ with $0<p<1$, equivalently
a total photon number $\mathcal{N}$ and beam-splitter angle $p=\sin^2\theta$, and
amplitudes $\{a_k\}_{k=0}^{\mathcal{N}}$ in the state
$\ket{\phi}_{\mathrm{Kv}}=\sum_{k=0}^{\mathcal{N}}a_k\ket{k}$.

\textbf{Output:} The Krawtchouk polynomial state
$\sum_{k=0}^{\mathcal{N}}a_k\, U_{\mathrm{Kv}}(p,\mathcal{N})\ket{k}$, up to the
chosen digital approximation error.

\begin{enumerate}
    \item \textbf{Fixed-Number Sector Embedding.} Reversibly compute the second-mode occupation and apply $A_{\mathcal{N}}:\sum_{k=0}^{\mathcal{N}}a_k\ket{k}\ket{0}\mapsto\sum_{k=0}^{\mathcal{N}}a_k\ket{k,\mathcal{N}-k}$.
    \item \textbf{Beam-Splitter Evolution.} Apply $B(\theta)$ on the two oscillator modes using the phase-space decomposition of Section~\ref{subsec:phasespace} and Theorem~\ref{thm:digital_implementation_gaussian_optical_gates}.
    \item \textbf{Sector Unembedding.} Since the beam splitter preserves total photon number, apply $A_{\mathcal{N}}^{\dagger}$ by reversing the arithmetic used in the first step.
    \item \textbf{Phase Gauge.} If the convention without the factor $(-1)^k$ is desired, apply the corresponding diagonal number-basis sign gate before or after the beam splitter.
\end{enumerate}
\textbf{Return:} $\sum_{k=0}^{\mathcal{N}}a_k\, U_{\mathrm{Kv}}(p,\mathcal{N})\ket{k}$.
\end{algorithm}

\begin{theorem}[Quantum Krawtchouk transform]
\label{thm:Krawtchouk_transform}
Fix $0<p<1$, total photon number $\mathcal{N}$, and target error $\epsilon<1$. Let
\begin{align}
    \ket{\psi}=\sum_{m=0}^{\mathcal{N}}a_m\ket{m},
    \qquad \sum_{m=0}^{\mathcal{N}}|a_m|^2=1,
\end{align}
be an arbitrary state in the fixed-number input register. 

Choose a digital
oscillator cutoff
\begin{align*}
    {
    N=\Theta\!\left(\mathcal{N}+\log(1/\epsilon)\right),
    \qquad N\geq\mathcal N+1,}
\end{align*}
and let $r_{\mathrm{Kv}}$ denote the total number of bounded-time
elementary Gaussian factors in the three-shear beam-splitter implementation.
As shown below, $r_{\mathrm{Kv}}=\mathcal O(1)$ uniformly for
$0<p<1$. Choose operational grid dimension
\begin{align*}
    {
    L=\Theta\!\left((5^{r_{\mathrm{Kv}}}N)^{9/4}\epsilon^{-13/4}\right)
    =\Theta\!\left(N^{9/4}\epsilon^{-13/4}\right).}
\end{align*}
Then there is a digital quantum circuit $\overline{U}$ on
$\mathcal{O}(\log L)$ qubits with gate complexity
\begin{align*}
    {
    \mathcal O\!\left(
    \begin{aligned}
        &\bigl(\log(5^{r_{\mathrm{Kv}}}N)+\log(2/\epsilon)\bigr)^3
        \log(2/\epsilon)
        &\quad{}+r_{\mathrm{Kv}}\log^2L
        +\operatorname{polylog}(\mathcal N+1)
    \end{aligned}
    \right)}
\end{align*}
{such that}
\begin{align*}
    \lVert (U_{\mathrm{Kv}}(p,\mathcal{N})-\overline{U})P_{\mathcal{N}+1}\rVert<\epsilon,
\end{align*}
where $P_{\mathcal{N}+1}$ is the orthogonal projection onto the span of
$\{\ket{0},\ldots,\ket{\mathcal{N}}\}$.
\end{theorem}

\begin{proof}
The map $A_{\mathcal{N}}$ is implemented reversibly by copying the input value $m$ into a
second register and computing $\mathcal{N}-m$ by quantum arithmetic. Its inverse is the
same arithmetic run backwards after the beam splitter, since $B(\theta)$ preserves
the total photon number and hence maps the sector
$\textup{span}\{\ket{k,\mathcal{N}-k}:0\leq k\leq \mathcal{N}\}$ to itself.

After applying $A_{\mathcal{N}}$, every basis state in the sector has both mode occupations
at most $\mathcal{N}$, and hence lies in $\mathcal{H}_{N}^{\otimes 2}$ for the above choice
of $N$. It remains to implement the beam splitter on these two oscillator modes.
By the phase-space decomposition in Section~\ref{subsec:phasespace},
\begin{align}
    B(\theta)
    =e^{-i\,t(\theta)\,\hat{x}_2\hat{p}_1}
    e^{i\,s(\theta)\,\hat{x}_1\hat{p}_2}
    e^{-i\,t(\theta)\,\hat{x}_2\hat{p}_1},
\end{align}
where $t(\theta)=\tan(\theta/2)$ and $s(\theta)=\sin\theta$. For fixed
$0<\theta<\pi/2$, both $|t(\theta)|$ and $|s(\theta)|$ are less than $1$.
With $\delta_*=1/(4\sqrt{2}e)$, splitting each shear into pieces whose
coefficients have magnitude at most $\delta_*$ gives
$r_{\mathrm{Kv}}\leq3\lceil1/\delta_*\rceil=\mathcal O(1)$. Applying
Theorem~\ref{thm:digital_implementation_gaussian_optical_gates} with $K=2$,
effective Fock cutoff $N$, and total bounded-factor count
$r_{\mathrm{Kv}}$ gives the factor $5^{r_{\mathrm{Kv}}}$ in $L$; it is absorbed
into the second $\Theta$ expression because $r_{\mathrm{Kv}}$ is uniformly
constant.
Composing this approximation with the exact arithmetic embeddings $A_{\mathcal{N}}$ and
$A_{\mathcal{N}}^\dagger$ gives the claimed implementation of $U_{\mathrm{Kv}}(p,\mathcal{N})$.
\end{proof}

\subsection{Hahn Transform} \label{subsec:hahn_transform}
We recall from the definition of the hypergeometric equation \eqref{Hdelta} that the Hahn polynomials on the discrete lattice $\{0,\ldots,N\}$ form a two-parameter family $(\alpha,\beta)$ of polynomial solutions to the equation, 
\begin{align}
x(N+ \beta + 1 -x) \Delta \nabla Q_{\ell}(x) + [N(\alpha+1) - (\alpha + \beta +2)x] \Delta Q_{\ell}(x) + \ell(\ell + \alpha + \beta + 1)Q_{\ell}(x) =0. 
\end{align}
The Hahn transform differs from the discrete transforms considered so far because its difference equation has the quadratic leading coefficient $P_{2}(x)=x(N+\beta+1-x)$, rather than the linear coefficient $P_{2}(x)=x$. Consequently, the Hahn transform is not realized by a single optical gate acting on one $SU(2)$ or $SU(1,1)$ representation. Instead, it arises from the Clebsch--Gordan decomposition of two $SU(2)$ representations.

We define the Hahn polynomials explicitly as a hypergeometric series and specify the Hahn transform as an $(N+1)$-dimensional unitary matrix. 
\begin{definition}[Hahn polynomials]
\label{def:hahn_poly}
Let $\alpha,\beta>-1$ and $N\in\Z_{\geq0}$.  For
$l,x\in\{0,\ldots,N\}$ define
\begin{align}
    Q_l(x;\alpha,\beta,N)
    &:={}_3F_2\!\left(\begin{matrix}-l,\ l+\alpha+\beta+1,\ -x\\
    \alpha+1,\ -N\end{matrix};1\right),                                      \\
    \rho(x)&:=\binom{\alpha+x}{x}\binom{\beta+N-x}{N-x},
    \label{eq:hahn_weight}                                                     \\
    h_l&:=\sum_{x=0}^{N}\rho(x)Q_l(x;\alpha,\beta,N)^2,\qquad
    \widehat Q_l(x):=h_l^{-1/2}\sqrt{\rho(x)}Q_l(x;\alpha,\beta,N).
    \label{eq:hahn_orthonormal}
\end{align}
The Hahn transform matrix $U_{\mathrm H}$ is the unitary matrix with entries
$(U_{\mathrm H})_{x l}=\widehat Q_l(x)$.  
\end{definition}

The inverse quantum Hahn transform (in our convention) is thus the operation
\begin{align*}
     U_{\mathrm H}^{\dagger}:
    \sum_{x=0}^{N}f(x)\ket{x}
    &\longmapsto
    \sum_{l=0}^{N}\left(\sum_{x=0}^{N}\widehat Q_l(x)f(x)\right)\ket{l},
\end{align*}
for $f : \{0, \dots, N\} \rightarrow \mathbb{R}$. The Hahn transform matrix is not symmetric; unlike the Krawtchouk transform, it does not square to the identity. In the standard convention,
\begin{align*}
Q_l(x;\alpha,\beta,N)=R_x(\lambda(l);\alpha,\beta,N),
\qquad \lambda(l)=l(l+\alpha+\beta+1).
\end{align*}
The normalized transposed matrix corresponds to the dual Hahn polynomials. The same circuit therefore implements the normalized dual Hahn transform when the roles of $x$ and $l$ are exchanged.

The Hahn polynomials are closely connected to the coupling of  $\textup{SU}(2)$ representations, which we leverage to implement the transform quantumly.
Let $V_j$ denote the spin-$j$ irreducible representation of $SU(2)$, with its standard basis $|j, m\rangle$ labeled by the magnetic quantum numbers $m=-j,\ldots,j$.

The tensor product of two spin representations has a multiplicity free decomposition into a direct sum of other spin representations, 
\begin{align*}
    V_{j_1} \otimes V_{j_2} \cong \bigoplus_{k=\lvert j_1 - j_2\rvert}^{j_1+j_2} V_{k}.
\end{align*}
The Clebsch--Gordan unitary expresses each uncoupled basis state as
\begin{align*}
    U^{j_1, j_2}_{\textup{CG}}|j_1, m_1\rangle|j_2, m_2\rangle = \sum_{k=\lvert j_1-j_2 \rvert}^{j_1+j_2}C_{j_1, m_1; j_2, m_2}^{k, m_1 + m_2}|k, m_1 + m_2\rangle,
\end{align*}
where the linear combination contains at most one weight vector from each $V_k$. The coefficients $C_{j_1, m_1; j_2, m_2}^{k, m_1 + m_2}$ are the CG coefficients. This is a square unitary basis change because
\begin{align*}
    (2j_1+1)(2j_2+1)=\sum_{k=\lvert j_1-j_2\rvert}^{j_1+j_2}(2k+1).
\end{align*}
The CG transform preserves the total magnetic number $M=m_1+m_2$, so it decomposes into square unitary fixed-$M$ blocks, that are described below.

Define $m_{-}=\max(-j_2,M-j_1)$, $m_+=\min(j_2,M+j_1)$, and $N_M=m_+-m_-$. Let $m_1=M-m_--x$ and $m_2=m_-+x$. We consider the subspace of $V_{j_1}\otimes V_{j_2}$ with fixed value $m_1+m_2=M$:
\begin{align*}
\mathcal{V}_{M}^{j_1, j_2} :=\textup{span}_{\mathbb{C}}\left\{ |j_1, M - m_{-} - x \rangle|j_2, m_{-} + x\rangle ~|~  x \in \{0, \dots, N_M\} \right\},
\end{align*}
which has dimension $N_M+1$. Its coupled basis consists of the vectors $\ket{k,M}$ for the admissible total spins $k\geq |M|$, and therefore has the same dimension. The restriction of the CG transform is an $(N_M+1)\times(N_M+1)$ orthogonal matrix. The Hahn polynomials arise from this block for a particular choice of $M,j_1,j_2$.

Concretely, at the parameters
\begin{align}
\label{eq:hahn-cg-parameters}
    j_1=\tfrac{N+\alpha+\beta}{2},\qquad
    j_2=\tfrac{N}{2},\qquad
    M=\tfrac{\beta-\alpha}{2}
\end{align}
one has $m_\pm=\pm\tfrac N2$ and $J_{\min}=\tfrac{\alpha+\beta}{2}$, and the
Hahn transform is the restriction of the CG transform to
$\mathcal V_M^{j_1,j_2}$, up to a diagonal sign and the relabelling
$\ell\mapsto k=J_{\min}+\ell$ of the total-spin register.

\begin{restatable}[]{Lemma}{lemHahnCG}
\label{lem:equiv_hahn_cg}
Let $\alpha,\beta,N\in\Z_{\geq0}$ and let $j_1,j_2,M$ be as in
\eqref{eq:hahn-cg-parameters}.  Then $m_-=-\tfrac N2$, $m_+=\tfrac N2$,
$J_{\min}=\lvert j_1-j_2\rvert=\tfrac{\alpha+\beta}{2}$, and for all
$x,\ell\in\{0,\ldots,N\}$
\begin{align}
\label{eqn:hahn_cg_equivalent}
    \widehat Q_\ell(x;\alpha,\beta,N)
    =(-1)^{x}\,C_{j_1,\,M-m_--x;\;j_2,\,m_-+x}^{\,J_{\min}+\ell,\;M}.
\end{align}
Equivalently, with $D:=\sum_{x=0}^{N}(-1)^x\ket{x}\!\bra{x}$,
\begin{align}
\label{eq:hahn-cg-operator}
    U_{\mathrm H}^{\dagger}
    &=\left(U^{j_1,j_2}_{\textup{CG}}\big|_{\mathcal V_M^{j_1,j_2}}\right)D,\\
    U_{\mathrm H}
    &=D\left(U^{j_1,j_2}_{\textup{CG}}\big|_{\mathcal V_M^{j_1,j_2}}\right)^{\dagger}.\notag
\end{align}
where the total-spin label $k=J_{\min}+\ell$ is identified with the Hahn
degree $\ell$.
\end{restatable}
\begin{proof}
See Appendix~\ref{app:hahn-cg}.
\end{proof}

For the digital algorithm, spin and magnetic labels are stored in binary registers. Only the second factor $V_{j_2}$ is expanded into qubits: it is realized as the symmetric subspace $\operatorname{Sym}^{2j_2}(\mathbb C^2)$ of $2j_2=N$ spin-$\tfrac12$ constituents. Under this realization, every weight vector $\ket{j_2,m_2}$ is represented as a Dicke state with $j_2+m_2$ spin-up constituents, 
\al{ 
\ket{j_{2}, m_{2}} = \ket{ D_{2j_{2}}^{j_{2} + m_{2}}} :=  
\binom{2j_{2}}{ j_{2}+ m_{2} } ^{-1/2} \sum_{z \in \{0, 1\}^{2j_{2}}, |z|= j_{2} + m_{2}} \ket{z}. 
} 
The first factor $V_{j_1}$ remains in the compact binary spin-label representation.

We now explain how to leverage Lemma \ref{lem:equiv_hahn_cg} to perform $U^{\dagger}_H$. First we apply an encoding unitary to 
\begin{align*}
   |f\rangle = \sum_{x=0}^{N}f(x)|x\rangle
\end{align*} defined by 
\begin{align*}
    E_{M}|x\rangle = |j_1, M - m_{-} - x\rangle |j_2, x + m_{-}\rangle.
\end{align*}
Here $E_M$ is shorthand for $E_{j_1,j_2,M}$, and $m_-$ is computed from these three parameters. For the canonical embedding, all four quantities are fixed functions of $(\alpha,\beta,N)$. The map is implemented by reversible affine arithmetic on doubled spin labels, so all stored values are integers.

Additionally, we need a simple phase operation $G_{\mathsf{X}}$ to correct  $(-1)^{x}$  phase factor between the CG coefficients and the Hahn polynomials. Specifically,
\begin{align*}
    G_{\mathsf{X}}|j_1, M- m_{-} -x\rangle|j_2, x+m_{-}\rangle = (-1)^{x}|j_1, M- m_{-} -x\rangle|j_2, x+m_{-}\rangle.
\end{align*}

The inverse quantum Hahn transform is the following sequence:
\begin{align*}
    U_{\textup{CG}}^{j_1,j_2}G_{\mathsf X}E_M|f\rangle
    &=\sum_{l=0}^{N}\left[\sum_{x=0}^{N}f(x)(-1)^x C_{j_1,M-m_--x;\,j_2,m_-+x}^{J_{\min}+l,M}\right]\ket{J_{\min}+l,M}\\
    &=\sum_{l=0}^{N}\left[\sum_{x=0}^{N}f(x)\widehat Q_l(x;\alpha,\beta,N)\right]\ket{J_{\min}+l,M}.
\end{align*}
After the reversible relabeling $\ket{J_{\min}+l,M}\mapsto\ket l\ket M$, the first register is $U_{\mathrm H}^{\dagger}\ket f$. Thus it remains to implement $U_{\textup{CG}}^{j_1,j_2}$.

A naive synthesis of $U_{\textup{CG}}^{j_1, j_2}$ into Givens rotations decomposes the matrix in $\mathcal{O}((N+1)^2)$ basic operations.
Instead, we use $\mathcal{O}(N)$ elementary coupling and recoupling operations. The symmetric realization of $V_{j_2}$ uses $2j_2=N$ constituent qubits, while $V_{j_1}$ and all running labels use binary registers. The total width is therefore $\mathcal{O}(N)$ plus the polylogarithmic arithmetic workspace, rather than $\mathcal{O}(j_1+j_2)$.  Our procedure relies on the following CG primitive, which couples a spin-$\frac12$ irrep to a spin-$j$ irrep:
\begin{align*}
    V_{j} \otimes V_{\frac{1}{2}} \cong V_{j - \frac{1}{2}} \oplus V_{j +\frac{1}{2}}.
\end{align*}
Mathematically, the elementary basis change is, 
\begin{align}
\label{eq:spin-half-cg-up}
    |j, m\rangle|\frac{1}{2}, \frac{1}{2}\rangle
    &=
    \sqrt{\frac{j+m+1}{2j+1}}
    \ket{j + \frac{1}{2}, m+ \frac{1}{2}}
    -
    \sqrt{\frac{j-m}{2j+1}}
    \ket{j-\frac{1}{2},m+\frac{1}{2}}, \\
\label{eq:spin-half-cg-down}
    \ket{j,m}\ket{\frac{1}{2},-\frac{1}{2}}
    &=
    \sqrt{\frac{j-m+1}{2j+1}}
    \ket{j+\frac{1}{2},m-\frac{1}{2}}
    +
    \sqrt{\frac{j+m}{2j+1}}
    \ket{j-\frac{1}{2},m-\frac{1}{2}}.
\end{align}
Operationally, the second register is a qubit from the Dicke state, whereas $j,m\in\tfrac12\mathbb Z$ are spin labels. The circuit stores $2j$ and $2m$ as integers, further $j_{2}+m_{2} \in \Z$ for the Dicke state encoding to be well defined. The circuit computes the signed rotation amplitudes in \eqref{eq:spin-half-cg-up}--\eqref{eq:spin-half-cg-down} reversibly from the binary labels, applies the corresponding controlled rotation to one branch qubit, updates the doubled spin labels, and uncomputes the angle register.

We denote this operation by $\mathsf{CG}_{1/2}$. Our goal is to couple $V_{j_1}$ to $V_{j_2}$ rather than to a single spin-$\tfrac12$ irrep. We use the canonical realization $V_{j_2}\cong\operatorname{Sym}^{2j_2}(\mathbb C^2)$ and expand $V_{j_2}$ into $2j_2$ spin-$\tfrac12$ constituents by repeatedly applying the inverse maximal-spin CG map. 

For a current symmetric block of spin $j$, the inverse maximal-spin CG map is
\begin{align*}
    |j,m \rangle  = \sqrt{\frac{j+m}{2j}}|j - \frac{1}{2}, m -\frac{1}{2}\rangle|\frac{1}{2}, \frac{1}{2}\rangle + \sqrt{\frac{j-m}{2j}}|j - \frac{1}{2}, m + \frac{1}{2}\rangle|\frac{1}{2}, -\frac{1}{2}\rangle.
\end{align*}
Starting from the compactly encoded state $\ket{j_2,m_2}$, repeated application of this peeling map removes one spin-$\tfrac12$ constituent at a time. After $2j_2$ applications, the output is the Dicke state, 
\begin{align*}
    \ket{D^{2j_2}_{\,j_2+m_2}}
    :=\binom{2j_2}{j_2+m_2}^{-1/2}
    \sum_{\substack{\sigma_1,\ldots,\sigma_{2j_2}=\pm\frac12\\
    \sum_t\sigma_t=m_2}}
    \ket{\sigma_1,\ldots,\sigma_{2j_2}}.
\end{align*}
The binomial coefficient counts the strings with exactly $j_2+m_2$ spin-up constituents and therefore gives unit normalization.

For an input $\ket{j_1,m_1}\ket{j_2,m_2}$, we first expand the second factor and then couple its constituents successively to the first factor. Put $T=2j_2$, $K_0=j_1$, and $M_0=m_1$. At layer $t$, let $\sigma_t\in\{\pm\tfrac12\}$ be the magnetic label of the constituent qubit and let $\rho_t\in\{\pm\tfrac12\}$ record the selected total-spin branch. Define
\begin{align*}
    K_t=K_{t-1}+\rho_t,
    \qquad
    M_t=M_{t-1}+\sigma_t.
\end{align*}
The local operation has the form
\begin{align*}
    \ket{K_{t-1},M_{t-1}}\ket{\tfrac12,\sigma_t}
    \longmapsto
    \sum_{\rho_t=\pm\frac12}
    c_{\rho_t}^{\sigma_t}(K_{t-1},M_{t-1})
    \ket{K_t,M_t}\ket{\rho_t}_{\mathcal P_t},
\end{align*}
where $c_{\rho_t}^{\sigma_t}(K_{t-1},M_{t-1})$ is the corresponding coefficient in \eqref{eq:spin-half-cg-up} or \eqref{eq:spin-half-cg-down}, with inadmissible branches assigned coefficient zero. 

The symbols $\sigma_t$ and $\rho_t$ have different roles: $\sigma_t$ changes the magnetic label $M_t$, while $\rho_t$ changes the total-spin label $K_t$. The branch register $\mathcal P_1\cdots\mathcal P_T$ retains the coupling history required for the cascade to be unitary. These histories distinguish the isomorphic copies of an irrep that occur when many spin-$\tfrac12$ representations are tensored together. More explicitly,
\begin{align*}
    (V_{1/2})^{\otimes N} \cong \bigoplus_{s} \mathbb{C}^{m_s} \otimes V_{s},
\end{align*}
where $m_s$ is the multiplicity of $V_s$ and the path labels provide a basis of $\mathbb C^{m_s}$.

The original decomposition of $V_{j_1}\otimes V_{j_2}$ is multiplicity free. Multiplicities appear only after embedding the second factor into $V_{1/2}^{\otimes 2j_2}$, whose irreducible decomposition is not multiplicity free. For each admissible $k$ with a nonzero CG coefficient, the state after the sequential CG cascade can be written as
\begin{align*}
    \sum_{k=J_{\min}}^{j_1+j_2} C_{j_1, m_1; j_2, m_2}^{k, m_1 + m_2}|k, m_1 + m_2\rangle|\eta_k\rangle.
\end{align*}
Here $|\eta_k\rangle$ is the normalized vector in the sequential coupling-path space obtained by summing the histories $(\rho_1,\ldots,\rho_T)$ that end at $K_T=k$. The original Dicke state is invariant under permutations of the $T=2j_2$ constituent qubits. In the sequential coupling basis, this symmetry is represented by the distinguished path-space vector $|\eta_k\rangle$. The vector $|\eta_k\rangle$ is a quantum state, but it is generally a superposition rather than one path basis vector.  Consequently, the total-spin output remains entangled with the path register, which cannot yet be discarded. Racah recoupling changes to a basis in which the symmetric internal path of $V_{j_2}$ is deterministic.

The sequentially coupled space decomposes as
\begin{align*}
    V_{j_1}\otimes V_{1/2}^{\otimes T}
    \cong\bigoplus_k V_k\otimes\mathcal M_k,
\end{align*}
where $\mathcal M_k$ is the path multiplicity space. Since recoupling commutes with the global $SU(2)$ action, it has the block form 
\begin{align*}
    R=\bigoplus_k(I_{V_k}\otimes F_k),
\end{align*}
where $F_k$ acts only on $\mathcal M_k$.

The recoupling map sends $|\eta_k\rangle$ to the deterministic maximal-spin path
\begin{align*}
    |\mathbf S_{\max}\rangle
    :=|S_2=1,S_3=\tfrac32,\ldots,S_T=\tfrac T2=j_2\rangle.
\end{align*}

Because this path does not depend on $k$, it disentangles from the total-spin register and can be uncomputed. The full recoupling is a sequence of local changes of parenthesization, each implemented by one controlled rotation and reversible updates of binary labels.

To describe these changes, write $b_t=|\tfrac12,\sigma_t\rangle$. The register state
\begin{align*}
    |M_T,K_T,\ldots,K_1,j_1\rangle
\end{align*}
records the final magnetic label and every intermediate total spin of the left-associated coupling tree. The same basis vector is written in parenthesized notation as, 
\begin{align*}
    |(\cdots ((j_1b_1)_{K_1}b_2)_{K_2}\cdots b_T)_{K_T},M_T\rangle,
    \qquad K_T\in\{J_{\min},\ldots,j_1+j_2\}.
\end{align*}

Set $S_1=\tfrac12$ and define $B_1=b_1$ and $B_t=(B_{t-1}b_t)_{S_t}$. A local Racah move compares the two orthonormal bases obtained by coupling three representations in the orders $(V_{j_1}\otimes V_{S_{t-1}})\otimes V_{1/2}$ and $V_{j_1}\otimes(V_{S_{t-1}}\otimes V_{1/2})$. Their overlap is the Racah coefficient
\begin{align*}
    \ket{((j_1B_{t-1})_{K_{t-1}}b_t)_{K_t}}
    =\sum_{S_t=S_{t-1}\pm\frac12}
    F_{K_{t-1},S_t}^{j_1,S_{t-1},\frac12,K_t}
    \ket{(j_1(B_{t-1}b_t)_{S_t})_{K_t}}.
\end{align*}

At step $t$, this operation replaces only $K_{t-1}$ by $S_t$ and leaves all labels outside the displayed three-factor subtree unchanged. 

Repeating the move for $t=2,\ldots,T$ replaces $K_1,\ldots,K_{T-1}$ by $S_2,\ldots,S_T$ without changing the final label $K_T=k$. The sequence is therefore block diagonal in $k$. 

The final state is
\begin{align*}
\sum_{k=J_{\min}}^{j_1+j_2} C_{j_1,m_1;j_2,m_2}^{k,m_1+m_2}|k,m_1+m_2\rangle\otimes|\mathbf S_{\max}\rangle.
\end{align*}
The fixed path is then uncomputed. The explicit $2\times2$ Racah matrix, including all four entries and their signs, is derived in Appendix~\ref{sec:racah_computation}. 

The circuit computes each rotation angle from $(j_1,S_{t-1},K_t)$, applies the controlled one-qubit rotation, updates the path label reversibly, and uncomputes the arithmetic workspace.

\paragraph{Register encoding and precision.}
The circuit stores every doubled spin and magnetic label in binary using $\mathcal O(\log(N+\alpha+\beta+1))$ qubits per live label. The constituent register contains $T=N$ qubits, and the path register contains $\mathcal O(N)$ branch qubits. A reusable fixed-point work register computes the square roots and rotation angles. Since the algorithm applies $\mathcal O(N)$ rotations, each rotation is synthesized to operator-norm error $\delta=\Theta(\epsilon/(N+1))$. The work register is uncomputed after every layer.

We now present the quantum Hahn transform $U_{\mathrm H}$ by reversing the preceding CG and Racah operations.

\begin{algorithm}[H]
\caption{Quantum Hahn Transform}
\label{alg:hahn_cg_qubit_registers}
\textbf{Input:} Hahn parameters $\alpha,\beta,N\in\Z_{\geq0}$, a target error
$0<\epsilon<1$, and a degree-register state
$\sum_{l=0}^{N}a_l\ket{l}_{\mathcal L}$.

\textbf{Output:} The coherent Hahn polynomial state
$\sum_{l=0}^{N}a_l\ket{\widehat Q_l}$, where
$\ket{\widehat Q_l}=\sum_{x=0}^{N}\widehat Q_l(x;\alpha,\beta,N)\ket{x}$, up to
error $\epsilon$. Equivalently, the output is $U_{\mathrm H}\sum_{l=0}^{N}a_l\ket l$.

\begin{enumerate}
\item \textbf{Setup.} Compute
\begin{align*}
    j_1&=\frac{N+\alpha+\beta}{2}, & j_2&=\frac N2, \\
    M&=\frac{\beta-\alpha}{2}, & J_{\min}&=\frac{\alpha+\beta}{2}, \\
    m_-&=-j_2, & T&=2j_2=N.
\end{align*}
Reversibly relabel
$l$ as $J=J_{\min}+l$ into a register $\mathcal{K}$, initialize $M$ into a register $\mathcal{M}$, and clear register
$\mathcal{L}$ containing the $\ell$ labels. Initialize register $\mathcal P$ to the basis state
$\mathbf S_{\max}
    =\left(S_2=1,S_3=\frac32,\ldots,S_T=\frac T2=j_2\right)$,
with an empty path when $T<2$.  

\item \textbf{Inverse Racah sequence.} The previous state corresponds to the ordering
\begin{align*}
    \left(j_1(((b_1b_2)_{S_2}b_3)_{S_3}\cdots b_T)_{S_T}\right)_{J},
\end{align*}
apply the inverse of the Racah rotations in the order
$t=T,T-1,\ldots,2$.  

This restores the sequential path
$K_1,\ldots,K_{T-1}$ and maps to the ordering
$(((j_1b_1)_{K_1}b_2)_{K_2}\cdots b_T)_{J}$.  There are $T-1$ moves for
$T\geq1$.

\item \textbf{Undo the sequential coupling.} For
$t=T,T-1,\ldots,1$, apply $\mathsf{CG}_{1/2}^{\dagger}$.

\item \textbf{Compress the symmetric constituent register.} Apply the inverse of
the peeling sequence for
$s=\frac12,1,\ldots,j_2$, to implement $\ket{D^T_{\,j_2+m_2}}\longmapsto\ket{j_2,m_2}.$ Here $m_2=m_-+x$.  Uncompute the constituent and work registers used by the
fusion steps; the path register was already consumed by the inverse CG cascade.

\item \textbf{Postprocessing.} Apply the
inverse encoding $E_M^{\dagger}: \ket{j_1,M-m_--x}\ket{j_2,m_-+x}
    \longmapsto \ket{x}_{\mathcal X}$,
and apply $G_{\mathsf X}\ket{x}=(-1)^x\ket{x}$.  
\end{enumerate}
\textbf{Return:}
$\sum_{l=0}^{N}a_l\sum_{x=0}^{N}\widehat Q_l(x;\alpha,\beta,N)\ket{x}_{\mathcal X}$.
\end{algorithm}

\begin{theorem}[Quantum Hahn Transform]
\label{thm:quantum_hahn}
Let $\alpha,\beta,N\in\Z_{\geq0}$ and $0<\epsilon<1$.  In the reversible
arithmetic and controlled-rotation model described above, Algorithm
\ref{alg:hahn_cg_qubit_registers} implements, to operator-norm error at most
$\epsilon$, the quantum Hahn transform
\begin{align}
    U_{\mathrm H}:\quad
    \sum_{l=0}^{N}a_l\ket{l}
    \longmapsto
    \sum_{x=0}^{N}
    \left(
        \sum_{l=0}^{N}
        \widehat Q_l(x;\alpha,\beta,N)a_l
    \right)\ket{x},
\end{align}
with all ancillary registers returned to their initial states.  This uses
\begin{align}
        \mathcal{O}\!\left(
            N\operatorname{polylog}\!\left(\frac{N+\alpha+\beta+1}{\epsilon}\right)
        \right)
        &\quad\text{one-qubit and two-qubit gates},                              \\
        \mathcal{O}\!\left(
            N+\operatorname{polylog}\!\left(\frac{N+\alpha+\beta+1}{\epsilon}\right)
        \right)
        &\quad\text{qubits}.
\end{align}

\end{theorem}

\begin{proof}
For the  labels in Algorithm~\ref{alg:hahn_cg_qubit_registers},
\begin{align*}
    j_1=\frac{N+\alpha+\beta}{2},\qquad
    j_2=\frac N2,\qquad
    M=\frac{\beta-\alpha}{2},\qquad
    J_{\min}=\frac{\alpha+\beta}{2}.
\end{align*}
They satisfy
\begin{align*}
    m_-=-j_2,\qquad m_+=j_2,\qquad
    T=2j_2=N.
\end{align*}
After the setup step is done, an input
$\sum_l a_l\ket{l}$ has become
\begin{align*}
    \sum_{l=0}^{N}a_l
    \ket{J_{\min}+l,M}\ket{\mathbf S_{\max}}.
\end{align*}
For $K_T = J$, the forward Racah sequence maps
multiplicity vector $\ket{\eta_J}_{\mathrm{seq}}$ to $|\mathbf{S}_{\max}\rangle$.
Consequently, the inverse sweep maps
\begin{align*}
    \ket{J,M}\ket{\mathbf S_{\max}}
    \longmapsto
    \ket{J,M}\ket{\eta_J}_{\mathrm{seq}}.
\end{align*}
The inverse spin-$\frac12$ CG cascade then reverses the sequential coupling of $\left((V_{j_1}\otimes V_{1/2})\otimes \cdots V_{1/2}\right)$. Compressing the symmetric constituent register reconstructs the compact state $\ket{j_2,m_2}$. The resulting state is
\begin{align*}
    \sum_{x=0}^{N}
    \left(
        \sum_{l=0}^{N}a_l
        C_{j_1,M-m_--x;\,j_2,m_-+x}^{J_{\min}+l,M}
    \right)
    \ket{j_1,M-m_--x}\ket{j_2,m_-+x}.
\end{align*}
Applying $E_M^{\dagger}$ replaces the two uncoupled labels by $\ket{x}$.
The phase $(-1)^x$, and
Lemma~\ref{lem:equiv_hahn_cg} gives
\begin{align*}
    (-1)^x
    C_{j_1,M-m_--x;\,j_2,m_-+x}^{J_{\min}+l,M}
    =\widehat Q_l(x;\alpha,\beta,N).
\end{align*}
The returned state is therefore exactly $U_{\mathrm H}\sum_l a_l\ket l$. Every auxiliary label is either consumed by an
inverse basis change or reversibly uncomputed, so all ancillary registers are
clean at the end. 

It remains to bound the resources needed for an approximation. Step 4 uses $T$ inverse maximal-spin peeling rotations, Step 3 uses $T$ elementary inverse CG rotations, and inverse recoupling uses $T-1$ Racah rotations. Path preparation,
the sign gauge, and reversible label updates require $\mathcal O(T)$ further
elementary arithmetic operations.  Hence the construction has $\mathcal O(N)$
controlled-rotation and arithmetic layers. 

Every spin label has magnitude at most
$j_1+j_2=N+(\alpha+\beta)/2$, and therefore has
$\mathcal O(\log(N+\alpha+\beta+1))$ bits.  Standard reversible fixed-point
arithmetic evaluates each required rational square root or rotation angle,
synthesizes the corresponding controlled one-qubit rotation to error at most
$\delta$, and uncomputes its workspace using
\begin{align*}
    \operatorname{polylog}\!\left(
        \frac{N+\alpha+\beta+1}{\delta}
    \right)
\end{align*}
one-qubit and two-qubit gates and the same polylogarithmic number of reusable
work qubits.  Taking $\delta=\Theta(\epsilon/(N+1))$ and using the telescoping
inequality for products of unitaries bounds the total operator-norm error by
$\epsilon$.  Since $N+1\leq N+\alpha+\beta+1$, substituting this value of
$\delta$ and multiplying by the $\mathcal O(N)$ layers yields the stated gate
bound.

The constituent register and a branch-bit encoding of each coupling path use
$\mathcal O(N)$ qubits; the running spin labels and arithmetic workspace are
recomputed as needed and reused.  This proves the qubit bound.  Finally,
\begin{align*}
    \|\widetilde U^{\dagger}-U^{\dagger}\|
    =\|\widetilde U-U\|,
\end{align*}
so the reverse circuit has the same error and resource bounds.  

\end{proof}

\section{Continuous quantum polynomial transforms} \label{sec:cqopt}

The continuous QOPTs considered here correspond to the polynomial solutions to the continuous hypergeometric equation, 
\all{ 
\mathcal{H} = P_{2}(x) \frac{d^{2}}{dx^{2}} + P_{1}(x) \frac{d}{dx} + \lambda=0.  
} {hypergeo}
where $P_{1}, P_{2}$ are polynomials with $\deg(P_{2}(x)) \leq 2$ and $\deg(P_{1}(x)) \leq 1$. Due to these degree constraints, the spaces $P_{n}(x)$ of polynomials of degree at most $n$ are preserved by the hypergeometric operator, that is, $\mathcal{H}P_{n} \subset P_{n}$. The preservation of polynomial subspaces leads to a second structural property: if $f$ solves \eqref{hypergeo}, then $f'$ solves a hypergeometric equation of the same form. This property imposes strong constraints on the derivatives of the orthogonal polynomial families satisfying the hypergeometric equation. By a classical result of Bochner, the Hermite, Laguerre and Jacobi polynomials are the only orthogonal polynomial families that solve the hypergeometric equation.

The multiplier $\rho(x)$ converts the hypergeometric equation to self-adjoint Sturm--Liouville form. It satisfies the operator-level identity $\rho(x) P_{2}(x) \frac{d^{2}}{dx^{2}} + \rho(x) P_{1}(x) \frac{d}{dx} + \rho(x)\lambda=\frac{d}{dx}(\rho(x) P_{2}(x) \frac{d}{dx}) + \rho(x) \lambda$. This identity is equivalent to
$\frac{d}{dx}(\rho(x) P_{2}(x)) = \rho(x) P_{1}(x)$ and therefore yields the following condition on the logarithmic derivative of $\rho$:
\all{ 
\frac{\rho'(x)}{\rho(x)} = \frac{ P_{1}(x) - P_{2}'(x)} {P_{2}(x)}. 
} {hypergeom}
 Table \ref{tab:1} summarizes the calculation of the weight function for the Hermite, Laguerre and Jacobi polynomials. 

Further, Schr\"odinger operators having the classical orthogonal polynomials as  eigenstates can be obtained from the hypergeometric equation. The Sturm--Liouville equation
$$
\frac{d}{dx}\!\left[\rho(x) P_{2}(x)\,y'(x)\right] + \lambda\,\rho(x)\,y(x) = 0
$$
becomes a Schr\"odinger equation under the change of variables $\psi(x) = \sqrt{\rho(x)}\,y(x)$. After applying this ground state transformation, the Schr\"odinger equation has the form
$$
-\psi''(\xi) + V(\xi)\,\psi(\xi) = E\,\psi(\xi),
$$
for some computable potential $V(x)$ and with ground state energy $E_{0}=0$ and ground state given by $\psi_0 = \sqrt{\rho(x)}$.

\begin{table}

\renewcommand{\arraystretch}{1.6}
\centering
\resizebox{\textwidth}{!}{%
\begin{tabular}{|c|c|c|c|c|c|}
\hline
\textbf{Family} & $P_2(x)$ & $P_1(x)$ & \textbf{Interval} & \textbf{Weight} $\rho(x)$ & \textbf{Coefficient} $\rho(x)P_2(x)$ \\
\hline
Hermite & $1$ & $-2x$ & $(-\infty,\infty)$ & $e^{-x^{2}}$ & $e^{-x^{2}}$ \\
\hline
Laguerre $(\alpha)$ & $x$ & $\alpha+1-x$ & $(0,\infty)$ & $x^{\alpha}\,e^{-x}$ & $x^{\alpha+1}\,e^{-x}$ \\
\hline
Jacobi $(\alpha,\beta)$ & $1-x^{2}$ & $\beta-\alpha-(\alpha+\beta+2)x$ & $(-1,1)$ & $(1-x)^{\alpha}(1+x)^{\beta}$ & $(1-x)^{\alpha+1}(1+x)^{\beta+1}$ \\
\hline
\end{tabular}
}
\caption{Continuous orthogonal polynomial families in Sturm-Liouville form.}
\label{tab:1}
\end{table} 
The ground state transformed Schr\"odinger operators have the corresponding orthonormal polynomials as eigenstates. 
Table \ref{tab:2} computes the Schr\"odinger operators obtained after 
applying the ground state transform with weight $\sqrt{\rho(x)}$ and an appropriate change of variables.  The efficient quantum orthogonal polynomial transform is equivalent to fast forwardability of the Schr\"odinger operators in Table \ref{tab:2}. The Schr\"odinger operators corresponding to the Hermite, Laguerre and Jacobi polynomials are the quantum harmonic oscillator, the radial oscillator and the trigonometric P\"oschl-Teller operator.

\begin{table}[H] 

\renewcommand{\arraystretch}{2.2}
\centering
\resizebox{\textwidth}{!}{%
\begin{tabular}{|c|c|c|c|c|}
\hline
\textbf{Family} & \textbf{Variable} & \textbf{Schr\"odinger Operator} & \textbf{Eigenvalues} $E_n$ & \textbf{Eigenfunctions} $\psi_n$ \\
\hline
Hermite  &
$\xi = x$ &
$-\tfrac12\Delta + \tfrac12 x^{2} $ &
$n + \tfrac12$ &
$e^{-x^{2}/2}\,H_n(x)$ \\
\hline
Laguerre  &
$(x, \alpha) = (r^{2}, \ell+\tfrac12) $ &
$-\tfrac12 \Delta + \!\left[\dfrac{\ell(\ell+1)}{2r^{2}} + \tfrac12 r^{2}\right]\!$ &
$2n + \ell + \tfrac32$ &
$r^{\ell+1}\,e^{-r^{2}/2}\,L_n^{(\ell+\tfrac12)}(r^{2})$ \\
\hline
Jacobi &
$x = \cos 2\theta$ &
$-\tfrac12\Delta + \tfrac12\!\left[\dfrac{\alpha^{2}-\tfrac14}{\sin^{2}\!\theta} + \dfrac{\beta^{2}-\tfrac14}{\cos^{2}\!\theta}\right]\!$ &
$\tfrac12\!\left(2n+\alpha+\beta+1\right)^{2}$ &
$\sin^{\alpha+\tfrac12}\!\theta\,\cos^{\beta+\tfrac12}\!\theta\,P_n^{(\alpha,\beta)}(\cos 2\theta)$ \\
\hline
\end{tabular}
}
\caption{Schr\"odinger operators with classical orthogonal polynomial families as eigenfunctions.}
\label{tab:2}
\end{table}
The quantum Hermite transform has been addressed in \cite{jain2025efficient}. In this work we address the quantum Laguerre transform with two different approaches. The first method works for integer values of the Laguerre parameter and reduces the transform to fast-forwarded quantum harmonic oscillator evolution on a Cartesian grid. The second method uses an efficient quantum Hankel transform on a logarithmic grid and works for all values of the Laguerre parameter. We also consider implementing an efficient Jacobi transform by reducing it to oscillator evolution. The bottleneck is a poly-logarithmic implementation of the Hahn transform. Poly-logarithmic Jacobi transforms for integer-valued Jacobi parameters will be addressed in future work.

\subsection{Quantum Laguerre transform}

\label{sec:radial_osc_not}

The quantum Laguerre transform is the continuous QOPT associated with the one-parameter family of Laguerre polynomials. The Laguerre polynomials arise as eigenfunctions of radial oscillators, so the quantum Laguerre transform can be realized in two ways. The first method uses the separation of two-dimensional isotropic-oscillator eigenfunctions in Cartesian and polar bases. A global coordinate change reduces the Laguerre transform to discrete oscillator evolutions, which can then be implemented on a Cartesian grid using the quantum Hermite transform framework~\cite{jain2025efficient}. The second method uses a logarithmic radial grid and fast-forwards the radial oscillator Hamiltonian on that grid. We develop both methods. The Cartesian reduction applies when the Laguerre parameter $\nu \in \Z$, whereas the logarithmic-grid method applies for all $\nu>0$.

In order to set up the Laguerre transform problem, we first work out the relation between the Laguerre polynomials and the 2D isotropic oscillator eigenfunctions. The 2D isotropic quantum harmonic oscillator has the Hamiltonian
\begin{align}
    \hat{H}_{2}=\frac{1}{2}\left(\hat{x}^2+\hat{y}^2+\hat{p}_{x}^2+\hat{p}_{y}^2\right)
\end{align}
The 2D isotropic oscillator is an example of a superintegrable system. It has degenerate energy levels, and its eigenfunctions separate in both the Cartesian and polar coordinate bases. The splitting of the eigenfunctions in Cartesian coordinates and the Schwinger representations have been discussed in the preliminaries. Here we examine the factorization of the eigenfunctions in the polar basis. The splitting is obtained by observing that the angular momentum operator
\begin{align}
    L_z=-i\left(x\part{}{y}-y\part{}{x}\right)=-i\part{}{\theta}
\end{align}
commutes with $\hat{H}_2$ and separates each energy eigenspace into radial sectors indexed by $m\in \Z$. Separation in polar coordinates writes each joint eigenfunction as a Laguerre polynomial radial factor multiplied by the angular part $e^{im\theta}$. The normalized eigenfunctions are given by
\begin{align}
        \Psi_{n_r,m}(r,\theta)
        =\sqrt{\frac{n_r!}{\pi(n_r+|m|)!}}\,
            r^{|m|}e^{-r^2/2}L_{n_r}^{(|m|)}(r^2)e^{im\theta},
        \label{eq:radial-eigenfunction}
\end{align}
with radial degree $n_r\in\N_0$ and angular index $m\in\Z$. The two conserved
labels are the energy $2n_r+|m|+1$ and the angular momentum $m$.

Fixing an angular sector leaves a single-parameter Laguerre family in the radial
variable with order $\nu=|m|$. After the unitary conjugation $R(r)\mapsto r^{1/2}R(r)$ from $L^2(\R_+,r\,dr)$ to $L^2(\R_+,dr)$, the radial factors become
\begin{align}
        \phi_{n,\nu}(r)
        =\sqrt{\frac{2n!}{\Gamma(n+\nu+1)}}\,
            r^{\nu+1/2}e^{-r^2/2}L_n^{(\nu)}(r^2),
        \qquad 0\le n<N,
\end{align}
an orthonormal system in $L^2(\R_+,dr)$.  Inside a fixed order, the polynomial
degree $n$ is the remaining label.

The Laguerre transform is defined to be the isometry that maps a degree register to samples of the fixed order on the sampled grid, 
\begin{align}
    \sum_{n<N}a_n\ket n\ket0_{\rm grid}
    \longmapsto
    \ket0\sum_{n<N}a_n\ket{\phi_{n,\nu}^{(L)}}_{\rm grid},
\end{align}
where $\ket{n}$ is a label register and $\ket{\phi_{n,\nu}^{(L)}}$ denotes the normalized grid state corresponding to $\phi_{n,\nu}$ obtained by restricting the continuous state to the discrete grid and re-normalizing.

The above computation of the Laguerre eigenfunctions applies to $\nu \in \Z_{+}$. In order to extend to $\nu>0$, we use the Friedrichs realization of the radial oscillator defined below, 
\begin{align*}
    H_\nu=\frac12(B_\nu+r^2),
    \qquad
    B_\nu=-\frac{d^2}{dr^2}+\frac{\nu^2-1/4}{r^2}.
\end{align*}
The eigenfunctions are the associated Laguerre polynomials with parameter $\nu$ and the eigenvalues are $2n+\nu+1$, so the polynomial degree appears as a spectral
label. At the critical order $\nu=0$, the Friedrichs form selects the nonlogarithmic $r^{1/2}$ branch. This self-adjoint domain fixes the eigenbasis used by the transform. Phase estimation of
$\exp[-i\pi(H_\nu-(\nu+1)I)/M]$ for normalization constant $M$ reads this spectral label and erases the degree
register.

We present two finite grid based QLT constructions. The first route follows the idea that in an angular momentum sector $m$, the radial degree $n$ is represented by the circular mode occupations $(n+m_+, n+m_-)$ with $m_+=\max(m, 0)$ and $m_-=\max(-m, 0)$. This gives the Cartesian construction for integer orders. The state is converted to Cartesian modes, sampled by two quantum Hermite transforms~\cite{jain2025efficient}, and corrected by a diagonal phase that removes $e^{im\theta}$. 

For the Cartesian construction with integer $\nu=|m|$, introducing the circular
modes
\begin{align}
    a_{\pm}^{\dagger}=\frac{1}{\sqrt{2}}(a_1^{\dagger}\pm i a_2^{\dagger}),
    \qquad
    a_{\pm}=\frac{1}{\sqrt{2}}(a_1\mp i a_2),
\end{align}
In this basis, $\hat H_2=N_++N_-+1$ and $\hat L_z=N_+-N_-$. The state $\ket{n_r+m_+,n_r+m_-}_{+}$, where $m_+=\max(m,0)$ and $m_-=\max(-m,0)$, has radial quantum number $n_r$ and angular momentum $m$.

Fix $m\in\Z$ and write $m_+=\max(m,0)$ and
$m_-=\max(-m,0)$.  With the Hermite functions defined in the oscillator
notation, the circular-mode wavefunction carries a factor $(-1)^n$ relative
to the Laguerre normalization in Equation~\eqref{eq:radial-eigenfunction}.  Introduce the exact
diagonal correction
\begin{align}
    D_{\rm rad}\ket n=(-1)^n\ket n.
\end{align}
The alternate convention $\psi_j\mapsto(-1)^j\psi_j$ changes this formula by
the fixed sector-wide sign $(-1)^{|m|}$, which may be absorbed into the
circuit's global phase.
The reversible embedding
\begin{align}
    A_m\ket n\ket0=\ket{n+m_+,n+m_-}
\end{align}
loads the radial degree into two ordinary Cartesian Fock registers.  We take
$A_m$ to be an in-place basis permutation on suitably padded registers; its
carry bits are uncomputed before the mode change.  Lemma~\ref{CBS} fixes the
direction of that change:
\begin{align}
    \mathsf U_{\textup{C2R}}A_m\ket n
    =\ket{n+m_+,n+m_-}_{+}=\ket{n,m}.
\end{align}
The continuum Hermite transform on each Cartesian mode satisfies
\begin{align}
 (\mathsf{HT}\otimes\mathsf{HT})\mathsf U_{\textup{C2R}}A_m\ket n
 =(-1)^n\ket{\Psi_{n,m}}.
\end{align}
The diagonal correction therefore gives the phase-consistent identity
\begin{align}
 (\mathsf{HT}\otimes\mathsf{HT})\mathsf U_{\textup{C2R}}A_mD_{\rm rad}\ket n
 =\ket{\Psi_{n,m}},
 \label{eq:laguerre-cartesian-continuum-map}
\end{align}
where the coordinate wavefunction of the state on the right is the function
in Equation~\eqref{eq:radial-eigenfunction}.

For the finite circuit, use the centered Hermite grid of
Theorem~\ref{thm:qht}.  Let $L_{\rm C}$ be its one-coordinate dimension,
\begin{align}
    \Delta_{\rm C}=\sqrt{\frac{2\pi}{L_{\rm C}}},\qquad
    x_a=a\Delta_{\rm C},\qquad y_b=b\Delta_{\rm C},
\end{align}
where $-L_{\rm C}/2\leq a,b<L_{\rm C}/2$, and set
$r_{ab}=(x_a^2+y_b^2)^{1/2}$ and
$\theta_{ab}=\operatorname{atan2}(y_b,x_a)$.  Define the raw sampled column
\begin{align}
    \ket{\widetilde\Psi_{n,m}^{(L_{\rm C})}}
    =\Delta_{\rm C}\sum_{a,b}
      \Psi_{n,m}(r_{ab},\theta_{ab})\ket{a,b}.
\end{align}
Its normalized version is
\begin{align}
    \ket{\Psi_{n,m}^{(L_{\rm C})}}
    =\frac{\ket{\widetilde\Psi_{n,m}^{(L_{\rm C})}}}
    {\|\widetilde\Psi_{n,m}^{(L_{\rm C})}\|}.
\end{align}
The diagonal grid phase
\begin{align}
    \mathsf P_m\ket{a,b}
    =e^{-im\theta_{ab}}\ket{a,b}
    \label{eq:laguerre-angular-phase-removal}
\end{align}
removes the Fourier mode.  At the origin we take $\theta_{00}=0$; this choice
is immaterial for $m\ne0$ because the continuum amplitude vanishes there.
The radial Cartesian-grid column is
\begin{align}
 \ket{\Phi_{n,m}^{\mathrm{Cart},(L_{\rm C})}}
 :=\mathsf P_m\ket{\Psi_{n,m}^{(L_{\rm C})}}.
 \label{eq:laguerre-cartesian-target}
\end{align}
Its amplitudes are proportional to
$r_{ab}^{|m|}e^{-r_{ab}^2/2}L_n^{(|m|)}(r_{ab}^2)$, with the normalization
given in Equation~\eqref{eq:radial-eigenfunction}.

\begin{algorithm}[H]
\caption{Cartesian-grid Laguerre transform through the QHT}
\label{alg:laguerre-cartesian-qht}
\textbf{Input:} An integer angular sector $m$, a cutoff $N$, an error
$0<\epsilon<1$, and a normalized state
$\sum_{n<N}a_n\ket n$.

\textbf{Output:} A two-dimensional Cartesian-grid approximation to
$\sum_{n<N}a_n\ket{\Phi_{n,m}^{\mathrm{Cart},(L_{\rm C})}}$.

\begin{enumerate}
    \item Apply $D_{\rm rad}$ and then $A_m$ to load the corrected
    circular-mode occupations.
    \item Apply $\mathsf U_{\textup{C2R}}$ from Lemma~\ref{CBS}.
    \item Apply $\mathsf{QHT}\otimes\mathsf{QHT}$ to the two Cartesian
    oscillator modes.
    \item Apply the diagonal phase $\mathsf P_m$ in
    \eqref{eq:laguerre-angular-phase-removal}.  Omit this step when the full
    angular wavefunction $\Psi_{n,m}$ is required.
\end{enumerate}
\textbf{Return:} The two Cartesian grid registers, with all arithmetic work
uncomputed.
\end{algorithm}

\begin{corollary}[Cartesian-grid Laguerre transform from the QHT]
\label{cor:laguerre-cartesian-qht}
Let $m\in\Z$, $N\geq1$, and $0<\epsilon<1$.  Put
\begin{align}
    N_{\rm Cart}=2N+|m|,
    \qquad
    N_{\rm F}=\Theta\!\left(N_{\rm Cart}+\log(1/\epsilon)+1\right),
\end{align}
where the implicit constant is chosen as in
Theorem~\ref{thm:digital_implementation_gaussian_optical_gates}.
There is a power-of-two Cartesian grid with
\begin{align}
    L_{\rm C}=O\!\left(N_{\rm F}^{9/4}\epsilon^{-13/4}\right)
\end{align}
points per coordinate and a digital circuit
$\mathsf{LagCart}_{m,N,\epsilon}$ such that, for every normalized
$\sum_{n<N}a_n\ket n$,
\begin{align}
 \left\|
 \mathsf{LagCart}_{m,N,\epsilon}
 \left(\sum_{n<N}a_n\ket n\ket0\ket0_{\rm work}\right)
 -\left(\sum_{n<N}a_n
 \ket{\Phi_{n,m}^{\mathrm{Cart},(L_{\rm C})}}\right)
 \ket0_{\rm work}
 \right\|_2\leq\epsilon.
 \label{eq:laguerre-cartesian-qht-guarantee}
\end{align}
The gate count and live workspace are polynomial in
$\log(N_{\rm F}+1)$ and $\log(1/\epsilon)$.
\end{corollary}

\begin{proof}
For $n<N$, the state $A_m\ket n$ has total occupation
$2n+|m|<N_{\rm Cart}\leq N_{\rm F}$.  Every factor in
$\mathsf U_{\textup{C2R}}$ preserves this
total, so both Cartesian modes remain below the cutoff $N_{\rm F}$.  Lemma
\ref{CBS}, the stated Hermite convention, and Equation~\eqref{eq:radial-eigenfunction} give
\eqref{eq:laguerre-cartesian-continuum-map}; multiplication by
$\mathsf P_m$ gives \eqref{eq:laguerre-cartesian-target}.

Apply Theorem~\ref{thm:qht} to the two oscillator registers and use
Theorem~\ref{thm:digital_implementation_gaussian_optical_gates} for the four
fixed-angle Gaussian factors in $\mathsf U_{\textup{C2R}}$.  Here $K=2$, and
the number of factors and their amplification parameters are constants.
In the composed circuit, the inverse QHTs at the output of the digital
Gaussian implementation cancel against the two forward QHTs in
Algorithm~\ref{alg:laguerre-cartesian-qht}.  One may therefore run the two
initial QHTs followed directly by the quadrature-grid realization of
$\mathsf U_{\textup{C2R}}$.
Allocating a constant fraction of the error to each call gives the displayed
choice of $L_{\rm C}$ and polynomial dependence on the register lengths.  The
phase in \eqref{eq:laguerre-angular-phase-removal} is evaluated reversibly to
the same accuracy.  The QHT sampling estimate also places every raw column
norm within the allocated error of one; normalizing the columns changes the
bound by another constant multiple of that error.  The cited bounds act on
the full cutoff space, so the estimate holds for arbitrary coherent
superpositions of $n$.
\end{proof}

The Cartesian construction is tied to the angular sectors of the planar
oscillator and hence to integer order $\nu=|m|$.  Its output occupies an
$L_{\rm C}\times L_{\rm C}$ grid.  The construction on the logarithmic grid below uses one
radial register and accepts any fixed nonnegative order. The finite targets differ, while both sample the Laguerre functions
obtained from the same radial eigenvalue equation.

\subsection{Quantum Jacobi transform}
We briefly describe the relation of the quantum Jacobi transform to coupled radial oscillators and the Hahn transform. The discussion identifies the bottlenecks for obtaining an efficient quantum Jacobi 
transform. 

The quantum Hahn and Jacobi transforms are closely related to the embedding of the angular
P\"oschl--Teller Hamiltonian into the Smorodinsky--Winternitz oscillator
\all{
H_{\mathrm{SW}} = \tfrac{1}{2}\!\left(-\Delta + x_1^{2}+x_2^{2}
+ \tfrac{g_1}{x_1^{2}} + \tfrac{g_2}{x_2^{2}}\right),
\qquad g_j = k_j(k_j-1),\quad k_j > \tfrac{1}{2},
} {SWO} 
acting on $L^{2}((0,\infty)^{2}, dx_1\,dx_2)$. The Smorodinsky--Winternitz oscillator is superintegrable, as is the isotropic 2D quantum harmonic oscillator.
The polar decomposition of the eigenfunctions of the oscillator is given by the following theorem. 

\begin{theorem}\label{thm:sw-polar}
The eigenfunctions for the Smorodinsky--Winternitz Hamiltonian
$H_{\mathrm{SW}}$ are
\al{
\ket{n_r, l}
&= C_{n_r,l}\;
r^{2l + k_1 + k_2}\,e^{-r^{2}/2}\,
L_{n_r}^{\,2l + k_1 + k_2}(r^{2}) \\
&\quad \times (\sin\theta)^{k_2}(\cos\theta)^{k_1}\,
P_l^{(k_2 - 1/2,\, k_1 - 1/2)}(\cos 2\theta),
}
with normalization constant $C_{n_r,l}$ determined by the orthogonality
calculation. The eigenvalues for the commuting Hamiltonians are given by
\al{
H_{\mathrm{SW}}\,\ket{n_r, l}
&= (2 n_r + 2 l + k_1 + k_2 + 1)\,\ket{n_r, l}, \\
A(\theta)\,\ket{n_r, l}
&= (2 l + k_1 + k_2)^{2}\,\ket{n_r, l},
}
for $n_r, l \in \N_0$, where $L_{n_r}^{\alpha}$ is the associated
Laguerre polynomial and $P_l^{(\alpha,\beta)}$ is the Jacobi polynomial.
\end{theorem}
Besides the two commuting operators $(H_{\mathrm{SW}}, A)$ that define the polar separation, there is a third algebraically independent second--order
constant of motion. It is the separation constant/number difference operator for the coupled radial oscillators,
\al{
B \;:=\; h_{k_1} - h_{k_2}
\;=\; \tfrac{1}{2}\!\left(
-\partial_{x_1}^{2} + x_1^{2} + \tfrac{g_1}{x_1^{2}}
+ \partial_{x_2}^{2} - x_2^{2} - \tfrac{g_2}{x_2^{2}}
\right),
}
diagonal in the Cartesian basis $\ket{n_1, n_2}$ with eigenvalue
$2(n_1 - n_2) + (k_1 - k_2)$, while $A$ is diagonal in the polar basis
$\ket{n_r, l}$ with eigenvalue $(2l + k_1 + k_2)^{2}$. Both commute with
the total energy,
\al{
[H_{\mathrm{SW}}, A] = [H_{\mathrm{SW}}, B] = 0,
}
but $A$ and $B$ are the second members of the two distinct maximal
commuting sets $\{H_{\mathrm{SW}}, A\}$ and $\{H_{\mathrm{SW}}, B\}$, and
they do not commute with each other. However, the operators $H_{SW}, A, B$ do not close into a Lie algebra. This implies that the method used to obtain the Laguerre transform does not generalize to the Jacobi transform.

The Hahn transform appears as the overlap between the uncoupled $\ket{n_{1}, n_{2}}$ and coupled $\ket{n_{r}, l}$ bases on a single energy shell for the Smorodinsky--Winternitz oscillator. A Hahn transform with poly-logarithmic resource requirements, coupled with efficient Laguerre transforms of integer order, would therefore suffice to obtain a superposition of Jacobi polynomial states $\sum_{l} \alpha_{l} \ket{ N-l, l}$. However, this state is still not the full Jacobi transform because the radial label $l$ needs to be erased. The bottleneck for an efficient quantum Jacobi transform is therefore a procedure that measures the label $l$ efficiently. The development of an efficient quantum Jacobi transform therefore requires techniques beyond the ones considered here and is left to future work.

\section{Quantum Laguerre transform on a logarithmic grid}

This section constructs a fixed-order quantum Laguerre transform whose output is represented on a logarithmic radial grid. The grid points are uniform in $x=\log r$, rather than in the radial coordinate $r$. Every step of the construction uses this same representation. In particular, the result does not include a conversion from logarithmically spaced radii to a grid that is uniform in $r$.

The logarithmic coordinate is useful because it makes the Hankel transform accessible to the quantum Fourier transform. The Hankel kernel depends on the product of its radial arguments. After taking logarithms, this product becomes a sum, so the Hankel transform becomes a Fourier multiplier followed by a reflection. The potential term of the radial oscillator remains diagonal because $r^2=e^{2x}$ can be evaluated from the log-grid coordinate.

The construction has three parts. First, we identify the Laguerre functions with the eigenfunctions of a fixed-order radial oscillator. Second, we factor radial oscillator evolution into four Hankel transforms and three diagonal chirps. Third, we discretize the Hankel transform by FFTLog and use the resulting radial evolution in phase estimation. A WKB trial state supplies a nonzero component of the requested eigenfunction before phase estimation isolates its degree. The endpoint, aliasing, state-preparation, and resource estimates are proved in Appendix~\ref{app:log-grid-analysis}.

\subsection{Radial oscillators and Laguerre functions} \label{sec:radial-normal-form}
This subsection defines the radial oscillator whose eigenbasis will be prepared by the circuit. We begin with the two-dimensional isotropic oscillator because its angular momentum sectors explain the radial operator for integer order. The Friedrichs realization then extends the same eigenvalue problem to every fixed real order $\nu\ge0$.

The two-dimensional isotropic quantum harmonic oscillator $\hat{H}_{2}=\frac{1}{2}\left(\hat{x}^2+\hat{y}^2+\hat{p}_{x}^2+\hat{p}_{y}^2\right)$ commutes with the angular momentum operator $L_z=-i\part{}{\theta}$ and can be written as
\begin{align}
    \hat H_2=\frac12\left(-\partial_r^2-r^{-1}\partial_r
    -r^{-2}\partial_\theta^2+r^2\right).
\end{align}
Its eigenfunctions therefore separate into angular and radial factors.
The angular eigenfunctions are $e^{im\theta}$, and the radial equation in the $m$ sector becomes
\begin{align}
    \frac12\left(-R''-r^{-1}R'+m^2r^{-2}R+r^2R\right)=ER
\end{align}
Substituting $s=r^2$ and using the ansatz $R(r)=r^{|m|}e^{-s/2}y(s)$ gives the Laguerre equation
\begin{align}
    sy''+(|m|+1-s)y'+\frac{E-|m|-1}{2}y=0.
\end{align}
The solution regular at $r=0$ is the confluent hypergeometric solution of this equation. For a generic value of $E$, that solution grows too rapidly as $s\to\infty$ after the factor $e^{-s/2}$ is restored, and the radial function is not square integrable. Square integrability requires termination of the hypergeometric series. Thus
\begin{align}
    \frac{E-|m|-1}{2}=n\in\mathbb Z_{\ge0},
\end{align}
so $y=L_n^{(|m|)}$ and $E=2n+|m|+1$, yielding the standard quantization conditions for the radial oscillator. 

The one-dimensional radial Hilbert space is $L^2(\R_+,dr)$. The map $R(r)\mapsto r^{1/2}R(r)$ sends $L^2(\R_+,r\,dr)$ unitarily to $L^2(\R_+,dr)$ and transforms the radial kinetic term into 
\begin{align}
    -\frac{d^2}{dr^2}+\frac{m^2-1/4}{r^2}.
\end{align}
For any fixed real $\nu\ge0$, set
\begin{align}
    B_\nu=-\frac{d^2}{dr^2}+\frac{\nu^2-1/4}{r^2},
    \qquad
    H_\nu=\frac12(B_\nu+r^2).
\end{align}
We use the Friedrichs realization of this operator. Its endpoint condition selects the regular $r^{\nu+1/2}$ branch, and at $\nu=0$ it excludes the independent $r^{1/2}\log r$ branch.

\begin{proposition}[Laguerre basis choice]
\label{prop:laguerre-radial-basis}
For every fixed $\nu\ge0$, the functions
\begin{align}
    \phi_{n,\nu}(r)=
    \sqrt{\frac{2n!}{\Gamma(n+\nu+1)}}\,
    r^{\nu+1/2}e^{-r^2/2}L_n^{(\nu)}(r^2),
    \qquad n\ge0,
\end{align}
form an orthonormal basis of $L^2(\R_+,dr)$ for the Friedrichs realization of
$H_\nu$, and
\begin{align}
    H_\nu\phi_{n,\nu}=(2n+\nu+1)\phi_{n,\nu}.
\end{align}
\end{proposition}

\begin{proof}
Put $s=r^2$ and $f(r)=r^{\nu+1/2}e^{-s/2}y(s)$.  Direct differentiation gives
\begin{align}
 H_\nu f=r^{\nu+1/2}e^{-s/2}
 \left[-2\bigl(sy''+(\nu+1-s)y'\bigr)+(\nu+1)y\right].
\end{align}
The Laguerre equation yields the displayed eigenvalue for $y=L_n^{(\nu)}$.
The same substitution gives
\begin{align}
 \langle\phi_{n,\nu},\phi_{k,\nu}\rangle
 =\frac{\sqrt{n!k!}}
 {\sqrt{\Gamma(n+\nu+1)\Gamma(k+\nu+1)}}
 \int_0^\infty s^\nu e^{-s}L_n^{(\nu)}(s)L_k^{(\nu)}(s)\,ds
 =\delta_{nk}.
\end{align}
Completeness follows from density of polynomials in
$L^2(\R_+,s^\nu e^{-s}ds)$ for $\nu>-1$.
\end{proof}

Proposition~\ref{prop:laguerre-radial-basis} makes the polynomial degree accessible to phase estimation. Since
\begin{align}
    H_\nu\phi_{n,\nu}=(2n+\nu+1)\phi_{n,\nu}
\end{align}
subtracting the fixed offset $\nu+1$ leaves the eigenvalue $2n$. We therefore use the offset-corrected evolution
\begin{align}
    U_M=\exp\!\left[-\frac{i\pi}{M}(H_\nu-(\nu+1)I) \right]
\end{align}
On a Laguerre eigenfunction,
\begin{align}
    U_M\phi_{n,\nu}=e^{-2\pi in/M}\phi_{n,\nu}
\end{align}
Thus phase estimation on $U_M$ returns the degree $n$ modulo $M$. Choosing $M$ above the retained spectral cutoff makes the phase label equal to $n$ on the certified subspace. This identifies the spectral labels that phase estimation must recover. What remains is to implement the controlled powers of $U_M$, and therefore the time evolution under $H_\nu$.

\subsection{Radial factorization}
We next reduce radial oscillator evolution to operations that can be implemented on the logarithmic grid. The potential $r^2$ is diagonal in the radial coordinate. The Hankel transform makes the kinetic operator $B_\nu$ diagonal in the same way that the Fourier transform diagonalizes momentum on the line. This exchange permits an exact seven-factor decomposition into Hankel transforms and radial chirps.

Fix a real $\nu\ge0$. Let $J_\nu$ be the Bessel function of the first kind of
order $\nu$, the solution of
$z^2J_\nu''+zJ_\nu'+(z^2-\nu^2)J_\nu=0$ regular at the origin, normalized by
$J_\nu(z)\sim(z/2)^\nu/\Gamma(\nu+1)$ as $z\to0$. 
Define the symmetric Hankel transform by
\begin{align}
    (\mathcal H_\nu f)(r)
    =\int_0^\infty \sqrt{rs}\,J_\nu(rs)f(s)\,ds
\end{align}
The Bessel kernel is chosen because it diagonalizes the kinetic term of the radial oscillator: 
for every $k>0$ the
function $u_k(r)=\sqrt{r}\,J_\nu(kr)$ is a generalized eigenfunction for $B_\nu$,
\begin{align}
\label{eq:bessel-eigen}
    B_\nu u_k=k^2u_k. 
\end{align}
The transform $\mathcal H_\nu$ changes basis from the radial
coordinate to $\{u_k\}_{k>0}$, in which $B_\nu$ is multiplication by $k^2$. It is
its own inverse. The next lemma states both facts precisely, together with the
consequence used below: $e^{-i\gamma B_\nu^{\mathrm F}}=\mathcal H_\nu C_\gamma
\mathcal H_\nu$, so kinetic evolution is a radial chirp between two Hankel
transforms.

\begin{Lemma}[Hankel exchange]
\label{lem:hankel-exchange}
On the Friedrichs radial space,
\begin{align}
    \mathcal H_\nu^2=I,
    \qquad
    \mathcal H_\nu B_\nu\mathcal H_\nu=r^2,
    \qquad
    \mathcal H_\nu r^2\mathcal H_\nu=B_\nu.
\end{align}
\end{Lemma}

\begin{proof}
The kernel $\sqrt{rs}\,J_\nu(rs)$ is real and symmetric. Bessel orthogonality first gives $\mathcal H_\nu^2=I$ for smooth compactly supported functions, where the integral calculation is justified. Since these functions are dense in $L^2(\R_+,dr)$ and $\mathcal H_\nu$ is bounded, the identity holds on all of $L^2(\R_+,dr)$. 

The exchange with $B_\nu$ comes from the Bessel equation
\begin{align}
    B_{\nu,s}\bigl(\sqrt{rs}J_\nu(rs)\bigr) = r^2\sqrt{rs}J_\nu(rs)
\end{align}
For smooth functions satisfying the Friedrichs endpoint condition, integration by parts transfers $B_\nu$ from the input function to the Bessel kernel. The endpoint condition removes the boundary term at $s=0$, and decay removes the boundary term at infinity.
\begin{align}
    \mathcal H_\nu B_\nu f=r^2\mathcal H_\nu f
\end{align}
for those smooth functions. Every function in the Friedrichs domain can be approximated by such smooth functions together with the action of $B_\nu$.  Passing to the limit gives the same identity on the full Friedrichs domain. Conjugating by the unitary involution $\mathcal H_\nu$ gives
\begin{align}
    \mathcal H_\nu B_\nu\mathcal H_\nu=r^2,
    \qquad
    \mathcal H_\nu r^2\mathcal H_\nu=B_\nu
\end{align}
\end{proof}
The two terms of $H_\nu$ are now diagonal in bases connected by $\mathcal H_\nu$. Thus $e^{-i\gamma r^2}$ is a radial chirp, while $e^{-i\gamma B_\nu}$ is the same chirp conjugated by two Hankel transforms. It remains to express the full oscillator evolution as a constant number of these operations.

Such an expression exists and is exact:
the real span of $ir^2$, $iB_\nu^{\mathrm F}$ and the dilation generator
$r\partial_r+\frac12$, which is skew-adjoint on $L^2(\R_+,dr)$, is a Lie algebra
isomorphic to $\mathfrak{sl}(2,\mathbb{R})$, and the factorization below is the
operator counterpart of the factorization of a rotation in $SL(2,\mathbb{R})$
into three shear matrices.

\begin{proposition}[Radial Hankel--chirp factorization]
\label{prop:radial-factorization}
For $|t|<\pi$, put
\begin{align}
    a(t)=\frac12\tan(t/2),
    \qquad
    b(t)=\frac12\sin t,
    \qquad
    C_\gamma=e^{-i\gamma r^2}
\end{align}
Then the Friedrichs radial oscillator satisfies
\begin{align}
    e^{-itH_\nu}
    &=e^{-ia(t)B_\nu}e^{-ib(t)r^2}e^{-ia(t)B_\nu} \\
    &=\mathcal H_\nu C_{a(t)}\mathcal H_\nu C_{b(t)}
      \mathcal H_\nu C_{a(t)}\mathcal H_\nu
\end{align}
\end{proposition}

\begin{proof}
Let $X=r^2$ and $\mathcal A=r\partial_r+1/2$. On a common analytic core,
\begin{align}
    [\mathcal A,X]=2X,
    \qquad
    [\mathcal A,B_\nu]=-2B_\nu,
    \qquad
    [X,B_\nu]=4\mathcal A.
\end{align}
These commutators close an $\mathfrak{su}(1,1)$ representation. If $U(t)=e^{-ia(t)B_\nu}e^{-ib(t)X}e^{-ia(t)B_\nu}$, product differentiation and the adjoint action of these commutators give
\begin{align}
    \dot U(t)U(t)^{-1}=-\frac{i}{2}(B_\nu+X)=-iH_\nu,
    \qquad U(0)=I
\end{align}
Thus $U(t)=e^{-itH_\nu}$. Lemma~\ref{lem:hankel-exchange} converts each $B_\nu$ shear into a Hankel-conjugated chirp.
\end{proof}

Proposition~\ref{prop:radial-factorization} leaves only the Hankel transform to implement. Each chirp is diagonal and requires only the phase $e^{-i\gamma r^2}$, if we work on a logarithmic grid the phase factor remains diagonal. We can therefore move to the logarithmic coordinate, where the Hankel transform becomes a Fourier multiplier followed by frequency reflection.

\subsection{Hankel transform in logarithmic coordinates}
\label{sec:fftlog}
We argue that in logarithmic coordinates the Hankel transform becomes a circulant matrix composed with an index reversal, and is therefore diagonal in Fourier space. Recall that a kernel matrix that depends on the \emph{difference} of its arguments 
is circulant, and that circulant matrices are exactly the matrices that are diagonalized by the discrete Fourier transform. The Hankel kernel
$\sqrt{rs}J_\nu(rs)$ depends on $r$ and $s$ only through their product, so under the logarithmic change of coordinates $(x, y) = (\log r, \log s)$
(equivalently $r=e^x$, $s=e^y$), the kernel entry is a function of the 
sum $x+y$. A matrix whose entries
are constant along antidiagonals is a circulant matrix composed with the reversal
$y\mapsto-y$; conjugating by the Fourier transform therefore leaves a diagonal
matrix together with the reflection $\omega\mapsto-\omega$ that the reversal
induces on frequencies. The diagonal entries are the Fourier coefficients of the
kernel in the log coordinate, and they evaluate in closed form to a ratio of
Gamma functions,
\begin{align}
\label{eq:mellin-multiplier}
    m_\nu(\omega)=2^{-i\omega}
    \frac{\Gamma((\nu+1-i\omega)/2)}{\Gamma((\nu+1+i\omega)/2)} .
\end{align}
For real $\nu$ and $\omega$ the two Gamma arguments are complex conjugates, so
with $z=(\nu+1+i\omega)/2$ the ratio is $\overline{\Gamma(z)}/\Gamma(z)$. Let $\operatorname{Log}\Gamma$ denote the continuous logarithm on the right half-plane that is real on the positive axis. Then
\begin{align}
\label{eq:mellin-phase}
    m_\nu(\omega)=e^{i\theta_\nu(\omega)},
    \qquad
    \theta_\nu(\omega)=-\omega\log 2-2\operatorname{Im}\operatorname{Log}\Gamma(z) ;
\end{align}
every diagonal entry is a pure phase. An integral transform with a transcendental
kernel is thus reduced to a Fourier transform, a table of phases, and a
reflection, with the order $\nu$ entering only through the table, so that
arbitrary real $\nu$ costs no more than the integer case.

The preceding anti-diagonal description anticipates the finite circuit. We first establish the corresponding continuous identity. Define the unitary
\begin{align}
    (Sf)(x)=e^{x/2}f(e^x),
    \qquad x=\log r,
\end{align}
from $L^2(\R_+,dr)$ to $L^2(\R,dx)$. After this change of variables, the Hankel transform acts on Fourier modes in $x$ by a scalar phase and a frequency reflection. The scalar is obtained by evaluating the classical Mellin transform of the Bessel kernel.

\begin{proposition}[Mellin diagonalization of the Hankel transform]
\label{prop:mellin-hankel}
With
\begin{align}
    (\mathcal F_-u)(\omega)=\int_\R e^{-i\omega x}u(x)\,dx,
    \qquad
    (R\widehat u)(\omega)=\widehat u(-\omega),
\end{align}
one has
\begin{align}
    S\mathcal H_\nu S^{-1}=\mathcal F_-^{-1}M_\nu R\mathcal F_-,
\end{align}
where
\begin{align}
    (M_\nu\widehat u)(\omega)=m_\nu(\omega)\widehat u(\omega),
    \qquad
    m_\nu(\omega)=2^{-i\omega}
    \frac{\Gamma((\nu+1-i\omega)/2)}{\Gamma((\nu+1+i\omega)/2)}.
\end{align}
For real $\omega$, $|m_\nu(\omega)|=1$ and
$m_\nu(-\omega)=\overline{m_\nu(\omega)}$.
\end{proposition}

\begin{proof}
For $u\in C_c^\infty(\R)$,
\begin{align}
    (S\mathcal H_\nu S^{-1}u)(y)
    =\int_\R e^{x+y}J_\nu(e^{x+y})u(x)\,dx.
\end{align}
Taking $\mathcal F_-$ and substituting $t=e^{x+y}$ gives
\begin{align}
 (\mathcal F_-S\mathcal H_\nu S^{-1}u)(\omega)
 =\left(\int_0^\infty t^{-i\omega}J_\nu(t)\,dt\right)
     (\mathcal F_-u)(-\omega).
\end{align}
The Mellin-Bessel identity
\begin{align}
 \int_0^\infty t^zJ_\nu(t)\,dt
 =2^z\frac{\Gamma((\nu+1+z)/2)}{\Gamma((\nu+1-z)/2)}
\end{align}
at $z=-i\omega$ gives the multiplier.
For real $\omega$, the two Gamma inputs are complex conjugates. Since $\Gamma(\overline z) = \overline{\Gamma(z)}$, $|m_\nu(\omega)|=1$, and replacing $\omega$ by $-\omega$ swaps the two Gamma factors and conjugates $2^{-i\omega}$, which gives $m_\nu(-\omega) = \overline{m_\nu(\omega)}$.
\end{proof}

The identity of Proposition~\ref{prop:mellin-hankel} is exact; discretization
enters now, and with it the only approximation in the construction. Sampling the
log coordinate on a finite periodic grid replaces the kernel by its
anti-circulant approximant and the resulting block is exactly
diagonalized by the discrete Fourier transform.

Fix endpoints $x_-<x_+$, let $L=2^{q_L}$, and set
\begin{align}
    \Delta x=\frac{x_+-x_-}{L},
    \qquad
    x_j=x_-+j\Delta x,
    \qquad 0\le j<L.
\end{align}
Let $\omega_L$ be a smooth window supported on $[x_-,x_+]$. The sampling map used throughout the construction is
\begin{align}
    (\mathsf S_Lu)_j=\sqrt{\Delta x}\,\omega_L(x_j)u(x_j).
    \label{eq:laguerre-log-sampling}
\end{align}
For each degree $n$, define the normalized sampled Laguerre state
\begin{align}
    \ket{\Phi_{n,\nu}^{(L)}}
    =\frac{\mathsf S_LS\phi_{n,\nu}}{\|\mathsf S_LS\phi_{n,\nu}\|_2}.
    \label{eq:laguerre-log-target}
\end{align}
The basis state $\ket j$ represents the log coordinate $x_j$ and hence the physical radius $r_j=e^{x_j}$. The radii are geometrically spaced. The construction never resamples these amplitudes onto a grid that is uniform in $r$.

The $L$-point quantum Fourier transform is
\begin{align}
    F_L\ket{j}=L^{-1/2}\sum_{k=0}^{L-1}e^{2\pi ijk/L}\ket{k}
\end{align}
and its output is periodic. To evaluate the continuum multiplier at the correct signed frequency, define
\begin{align}
    \kappa_L(k)=
    \begin{cases}
        k, & 0\le k\le L/2,\\
        k-L, & L/2<k<L
    \end{cases}
    \qquad
    \omega_k=\frac{2\pi\kappa_L(k)}{L\Delta x},
    \qquad
    R_L\ket{k}=\ket{-k\bmod L}
\end{align}

The finite FFTLog block is
\begin{align}
    W_{\nu,L}=F_L^\dagger D R_L F_L,
    \qquad
    D\ket k=e^{-i\theta_\nu(\omega_k)-ic_{\rm lr}\omega_k}\ket k
\end{align} 
where $m_\nu(\omega)=e^{i\theta_\nu(\omega)}$ as in \eqref{eq:mellin-phase}. The conjugate phase appears because the quantum Fourier transform above uses the positive exponential convention. The low-ringing center correction $c_{\rm lr}$ and the exact Nyquist sign are specified in Appendix~\ref{app:log-grid-analysis}. They are chosen so that $R_LDR_L=D^\dagger$ and hence $W_{\nu,L}^2=I$.

The chirps are diagonal on the same register:
\begin{align}
    C_{\gamma,L}\ket j=e^{-i\gamma e^{2x_j}}\ket j.
\end{align}
The finite radial propagator corresponding to Proposition~\ref{prop:radial-factorization} is therefore
\begin{align}
    P_L(t)=W_{\nu,L}C_{a(t),L}W_{\nu,L}C_{b(t),L}W_{\nu,L}C_{a(t),L}W_{\nu,L}.
    \label{eq:finite-radial-propagator}
\end{align}

The finite propagator supplies the controlled radial evolution needed for phase estimation. It does not by itself prepare a Laguerre eigenfunction. To prepare the degree-$n$ eigenfunction, the algorithm first constructs a WKB trial state with a fixed positive overlap with that eigenfunction. Phase estimation then separates the trial state into radial energy eigenspaces, and an equality test selects the eigenspace with degree $n$. Fixed-point amplification raises the selected component close to unit amplitude without measuring the degree register. The circuit evaluates the trial amplitudes and all diagonal phases reversibly from the binary degree and grid indices; it does not load an $L$-entry table. The trial-state, filtering, and arithmetic estimates are proved in Appendix~\ref{app:log-grid-analysis}.

\begin{algorithm}[H]
\caption{Fixed-order quantum Laguerre transform on a logarithmic grid}
\label{alg:laguerre-qopt}
\textbf{Input:} Fixed $\nu\ge0$, cutoff $N\ge1$, error $0<\epsilon<1/4$, and a normalized input state $\sum_{n<N}a_n\ket n$.

\textbf{Output:} A state approximating
$\sum_{n<N}a_n\ket{\Phi_{n,\nu}^{(L)}}$ on the logarithmic grid in \eqref{eq:laguerre-log-sampling}.

\begin{enumerate}
    \item Choose an auxiliary spectral cutoff $K_{\rm aux}>N$, a power of two $M=2^m>K_{\rm aux}$, and one logarithmic grid that controls the finite radial evolution and the sampled inner products below $K_{\rm aux}$.
    \item Controlled by $n$, prepare the sampled WKB trial state $\ket{\mu_n}$ on this logarithmic grid.
    \item Run phase estimation for $U_M$ using the finite propagator \eqref{eq:finite-radial-propagator} for each controlled power.
    \item Mark the branch on which the degree written by phase estimation equals $n$. Apply fixed-point amplification using a reflection about the trial state and the equality-marking reflection.
    \item Run the phase-estimation labeler on the prepared state, subtract the resulting degree from the input register modulo $M$, and reverse the labeler. Uncompute every remaining work register.
\end{enumerate}
\textbf{Return:} The logarithmic-grid register in a state within $\epsilon$ of
\begin{align*}
    \sum_{n<N}a_n\ket{\Phi_{n,\nu}^{(L)}}.
\end{align*}
\end{algorithm}

\begin{theorem}[Fixed-order quantum Laguerre transform on a logarithmic grid]
\label{thm:fixed-order-laguerre}
Fix a nonnegative real number $\nu$, independent of $N$ and $\epsilon$. For every $N\ge1$ and $0<\epsilon<1/4$, let
\begin{align}
    B = 1+ \lceil \log_2(N+1) \rceil + \lceil \log_2(\epsilon^{-1}) \rceil
\end{align}
There are effective endpoints $x_-<x_+$, a smooth window $\omega_L$, and a logarithmic radial grid $L=2^{q_L}=\operatorname{poly}_{\nu}(N,\epsilon^{-1})$, with $q_L=O_\nu(B)$, such that Algorithm~\ref{alg:laguerre-qopt} implements
\begin{align}
    \sum_{n=0}^{N-1}a_n\ket n_{\rm lab}\ket0_{\rm grid}\ket0_{\rm work}
    \longmapsto
    \ket0_{\rm lab}\left(\sum_{n=0}^{N-1}a_n\ket{\Phi_{n,\nu}^{(L)}}_{\rm grid}\right)\ket0_{\rm work}
    \label{eq:fixed-order-laguerre-log-map}
\end{align}
to Euclidean-norm error at most $\epsilon$, uniformly over all normalized coefficient vectors $(a_n)_{n=0}^{N-1}$. Here
\begin{align}
    \ket{\Phi_{n,\nu}^{(L)}}
    =\frac{1}{\|\mathsf S_LS\phi_{n,\nu}\|_2}
    \sum_{j=0}^{L-1}\sqrt{\Delta x}\,\omega_L(x_j)(S\phi_{n,\nu})(x_j)\ket j
\end{align}
is the normalized sample of the Laguerre eigenfunction on the logarithmic grid. Its basis states correspond to the geometrically spaced radii $r_j=e^{x_j}$. The circuit does not include resampling onto a grid that is uniform in $r$. The implementation uses
\begin{align}
    \mathcal O_\nu(B^{15})
    &\quad\text{gates},\\
    \mathcal O_\nu(B^{12})
    &\quad\text{qubits}
\end{align}
The constants depend on the fixed order $\nu$.
\end{theorem}

\begin{proof}
The construction has two tasks. It must approximate radial oscillator evolution on a finite grid, and it must create a state with a nonzero component in the requested eigenspace before phase estimation can identify that component. We use the same logarithmic grid for both tasks.

We first choose an auxiliary cutoff $K_{\rm aux}$ above the requested degree range. The WKB trial state for degree $n$ has a fixed positive overlap with $\phi_{n,\nu}$, while its squared norm above degree $K_{\rm aux}$ is at most $C_\nu(n+1)/(K_{\rm aux}+1)$ by Lemma~\ref{lem:laguerre-wkb-trial}. Since $n<N$, taking $K_{\rm aux}\ge C_\nu N\epsilon^{-2}$ makes the norm of this discarded component $O(\epsilon)$. The purpose of $K_{\rm aux}$ is only to identify a finite radial sector on which one grid and one evolution estimate work uniformly.

Theorem~\ref{thm:fftlog-control} now chooses $x_-$, $x_+$, $\omega_L$, and $L$ for the first $K_{\rm aux}$ radial eigenfunctions. It has two consequences. First, the finite propagator $P_L(t)$ approximates $e^{-itH_\nu}$ for every time $|t|\le\pi/2$ used by phase estimation. Second, the sampled Laguerre columns have a Gram matrix close to the identity. The second conclusion matters because the normalized sampled columns are not exactly orthogonal on a finite grid. Proposition~\ref{prop:laguerre-sampled-preparation} replaces them, for the analysis, by an exactly orthonormal family $\{\ket{v_k}\}_{k<K_{\rm aux}}$ whose columns differ from the requested normalized samples by at most the Gram error. The same proposition prepares a sampled WKB state $\ket{\mu_n}$ with a positive overlap, bounded below independently of $n<N$, with $\ket{v_n}$.

Choose $M=2^m>K_{\rm aux}$. On the retained sector, the eigenphases of
\begin{align}
    U_M=\exp\!\left[-\frac{i\pi}{M}(H_\nu-(\nu+1)I)\right]
\end{align}
are $e^{-2\pi i k/M}$ for $0\le k<K_{\rm aux}$, so phase estimation writes the exact degree $k$ rather than only its residue. Its $j$th controlled power requires one bounded-time radial evolution,
\begin{align}
    U_M^{2^j}=e^{i\pi2^j(\nu+1)/M}e^{-it_jH_\nu},
    \qquad
    t_j=\frac{\pi2^j}{M}\le\frac\pi2.
\end{align}
The circuit therefore uses one finite propagator for each power and never repeats a short-time simulation $2^j$ times. Corollary~\ref{cor:laguerre-polar-evolution} transfers each finite-evolution estimate to the orthonormal states $\ket{v_k}$. Summing the errors over the $m$ controlled powers bounds the complete phase-estimation labeler in operator norm on their span. The finite-evolution and arithmetic errors can therefore be chosen so that this labeler has arbitrarily small error on the retained sector.

The equality test keeps only the component whose phase-estimation output equals the input label $n$. This component has constant positive amplitude because of the WKB overlap. Proposition~\ref{prop:laguerre-spectral-filter} shows that one fixed-point amplification sequence works coherently for every $n<N$. It also shows that running the labeler once more, subtracting its output from the input label, and reversing the labeler erases the degree and all phase-estimation work registers without introducing an $n$-dependent phase.

Allocate a fixed fraction of the error budget to spectral truncation and fixed-point amplification. Divide the remaining budget among the finite propagators, sampled-column comparison, and reversible arithmetic calls. Theorem~\ref{thm:fftlog-control} permits the grid to be refined until its evolution and Gram errors meet these choices, while Propositions~\ref{prop:laguerre-sampled-preparation} and~\ref{prop:laguerre-spectral-filter} control the preparation and filtering errors. Their bounds imply that the circuit output differs by at most $\epsilon$ from the map in \eqref{eq:fixed-order-laguerre-log-map}. This estimate holds for arbitrary coherent superpositions because the evolution estimates are operator-norm bounds and the preparation and filtering circuits act as a direct sum over the unchanged degree label.

It remains to count resources. The choices above make $K_{\rm aux}$, $M$, and $L$ polynomial in $N$ and $\epsilon^{-1}$, so their register lengths are $O_\nu(B)$. Proposition~\ref{prop:resource-ledger} evaluates every FFTLog phase, chirp phase, trial-state rotation, and comparison with $O_\nu(B^{12})$ gates and reusable workspace. The full preparation, phase-estimation, amplification, and erasure circuit makes $O_\nu(B^3)$ such calls. The total gate count is therefore $O_\nu(B^{15})$, while reuse of the arithmetic workspace gives the stated $O_\nu(B^{12})$ qubit bound.

All circuit states in this proof use the coordinate $x=\log r$. The map $S$ defines the amplitudes represented on that grid; it is not a basis conversion performed during the circuit. Thus no resampling between a logarithmic grid and a grid uniform in $r$ is included in either the algorithm or its resource bound.
\end{proof}

\begin{remark}[Scope of the grid representation]
Theorem~\ref{thm:fixed-order-laguerre} prepares samples at the geometrically spaced radii $r_j=e^{x_j}$. It does not prepare samples at uniformly spaced radii or on a Cartesian grid. A coherent conversion to either representation would require a separate resampling algorithm and is not assumed here.
\end{remark}

\section{Additional Details on Digital Implementation}
\label{sec:additional_digital_details}

\subsection{Phase space representations of Gaussian optical gates}
\label{subsec:phasespace}

As mentioned in Section \ref{subsec:digital_implementation_optics}, we will leverage the quantum Hermite transform (QHT) of \cite{jain2025efficient} to approximately implement Gaussian optical operations on a digital device. In particular, suppose we have a system of $K$ one-dimensional Cartesian oscillator modes: $|n_1, \dots, n_K \rangle$. Gaussian operations are usually expressed in the oscillator basis, i.e. in terms of mode annihilation and creation operators $\{ a_j, a_j^{\dagger}\}_{j=1}^{K}$. Let $U(a_j, a_j^{\dagger}, a_k, a_k^{\dagger})$ denote a unitary that only acts on modes $j$ and $k$ and is expressed in the oscillator basis. The Hermite transform $\mathsf{HT}$ provides a means to move to the position basis and work efficiently with unitary operators described compactly in terms of the position/momentum quadrature operators $\{ \hat{x}_j, \hat{p}_j\}_{j=1}^{K}$:
\begin{align*}
U(a_j, a_j^{\dagger}, a_k, a_k^{\dagger})|n_1, \dots, n_K\rangle &= (\mathsf{HT}_j^{\dagger} \otimes \mathsf{HT}_k^{\dagger})U(\hat{x}_j, \hat{p}_j, \hat{x}_k, \hat{p}_k)(\mathsf{HT}_j \otimes \mathsf{HT}_k)|n_1, \dots, n_K\rangle\\
&=(\mathsf{HT}_j^{\dagger} \otimes \mathsf{HT}_k^{\dagger})U(\hat{x}_j, \hat{p}_j, \hat{x}_k, \hat{p}_k)|n_1, \dots, n_{j-1}, \psi_{n_j}^{c}, n_{j+1} \dots, n_{k-1}, \psi_{n_k}^{c}, n_{k+1}, \dots n_K\rangle,
\end{align*}
where $\psi_{n_k}^c$ is the $n_k$-th Hermite function/oscillator eigenstate in the position basis. Note that we will build more complicated Gaussian unitaries out of simpler ones $U$ that only act on at most two modes at a time.

To enable more efficient simulation, we will utilize the following decompositions for commonly used Gaussian optical gates. 
\begin{itemize}
    \item \textbf{Displacement}:
\begin{align*}
    D(x_0 + ip_0)  = e^{ip_0 \hat{x}}e^{-ix_0 \hat{p}}e^{-ix_0p_0/2}
\end{align*}

\item \textbf{Single-mode squeezing}
\begin{align*}
S(r) = 
e^{\,i\frac{\gamma(r)}{2}\hat{x}^2}\;
e^{\,i\frac{\beta(r)}{2}\hat{p}^2}\;
e^{\,i\frac{\alpha(r)}{2}\hat{x}^2} 
\end{align*}
with $\alpha(r)=\gamma(r)=-\tanh (\frac{r}{2})$ and $\beta(r)=\sinh(r)$. Note that this implements a squeeze operator in the rotated real quadrature basis $(\hat{x} \pm \hat{p})/\sqrt{2}$ and is not the diagonal matrix with diagonal entries $(e^{-r}, e^{r})$ implementing $S(r)$ in the standard basis.  

\item \textbf{(Real) Beamsplitter}
\begin{align*}
    \textup{BS}_{Y}(\theta)  = e^{-iq(\theta)\hat{x}_2\hat{p}_1}e^{is(\theta)\hat{x}_1\hat{p}_2}e^{-iq(\theta)\hat{x}_2\hat{p}_1}
\end{align*}

\textbf{(Complex) Beamsplitter}

\begin{align*}
    \textup{BS}_{X}(\theta)=  e^{-iq(\theta)\hat{x}_1\hat{x}_2}e^{is(\theta)\hat{p}_1\hat{p}_2}e^{-iq(\theta)\hat{x}_1\hat{x}_2}
\end{align*}
with $q(\theta) =  \tan\frac{\theta}{2}$ and $ s(\theta)=\sin(\theta)$ in both cases. 

\item \textbf{Two-mode Squeezing}
\begin{align*}
&S_2(r) = \text{BS}_{Y}(\pi/4)^{\dagger}[S'(r)\otimes S'(-r)] \text{BS}_{Y}(\pi/4)\\
&= e^{iq(\pi/4)\hat{x}_2\hat{p}_1}e^{-is(\pi/4)\hat{x}_1\hat{p}_2}e^{iq(\pi/4)\hat{x}_2\hat{p}_1}e^{\,i\frac{\gamma(r)}{2}\hat{x}_1^2}\;
e^{\,i\frac{\beta(r)}{2}\hat{p}_1^2}\;
e^{\,i\frac{\alpha(r)}{2}\hat{x}_1^2}e^{\,i\frac{\gamma(-r)}{2}\hat{x}_2^2}\;
e^{\,i\frac{\beta(-r)}{2}\hat{p}_2^2}\;
e^{\,i\frac{\alpha(-r)}{2}\hat{x}_2^2}\\
&\cdot e^{-iq(\pi/4)\hat{x}_2\hat{p}_1}e^{is(\pi/4)\hat{x}_1\hat{p}_2}e^{-iq(\pi/4)\hat{x}_2\hat{p}_1}
\end{align*}
The above identity is well known for the squeeze operator $S'(r)$ diagonal in the $(x, p)$ basis. For our implementation of the squeeze operator, the corresponding two-mode beam splitter identity is
$S_2(r) = \text{BS}_{X}(\pi/4)^{\dagger}[S'(r)\otimes S(-r)] \text{BS}_{X}(\pi/4)$.
\item \textbf{Phase/Rotation}
\begin{align*}
    R(\theta) = e^{-\frac{i}{2}q(\theta)\hat{x}^2}e^{-\frac{i}{2}s(\theta)\hat{p}^2}e^{-\frac{i}{2}q(\theta)\hat{x}^2}.
\end{align*}
\end{itemize}

Note that all of the above optical gates can be expressed as a finite product of simpler evolutions
\begin{align}
\label{eq:unitary_decomp}
\widehat{U}  = \prod_{b=1}^R\exp(i\vartheta_b \hat{O}_b),
\end{align}
where $R$ is a constant,  $\vartheta_b \in \R$ and 
\begin{align*}
    \widehat{O}_b \in \{ \hat{x}_j \hat{x}_k, \hat{x}_k\hat{p}_j, \hat{p}_j\hat{p}_k, \hat x_j^{2}, \hat p_j^{2}, \hat{x}_j, \hat{p}_j\} 
\end{align*}
is the set of elementary linear and quadratic operators. 
Thus, if we can digitally implement each of the simpler evolutions from $\hat{O}_b$, then we can also implement an
arbitrary unitary $\hat{U}$ compiled into Gaussian gates. This enables an optical-to-digital compilation that approximately preserves the number of basic gates, i.e. an $R$-gate optical circuit compiles to a $\widetilde{\mathcal{O}}(R^3)$ one- and two-qubit gate circuit. 

\subsection{The discrete oscillator and discrete quantum Hermite transform}

To implement the above operation digitally, we follow the framework of \cite{jain2025efficient, iyer2026efficientquantumcircuitshighdimensional}. Instead of the continuous harmonic oscillator Hamiltonian $H= \frac{1}{2} (x^{2}+ p^{2})$ we work with a discrete harmonic oscillator $\bar{H} = \frac{1}{2} (\bar{x}^{2} + \bar{p}^{2})$ on $L$ grid points. The discrete position operator $\bar{x}$ is an $L$-dimensional diagonal matrix with entries $(\bar{x})_{jj}= \sqrt{\frac{2\pi}{L}}(-L/2+j)$ and the discrete momentum operator $\bar{p} := F^{-1} \bar{x} F$ is defined using the $L$-dimensional discrete Fourier transform matrix $F$.

The continuous oscillator eigenfunctions $\phi_n(x) = \frac{(-1)^n}{\sqrt{2^n n! \sqrt{\pi}}}e^{-x^2/2}H_n(x)$ are replaced with their discrete approximation, 
\begin{align*}
&|\psi_n\rangle = \left(\frac{2\pi}{L}\right)^{1/4}\sum_{j = -L/2}^{L/2-1}\phi_n(x_j)|j\rangle
\end{align*}
with $x_{j} = j\sqrt{2\pi/L}$. 
The notation $|\psi^c_n\rangle := \ket{\phi_{n}(x)}$ is used to denote continuous harmonic oscillator eigenstates. The goal of the discretization error analysis is to show that the discrete states $\ket{\psi_{n}(x)}$ retain the properties of the continuous states $\ket{\psi_{n}^{c}(x)}$
required for our algorithms. 

The following fact shows the approximate orthonormality for low-energy oscillator states:
\begin{fact}[Fact 4 \cite{jain2025efficient}]
\label{fact:orthogonality}
There exist constants $\gamma, c \in (0,1)$
 such that for all $L$ sufficiently large and all 
$k,\ell \le cL$,
\begin{align}
\label{eqn:app_ortho}
\left| \langle \psi_k | \psi_\ell \rangle - \delta_{k,\ell} \right|
\le \exp(-\gamma L). 
\end{align}
\end{fact}

Further, it can be shown that, for $m\leq cL$, the states $\ket{\phi_{m}}$ are
approximate eigenfunctions of the discrete harmonic oscillator,
\al{ 
\norm{ \bar{H} - (m+ 1/2) \phi_{m} } \leq exp(-\gamma L). 
}
The above equation shows that the discrete eigenstates approximate the Hermite-polynomial-based continuous eigenstates well on the low-energy subspace. This leads us to define two qubit-encoded Hilbert spaces: (1) the \emph{effective space} $\mathcal{H}_N = \{ |n \rangle : 0\leq n < N\}$ (encoded using a $\lceil \log(N)\rceil$-qubit space), and (2) the \emph{operational space} $\mathcal{H}_{L} = \{ |n \rangle : 0\leq n < L\}$. Clearly $\mathcal{H}_N \subset \mathcal{H}_L$.

Working with these two spaces, the efficient quantum Hermite transform of \cite{jain2025efficient} provides the following algorithm mapping states in $\mathcal{H}_{N}$ to superpositions of discrete eigenfunctions of $\bar{H}$.
\quantumHermiteTransform*
In particular, the quantum circuit $\mathsf{QHT}$ applied to $\sum_{m=0}^{N-1} \alpha_m |m\rangle$ produces a normalized quantum state $|\widetilde{\phi}\rangle$ that is at most $\epsilon$ in $\ell_2$ distance from the almost-normalized state $\sum_{m=0}^{N-1} \alpha_m |\psi_m\rangle$. Note that the action of $\mathsf{QHT}$ on $|n\rangle$ for $n>N$ has no guarantees. Thus we will actually be working with:
\begin{align}
\label{eqn:quadrature_qht}
    (\mathsf{QHT}_j^{\dagger} \otimes \mathsf{QHT}_k^{\dagger})U(\overline{x}_j, \overline{p}_j, \overline{x}_k, \overline{p}_k)(\mathsf{QHT}_j \otimes \mathsf{QHT}_k)|n_1, \dots, n_K\rangle,
\end{align}
where the unitary $U(\overline{x}_j, \overline{p}_j, \overline{x}_k, \overline{p}_k)$ is applied to a state that is at most $\epsilon$ away from  \begin{align*}
    |n_1, \dots, n_{j-1}, \psi_{n_j}^{c}, n_{j+1} \dots, n_{k-1}, \psi_{n_k}^{c}, n_{k+1}, \dots n_K\rangle.
\end{align*}
Note that Theorem \ref{thm:qht} has the same $\epsilon$ error bound for inputs that are superpositions.
In contrast to the first pair of QHTs, the inverse step will actually work over larger effective and operational spaces. This is because, unless the operation is photon-number preserving, $U$ will cause some significant leakage outside of $\textup{span}\{ |\psi_n\rangle : 0\leq n < N\}$. Thus inverses actually work with different $N$ and $L$ parameters (or more precisely we will just use the maximum parameter choice for both).

\subsection{Efficient digital implementation of Gaussian optical gates}
\label{sec:eff_implement}
In this section we adapt/generalize the results of \cite{jain2025efficient,iyer2026efficientquantumcircuitshighdimensional, guan2026quantumKrawtchouktransformusing}  to handle the larger Gaussian gate set. We also improve upon \cite{guan2026quantumKrawtchouktransformusing} by handling the non-compact, $\textup{SU}(1,1)$ case. As a result, we also obtain improved bounds on the required digital resources.


Consider a state $|\phi \rangle \in \mathcal{H}_{E}^{\otimes K} = \textup{span}_{\mathbb{C}}\{ |\mathbf{m}\rangle, \max_j m_j \leq E\}$. We define the following mapping from $\mathcal{H}_{N}^{\otimes K}$ into the span of the $L$-point discrete Hermite states $|\psi_{m_1}, \dots, \psi_{m_K}\rangle$:
\begin{align}
\label{eq:sampled-hermite-synthesis}
    \mathsf S_{E,L}\!\left(\sum_{\max_jm_j\leq E}\alpha_{\mathbf m}
    |\mathbf m\rangle\right)
    :=\sum_{\max_jm_j\leq E}\alpha_{\mathbf m}
    |\psi_{m_1},\ldots,\psi_{m_K}\rangle .
\end{align}
We also equip $\mathcal{H}_{E}^{\otimes K}$ with the natural orthogonal projector $Q_{E}$ from $\mathcal{H}_{\infty}^{\otimes K}$.

The map $\mathsf S_{E,L}$ has controlled spectral norm.
\begin{Lemma}
\label{lem:synth_spec_op}
There exists a constant $c >0$ such that if $E \leq c L$, then
\begin{align*}
    L=\Omega(E+\log K) \implies \lVert S_{E, L} \rVert = \mathcal{O}(1).
\end{align*}
\end{Lemma}
\begin{proof}
By Fact
\ref{fact:orthogonality}, there are constants $c, \gamma > 0$ such that, for all $m, n \leq E \leq cL$,
\begin{align*}
    \lvert \langle \psi_{m}, \psi_{n}\rangle - \delta_{m, n}\rvert \leq \exp(-\gamma L).
\end{align*}
Hence, if $G_{E,L}$ is the Gram matrix $(G_{E,L})_{m, n}:=\langle \psi_{m}| \psi_{n}\rangle$, then, because $G_{E,L}$ is Hermitian,
\begin{align*}
    \lVert G_{E,L} - I\rVert &\leq \sqrt{\lVert G_{E,L} - I\rVert_1 \lVert G_{E,L} - I\rVert_{\infty}}\\
    &= {\lVert G_{E,L} - I\rVert_{\infty}} \\
    &\leq (E+1)\exp(-\gamma L).
\end{align*}
Then 
\begin{align*}
    \lVert S_{E,L}S_{E,L}^{\dagger} \rVert^2 = \lVert G_{E,L}^{\otimes K} \rVert \leq (1 +(E+1)\exp(-\gamma L))^{K}.
\end{align*}
Hence 
\begin{align*}
    \lVert S_{E, L}\rVert \leq (1 +(E+1)\exp(-\gamma L))^{K/2}.
\end{align*}
\end{proof}

For the lemma below, we derive an improved version of \cite{iyer2026efficientquantumcircuitshighdimensional}.
\begin{Lemma}[Lemma 3.3 \cite{iyer2026efficientquantumcircuitshighdimensional}, adapted]
\label{lem:basic_gaussian_operation_error}
Let $\widehat W=\exp(i\vartheta\widehat O)$, where
$|\vartheta|\leq1/(4\sqrt{2}e)$ and
\begin{align*}
\widehat O\in\{\hat x_j,\hat p_j,\hat x_j^2,\hat p_j^2,
\hat x_j\hat x_k,\hat x_j\hat p_k,\hat p_j\hat p_k:j\ne k\}.
\end{align*}
Let $\overline W$ be the corresponding $L$-point discretization.  There are
constants $C,c,\gamma>0$ such that, whenever $N\leq cL$,
\begin{align}
\label{eq:elementary-fock-tail}
    \|(I-Q_{5N})\widehat WQ_N\|
    &\leq \eta(N,L),\\
\label{eq:elementary-grid-bridge}
    \|\overline W\mathsf S_{N,L}
      -\mathsf S_{5N,L}Q_{5N}\widehat WQ_N\|
    &\leq \eta(N,L),
\end{align}
where
\begin{align*}
    \eta(N,L)=CN^2\exp(-\gamma N).
\end{align*}
\end{Lemma}
\begin{proof}
Let us define the truncated Taylor series operator
\begin{align*}
    R_{M}(\widehat{O}) := \sum_{k=0}^{M}\frac{(i\vartheta\widehat{O})^{k}}{k!}.
\end{align*}
 Note that we have the following recurrence (using the Fock basis):
\begin{align}
\label{eqn:recursion}
&\hat{x}|m\rangle = \sqrt{\frac{m}{2}}|{m-1}\rangle + \sqrt{\frac{m+1}{2}}|{m+1}\rangle \notag\\
&\hat{x}|0\rangle = \sqrt{\frac{1}{2}}|1\rangle,
\end{align}
By expanding the recurrence \eqref{eqn:recursion} we get
\begin{equation*}
\| \hat{x}^k\,|\psi^c_m\rangle \|
\leq \sqrt{2^k(m+1)(m+2)\cdots(m+k)}.
\end{equation*}
Then, if $m_1, m_2 \leq N$ and $k > N$,
\begin{align*}
\Big\| \frac{(i\vartheta)^k (\widehat{x}_1 \widehat{x}_2)^k}{k!}\,
|m_1, m_2\rangle \Big\|
&\leq \frac{|\vartheta|^k}{k!}
\sqrt{2^{2k}(m_1+1)\cdots(m_1+k)(m_2+1)\cdots(m_2+k)} \\
&\leq \frac{|{2}\vartheta|^k}{k!}\,(N + 1)\cdots(N + k) \\
&\leq \frac{|{2}\vartheta|^k\,(N + k)^k}{k!} \\
&\leq \Big( \frac{e{2}\,|\vartheta|\,(N + k)}{k} \Big)^k
\qquad \text{(Stirling's)} \\
&\leq (4e\lvert \vartheta\rvert)^k.
\end{align*}
Since $e\,|\vartheta| \leq 1/4\sqrt{2}$, we have
\begin{align*}
    \Big\| \frac{(i\vartheta)^k (\widehat{x}_1 \widehat{x}_2)^k}{k!}\,
|m_1, m_2\rangle \Big\| \leq 2^{-k/2}.
\end{align*}
 Similarly, for $m_1, m_2 \leq N, k > 2N$ we have
 \begin{align*}
     \Big\| \frac{(i\vartheta)^k (\widehat{x}_1^2)^k}{k!}\,
|m_1\rangle \Big\| \leq \left(\frac{5}{4\sqrt{2}}\right)^k.
 \end{align*}
Note that $\widehat{x}_1^{k}Q_{N} \subseteq Q_{N +k}$. Hence the $R_{N}(\widehat{x_1}, \widehat{x_2})$ maps $Q_{N}$ into $Q_{2N}$ and $R_{2N}(\widehat{x}_1^2)$ maps $Q_{N}$ into $Q_{5N}$. Thus  
\begin{align*}
    (I - Q_{5N})R_{\infty}Q_N = (I - Q_{5N})(R_{\infty}(\widehat{x}_1\widehat{x}_2)-  R_{N}(\widehat{x}_1\widehat{x}_2))Q_{N}.
\end{align*}
We currently have that for all $m_1, m_2 \leq N$
\begin{align*}
    \lVert (I - Q_{5N})(R_{\infty}(\widehat{x}_1\widehat{x}_2) - R_N(\widehat{x}_1\widehat{x}_2))Q_{N}|m_1, m_2\rangle \rVert \leq 2^{-N/2}
\end{align*}
Hence, we can promote the above to a spectral norm bound at the cost of a $(N+1)^2$ factor. A similar argument bounds $\lVert (I-Q_{5N})(R_{\infty}(\widehat{x}_1^2) - R_{2N}(\widehat{x}_1^2))Q_{N}\rVert$. We can obtain the bound for the momentum operators by conjugating by Fourier transforms and using that $F|\psi_n^c\rangle = i^{n}|\psi_n^c\rangle$.

For the other error bound:
\begin{align*}
 \lVert R_{\infty}(\overline{x}_1\overline{x}_1)S_{N, L} - S_{5N, L}Q_{5N}R_{\infty}(\widehat{x}_1\widehat{x}_2)Q_{N}|m_1, m_2\rangle\rVert 
 &\leq \lVert [R_{\infty}(\overline{x}_1\overline{x}_1)S_{N, L} - R_{N}(\overline{x}_1,\overline{x}_2)S_{N,L}]|m_1, m_2\rangle\rVert \\
 &+ \lVert [R_{N}(\overline{x}_1,\overline{x}_2)S_{N,L} - S_{5N,L}R_{N}(\widehat{x}_1\widehat{x}_2)Q_{N}]|m_1, m_2\rangle\rVert \\
 &+\lVert [S_{5N, L}Q_{5N}(R_{N}(\widehat{x}_1\widehat{x}_2) - R_{\infty}(\widehat{x}_1\widehat{x}_1))Q_{N}]|m_1, m_2\rangle\rVert.
\end{align*}
The middle term is zero. Since $\lVert S_{5N, L}\rVert = \mathcal{O}(1)$ by Lemma \ref{lem:synth_spec_op}, the last term is bounded by $\mathcal{O}(2^{-N/2})$ using the Taylor bound from above. For the first term, we use the fact that the second term is zero to bound: for $k \leq N$,  $\lVert \overline{O}^{k} S_{N, L} \rVert \leq \lVert S_{5N, L} \widehat{O}^{k} Q_N\rVert  = \mathcal{O}\left(\lVert \widehat{O} \rVert^{k}\right)$. Hence, by applying the integral form of the Taylor remainder, the first term is also $\mathcal{O}(2^{-N/2})$. We can promote it to a spectral norm bound by multiplying by $(N+1)^2$.

For the case of $\overline{x}_1^2$, a similar argument applies giving
\begin{align*}
    \lVert R_{\infty}(\overline{x}_1^2)S_{N, L} - S_{5N, L}Q_{5N}R_{\infty}(\widehat{x}_1^2)Q_{N}\rVert  \leq AN\left(\frac{5}{4\sqrt{2}}\right)^{2N},
\end{align*}
for some constant $A > 0$. For the momentum operator, we just need to insert Fourier transforms. The error analysis now becomes
\begin{align*}
 &\lVert [F^{-1}R_{\infty}(\overline{x}_1^2)FS_{N, L} - S_{5N, L}Q_{5N}F^{-1}R_{\infty}(\widehat{x}_1^2)FQ_{N}]|m_1\rangle\rVert\\ 
 &\leq  \lVert [F^{-1}R_{\infty}(\overline{x}_1^2)FS_{N, L} - F^{-1}R_{\infty}(\overline{x}_1^2)S_{N, L}FQ_{N}]|m_1\rangle\rVert\\
 &+\lVert [F^{-1}R_{\infty}(\overline{x}_1^2)S_{N, L}FQ_{N} - F^{-1}S_{5N, L}Q_{5N}R_{\infty}(\widehat{x}_1^2)FQ_{N}]|m_1\rangle\rVert\\
 &+\lVert [F^{-1}S_{5N, L}Q_{5N}R_{\infty}(\widehat{x}_1^2)FQ_{N}- S_{5N, L}Q_{5N}F^{-1}R_{\infty}(\widehat{x}_1^2)FQ_{N}]|m_1\rangle\rVert
\end{align*}
Note that $F$ and $Q_{E}$ commute. Also note that we are slightly abusing notation when switching between the discrete Fourier transform on the discrete Hermite states and the actual Fourier transform. However, we account for the error that the discrete version does not exactly satisfy $F|\psi_m\rangle = i^{m}|\psi_m\rangle$.

By the unitarity of the Fourier transform, the first term is exactly what we previously bounded. For the two other terms we use the following. Reference~\cite{jain2025efficient} notes that there exist positive constants $c >0, \gamma > 0$ such that for all $m \leq c L$:
\begin{equation}
\label{eqn:approximate_evec}
\big\| F\,|\psi_m\rangle - i^m\,|\psi_m\rangle \big\|
\leq \exp(-\gamma L).
\end{equation}
Hence $\lVert [S_{N,L}F - FS_{N,L}]|m_1, m_2\rangle\rVert = \mathcal{O}(\exp(-\gamma L))$, adjusting constants as needed. Note that $R_{\infty}(\widehat{x}_1^2)$ and $R_{\infty}(\overline{x}_1)^2$ are unitary. We can convert it to a spectral norm bound by multiplying by $\mathcal{O}((N+1)^2)$.
\end{proof}

We derive a slightly modified version of Lemma 3.1
\cite{iyer2026efficientquantumcircuitshighdimensional} for our purposes.

\begin{Lemma}
\label{lem:unitary_error}
Let
\begin{align*}
    \widehat U=\widehat W_r\cdots\widehat W_1,
    \qquad
    \overline U=\overline W_r\cdots\overline W_1,
\end{align*}
where every $\widehat W_b$ is an elementary Gaussian factor from Lemma
\ref{lem:basic_gaussian_operation_error}, acts nontrivially on at most two of the $K$ modes,
and $\overline W_b$ is its $L$-point discretization.  Set $E_b=5^bN$ and define
\begin{align}
\label{eq:truncated-fock-prefix}
    T_0=Q_N,
    \qquad
    T_b=Q_{E_b}\widehat W_bT_{b-1},
    \qquad b=1,\ldots,r.
\end{align}
If $L\geq c_0(5^rN+\log K)$ for a sufficiently large constant $c_0$, then
there are constants $C,\overline\gamma>0$, independent of $N,L,K$, and $r$,
such that
\begin{align}
\label{eq:composed-fock-and-grid-error}
\max\left\{
    \|\widehat UQ_N-T_r\|,
    \|\overline U\mathsf S_{N,L}-\mathsf S_{E_r,L}T_r\|
    \right\}
    \leq CN^2 e^{-\overline\gamma N}.
\end{align}
Consequently,
\begin{align}
\label{eq:composed-tail-and-bridge}
\max\left\{
    \|(I-Q_{E_r})\widehat UQ_N\|,
    \|\overline U\mathsf S_{N,L}
       -\mathsf S_{E_r,L}Q_{E_r}\widehat UQ_N\|
    \right\}
    \leq C'N^2 e^{-\overline\gamma' N}
\end{align}
for another pair of constants $C',\overline\gamma'>0$.
\end{Lemma}

\begin{proof}
The local two-mode estimate of Lemma~\ref{lem:basic_gaussian_operation_error} extends to
$K$ modes by tensoring additional $S_{N,L}$ maps. From Lemma \ref{lem:synth_spec_op}, we have that for $L = \Omega(N +\log(K))$, the spectral norm is $\mathcal{O}(1)$. Thus at step $b$
both local bounds in Lemma \ref{lem:basic_gaussian_operation_error} are at most
\begin{align*}
    \eta_b=AE_{b-1}^2e^{-\gamma E_{b-1}}
\end{align*}
for constants independent of $K$.

Fix a normalized input $|\phi\rangle\in\operatorname{ran}(Q_N)$.  Write
\begin{align*}
    |u_b\rangle=\widehat W_b\cdots\widehat W_1|\phi\rangle,
    \qquad
    |t_b\rangle=T_b|\phi\rangle.
\end{align*}
Since every exact factor is unitary and every $Q_{E_b}$ is an orthogonal
projector,
\begin{align*}
|u_b\rangle-|t_b\rangle
&=\widehat W_b(|u_{b-1}\rangle-|t_{b-1}\rangle)
 +(I-Q_{E_b})\widehat W_b|t_{b-1}\rangle,
\end{align*}
and $\|t_{b-1}\|\leq1$.  The first local estimate therefore gives
\begin{align*}
    \|u_b-t_b\|\leq\|u_{b-1}-t_{b-1}\|+\eta_b.
\end{align*}

For the discretized operators, set
\begin{align*}
    |g_b\rangle=\overline W_b\cdots\overline W_1
        \mathsf S_{N,L}|\phi\rangle,
    \qquad
    |h_b\rangle=\mathsf S_{E_b,L}|t_b\rangle.
\end{align*}
Then
\begin{align*}
|g_b\rangle-|h_b\rangle
&=\overline W_b(|g_{b-1}\rangle-|h_{b-1}\rangle)\\
&\quad+\left(\overline W_b\mathsf S_{E_{b-1},L}
-\mathsf S_{E_b,L}Q_{E_b}\widehat W_bQ_{E_{b-1}}\right)|t_{b-1}\rangle.
\end{align*}
The discretized factor $\overline W_b$ is unitary, and the second local estimate
gives
\begin{align*}
    \|g_b-h_b\|\leq\|g_{b-1}-h_{b-1}\|+\eta_b.
\end{align*}
Both errors are consequently bounded by $\sum_{b=1}^r\eta_b$.  Since
$E_b=5^bN$,
\begin{align*}
    \sum_{b=1}^r E_{b-1}^2e^{-\gamma E_{b-1}}
    \leq A'N^2e^{-\gamma' N},
\end{align*}
which proves \eqref{eq:composed-fock-and-grid-error}.  The first bound in
\eqref{eq:composed-tail-and-bridge} follows because $T_r$ is supported on
$Q_{E_r}$.  For the second, insert $\mathsf S_{E_r,L}T_r$ and use
Lemma \ref{lem:synth_spec_op} together with the two bounds in
\eqref{eq:composed-fock-and-grid-error}.
\end{proof}

Via the Hermite transform, any basic Gaussian operator $U_{G}(\theta)$ with evolution time $\theta$ can  be expressed in the quadrature basis as in \ref{subsec:phasespace}. Additionally, we can break them down as in \eqref{eq:unitary_decomp} with $R$ (denoted here $R_G$) being  a constant for each basic Gaussian gate. Then, we need to further decompose each component of $\eqref{eq:unitary_decomp}$ as 
\begin{align*}
    \exp(-i\vartheta_b \hat{O}_b') = (\exp(-i\varphi_b \hat{O}_b'))^{m_{G, b, \theta}},
\end{align*}
where $\varphi_b \in [-\frac{1}{4\sqrt{2}e}, \frac{1}{4\sqrt{2}e}]$. We will refer to $\sum_{b=1}^{R_G} m_{G, b, \theta}$ as the \emph{amplification factor} for the basic Gaussian gate $U_{G}(\theta)$.

Combining Lemma \ref{lem:unitary_error} with the efficient Hermite transform (Theorem \ref{thm:qht}) gives
the following digital implementation theorem for Gaussian optics.

\digitalImplTheorem*
\begin{proof}

After decomposing every Gaussian factor into bounded-time elementary
evolutions, let $r\leq RM$ be their total number and set
$E=5^rN$.  In the quadrature-grid representation write
\begin{align*}
    \overline U'=\overline W_r\cdots\overline W_1.
\end{align*}
The digital circuit is
\begin{align}
\label{eq:digital-gaussian-circuit-new}
    \overline U=
    \left(\bigotimes_{j=1}^K\mathsf{QHT}_{E,L,j}^{\dagger}\right)
    \overline U'
    \left(\bigotimes_{j=1}^K\mathsf{QHT}_{N,L,j}\right),
\end{align}
where the same operational grid dimension $L$ is used throughout.

Fix a normalized input $|\psi\rangle\in\operatorname{ran}(Q_N)$.  Applying
the first $K$ QHTs with individual error $\epsilon_H$ produces a state within
$K\epsilon_H$ of $\mathsf S_{N,L}|\psi\rangle$.  Unitarity of the grid
evolution preserves this error.  Lemma~\ref{lem:unitary_error}, followed by
the inverse QHTs at cutoff $E$, then gives
\begin{align}
\label{eq:main-gaussian-error-new}
    \|\overline U|\psi\rangle-\widehat U|\psi\rangle\|
    \leq 2K\epsilon_H+CN^2e^{-\overline\gamma N}.
\end{align}
Indeed, the second estimate in \eqref{eq:composed-tail-and-bridge} compares
the grid evolution with the sampled state associated with
$Q_E\widehat U|\psi\rangle$, while the first estimate bounds the discarded
exact Fock tail.  The adjoint QHT maps that sampled state back to its Fock
coefficients with another $K\epsilon_H$ error.

Take $\epsilon_H=\epsilon/(4K)$.  The QHT theorem at the largest effective
cutoff $E$ requires
\begin{align*}
    L=\Theta\!\left(
      K^{13/4}(5^rN)^{9/4}\epsilon^{-13/4}+\log K
    \right).
\end{align*}
This choice also satisfies the hypothesis
$L=\Omega(5^rN+\log K)$ of Lemma~\ref{lem:unitary_error}.  From \eqref{eq:main-gaussian-error-new}, the
condition
\begin{align*}
    N=\Omega\!\left(\log(1/\epsilon)\right)
\end{align*}
with a sufficiently large implicit constant makes that term at most
$\epsilon/2$, independently of $K$ and $r$.  Since $r\leq RM$, the choices in
the theorem suffice uniformly for every allowed decomposition.  Taking the
supremum over normalized inputs in $\operatorname{ran}(P_N)$ proves the
asserted operator-norm bound.

It remains to count gates.  The forward and inverse basis changes use $2K$
QHT circuits.  Theorem~\ref{thm:qht} bounds their total cost by
\begin{align*}
    \mathcal{O}\!\left(
      K\bigl(\log(5^{RM}N)+\log(K/\epsilon)\bigr)^3
      \log(K/\epsilon)
    \right).
\end{align*}
Each of the at most $RM$ grid evolutions is diagonal in the position basis up
to a constant number of $L$-point Fourier transforms and uses
$\mathcal{O}(\log^2L)$ one- and two-qubit gates.  This gives the stated gate
count.  The $K$ grid registers require $\mathcal{O}(K\log L)$ qubits.
\end{proof}

\section*{Acknowledgment}
We thank Rob Otter and Ruslan Shaydulin for executive support and valuable feedback on this project. AP thanks Alexander Krajenbrink and Harry Buhrman for helpful discussions. We also acknowledge our colleagues at the Global Technology Applied Research Center of JPMorganChase.

\section*{Disclaimer}
This paper was prepared for informational purposes by the Global Technology Applied Research center of JPMorgan Chase \& Co. This paper is not a product of the Research Department of JPMorgan Chase \& Co. or its affiliates. Neither JPMorgan Chase \& Co. nor any of its affiliates makes any explicit or implied representation or warranty and none of them accept any liability in connection with this paper, including, without limitation, with respect to the completeness, accuracy, or reliability of the information contained herein and the potential legal, compliance, tax, or accounting effects thereof. This document is not intended as investment research or investment advice, or as a recommendation, offer, or solicitation for the purchase or sale of any security, financial instrument, financial product or service, or to be used in any way for evaluating the merits of participating in any transaction.

\bibliographystyle{alpha}
\bibliography{references}

\appendix

\section{Log-grid discretization and error control}
\label{app:log-grid-analysis}
This appendix proves the two finite-grid facts used in Theorem~\ref{thm:fixed-order-laguerre}. First, one logarithmic grid approximates every Hankel transform that occurs during bounded-time radial evolution. Second, sampling on that grid approximately preserves inner products among the retained Laguerre eigenfunctions. The proof controls the tails before analyzing periodization and sampling, because the Fourier representation is accurate only after the relevant wavefunctions have been confined to a finite log interval.

For an integer $E\ge1$, let
\begin{align}
    P_E=\sum_{k=0}^{E-1}\ket{\phi_{k,\nu}}\!\bra{\phi_{k,\nu}},
    \qquad
    A_E\ket k=\mathsf S_LS\phi_{k,\nu}.
    \label{eq:log-grid-retained-maps}
\end{align}
The operator $P_E$ selects the radial eigenfunctions on which phase estimation must be accurate. The sampled-column map $A_E$ records how the same eigenfunctions appear on the finite logarithmic grid.

\subsection{Intermediate Hankel-chirp states and endpoint control}
\label{subsec:laguerre-radial-fftlog-details}
The radial factorization contains four Hankel transforms and three chirps. Controlling only its input and output would not justify applying the same finite Hankel block at the intermediate stages. We therefore place every intermediate state in a finite-dimensional family that remains stable under the exact Hankel and chirp operations. Let $U_1,\ldots,U_7$ be the factors in application order
\begin{align}
 U_1&=\mathcal H_\nu,&U_2&=C_{a(t)},&U_3&=\mathcal H_\nu,&U_4&=C_{b(t)},
 &U_5&=\mathcal H_\nu,&U_6&=C_{a(t)},&U_7&=\mathcal H_\nu
\end{align}
and set $G_0=I$, $G_j=U_j\cdots U_1$. For $\Re\alpha>0$, define the finite-dimensional Laguerre-Gaussian space
\begin{align}
    \mathcal Q_E(\alpha) = \{r^{\nu+1/2} P(r^2) e^{-\alpha r^2}:\deg P<E\}
\end{align}
Weber's integral gives
\begin{align}
    \mathcal H_\nu\mathcal Q_E(\alpha)=\mathcal Q_E((4\alpha)^{-1}),
\end{align}
and $C_\gamma$ sends $\mathcal Q_E(\alpha)$ to $\mathcal Q_E(\alpha+i\gamma)$. For $|t| \le \pi/2$, the intermediate states arising from the positive-time factorization, the negative-time factorization, and the same states after one additional Hankel transform belong to spaces $\mathcal Q_E(\alpha)$ whose width parameters satisfy
\begin{align}
    \frac14\le\Re\alpha\le\frac12,
    \qquad |\Im\alpha|\le\frac12
    \label{eq:controlled-prefix-width}
\end{align}
The singular endpoint $r=0$ (hard edge) is controlled uniformly over this compact width region. Removing the phase $e^{i\operatorname{Im}(\alpha)r^2}$ reduces the estimate to the Laguerre-Gaussian space with real Gaussian parameter
$\beta=2\operatorname{Re}(\alpha) \in[1/2,1]$ and orthonormal basis 
\begin{align}
    \psi_{k,\beta}(r)=
    \left(\frac{2\beta^{\nu+1}k!}{\Gamma(k+\nu+1)}\right)^{1/2}
    r^{\nu+1/2}e^{-\beta r^2/2}L_k^{(\nu)}(\beta r^2)
\end{align}
The fixed order Laguerre estimate,
\begin{align}
    e^{-y/2}|L_k^{(\nu)}(y)|
    \le L_k^{(\nu)}(0)
    =\frac{\Gamma(k+\nu+1)}{\Gamma(\nu+1)\Gamma(k+1)}
\end{align}
provides the following diagonal projection-kernel bound after summing over $0\le k<E$,
\begin{align}
    \sum_{k=0}^{E-1}|\psi_{k,\beta}(r)|^2
    \le C_\nu E^{\nu+1}r^{2\nu+1}.
    \label{eq:hard-edge-kernel}
\end{align}
Consequently, for every unit vector $F$ belonging to one of the spaces $\mathcal Q_E(\alpha)$ with $\alpha$ in the width range above,
\begin{align}
    \int_0^{r_-}|F(r)|^2\,dr
    \le C_\nu E^{\nu+1}r_-^{2\nu+2}
    \label{eq:hard-edge-tail}
\end{align}
The outer radial tail is estimated from the radial second moment of the same
intermediate states. For $F=G_j(\sigma t)f$, the bound $\int_0^\infty r^2|F(r)|^2\,dr\le C_\nu(E+1)\|f\|_2^2$ and Markov's inequality gives
\begin{align}
    \int_{r_+}^\infty |G_j(\sigma t)f(r)|^2\,dr
    \le C_\nu(E+1)r_+^{-2}\|f\|_2^2
    \label{eq:large-radius-tail}
\end{align}

\subsection{Periodization and sampling aliases}
The preceding tail bounds allow us to restrict attention to a finite radial range. Choose radii $0<r_-<r_+$ so that the near origin estimate~\eqref{eq:hard-edge-tail} and the outer tail estimate~\eqref{eq:large-radius-tail} fall within their assigned error budgets. Then choose a computational interval $I=[x_-,x_+]$ that contains $[\log r_-,\log r_+]$ with room for a smooth cutoff, and set $T=x_+-x_-$. Let $\omega_L$ be a $C^2$ window supported in $I$ and equal to one on
$[\log r_-,\log r_+]$. Let $\Theta(r) = \omega_L(\log r)$. Since the inverse square potential commutes with multiplication by $\Theta$, the product rule gives
\begin{align}
    B_\nu(\Theta F)
    =\Theta B_\nu F-2\Theta'F'-\Theta''F
\end{align}
The derivatives of $\Theta$ are supported only where the window varies and on that cutoff transition region, $F'$ is controlled by $F$ and $\mathcal A F=r\partial_rF+\frac12F$. The estimates for the intermediate states therefore control all three terms on the RHS, giving
\begin{align}
    \|B_\nu(\Theta F)\|_2
    \le C_\nu^{\rm cut}(E,\delta^{-1})\|f\|_2
    \label{eq:cutoff-Bnu-bound}
\end{align}

Put $u=SG_j(\sigma t)f$ and $v=\omega_Lu$. The function $v$ is supported in the finite interval $I$. Let $E_Tv$ denote its periodic extension with period $T$. The Hankel output is generally not compactly supported, so we fold its translated log coordinate intervals back into $I$ by
\begin{align}
    (P_T^\Sigma w)(x)=\sum_{\ell\in\Z}w(x+\ell T),
    \qquad x\in I
\end{align}
The bounds below show that this sum converges in the required $L^2$ norms. Use the Fourier basis on $I$ with origin at $x_-$,
\begin{align}
    e_k(x)=T^{-1/2}e^{i\omega_k^T(x-x_-)},
    \qquad \omega_k^T=\frac{2\pi k}{T}
\end{align}
Define $M_{\nu,T}$ on Fourier coefficients by
\begin{align}
    \langle e_k,M_{\nu,T}g\rangle_{L^2(I)}
    =e^{2i\omega_k^Tx_-}m_\nu(\omega_k^T)
    \langle e_{-k},g\rangle_{L^2(I)}
    \label{eq:exact-torus-multiplier}
\end{align}
The Mellin diagonalization gives the exact periodization identity
\begin{align}
    M_{\nu,T}E_Tv=P_T^\Sigma S\mathcal H_\nu S^{-1}v
    \label{eq:exact-periodization}
\end{align}
To justify the fold $P_T^\Sigma$, estimate the Hankel output on the intervals $I+\ell T$.  Let $r_{\rm out,+}=e^{x_+}$. On the intervals with $\ell\ge1$, the radial variable satisfies $r\ge r_{\rm out,+}e^{(\ell-1)T}$, and therefore $r^2\mathcal H_\nu=\mathcal H_\nu B_\nu$ together with \eqref{eq:cutoff-Bnu-bound} gives
\begin{align}
    \left\|1_{[r_{\rm out,+}e^{(\ell-1)T},\infty)}
    \mathcal H_\nu(\Theta F)\right\|_2
    \le r_{\rm out,+}^{-2}e^{-2(\ell-1)T}
    C_\nu^{\rm cut}(E,\delta^{-1})\|f\|_2
    \label{eq:right-cell-tail}
\end{align}
On the intervals $I-\ell T$, the fixed order Bessel estimate $|J_\nu(z)|\le C_\nu z^\nu$ gives decay proportional to $e^{-\ell(\nu+1)T}$. These tail bounds also apply after one log derivative as $\mathcal A=r\partial_r+\frac12$ is the generator of translations in the log coordinate and $\mathcal H_\nu\mathcal A\mathcal H_\nu=-\mathcal A$ transfers the same right and left tail estimates to $\mathcal A\mathcal H_\nu(\Theta F)$.

There are two finite sampling effects to control. First, sampling on the $L$-point grid folds together Fourier coefficients whose indices differ by multiples of $L$. Second, the multiplier error, which comes from applying one finite-grid phase to a group of Fourier coefficients. Before sampling, each coefficient $c_k$ carries the phase $m_\nu(\omega_k^T)$ attached to its own frequency. After sampling, all indices that differ by multiples of $L$ contribute to one grid frequency and receive a single finite-grid phase. If $E_Tv = \sum_{k\in\Z}c_ke_k$, then sampling on the $L$ point grid gives
\begin{align}
    \left(F_L(\sqrt{\Delta x}\,v(x_j))_{j=0}^{L-1}\right)_k
    =\sum_{\ell\in\Z}c_{-\kappa_L(k)+\ell L}.
    \label{eq:exact-spectral-aliasing}
\end{align}
This separates the two spectral errors. The sampling error is the contribution of the coefficients with $\ell\ne0$, while the multiplier error comes from assigning one finite grid phase to the whole coefficient sum instead of keeping the separate phases attached to the frequencies
$-\kappa_L(k)+\ell L$. Parseval and Cauchy-Schwarz give

\begin{align}
    \|\mathrm{sampling\ spectral\ alias}\|_2
    \le \frac{T}{2L}\|v'\|_2,
    \qquad
    \|\mathrm{multiplier\ spectral\ alias}\|_2
    \le \frac{T}{L}\|v'\|_2
    \label{eq:spectral-alias-bounds}
\end{align}

The finite block also has to respect the self-paired Nyquist frequency. Let $c_0=2x_-$ be the center appearing in the exact shifted-grid multiplier. Choose the integer $s$ nearest to
\begin{align}
    -\frac{\theta_\nu(\pi/\Delta x)}{\pi}-\frac{c_0}{\Delta x}
\end{align}
and set
\begin{align}
    c_{\rm lr}=\Delta x\left(-\frac{\theta_\nu(\pi/\Delta x)}{\pi}-s\right).
    \label{eq:low-ringing-center}
\end{align}
With certified rounding, $|c_{\rm lr}-c_0|\le3\Delta x/4$. For $1\le k<L/2$, use the diagonal phase in Section~\ref{sec:fftlog} and set the phase at $L-k$ to its complex conjugate. Set the zero-frequency phase to one and the Nyquist phase to $(-1)^s$. These choices imply $R_LDR_L=D^\dagger$ exactly. Replacing $c_0$ by $c_{\rm lr}$ changes the output by at most
\begin{align}
    \frac{3\Delta x}{4}\|v'\|_2,
    \label{eq:low-ringing-error}
\end{align}
because $|e^{-iy}-1|\le|y|$ and multiplication by frequency corresponds to one log derivative.

The remaining comparison is between the folded Hankel output and the windowed exact output on $I$. With $\mathcal M_\nu=S\mathcal H_\nu S^{-1}$,
\begin{align}
 P_T^\Sigma \mathcal M_\nu(\omega_Lu)-\omega_L\mathcal M_\nu u
 ={}&(1-\omega_L)\mathcal M_\nu u-\mathcal M_\nu((1-\omega_L)u)
 \notag\\
 &+\sum_{\ell\ne0}\mathcal M_\nu(\omega_Lu)(\cdot+\ell T)
 \label{eq:periodization-decomposition}
\end{align}
The first term is the part of the exact output outside the window, the second is the Hankel response to the input portion removed by the window, and the last term collects the shifted copies of the windowed output from $I+\ell T$ with $\ell\ne0$.

Sampling in the log coordinate is transferred by
\begin{align}
    \|w_L-w\|_{L^2(I)}\le \Delta x\|w'\|_{L^2(I)},
    \label{eq:H1-sampling}
\end{align}
where $w_L$ is the step function on the grid cells determined by the sampled values of $w$.

The interval length $T$ is chosen so that the shifted interval bounds following \eqref{eq:right-cell-tail} on intervals $I+\ell T$, together with the three terms in \eqref{eq:periodization-decomposition} contribute at most a fixed multiple of $\delta$ to the Hankel block error. Once $I$ and $\omega_L$ have been fixed, choose $L$ so that the two spectral errors in \eqref{eq:spectral-alias-bounds} and the grid reconstruction error in \eqref{eq:H1-sampling} are also at most fixed multiples of $\delta$.

The required derivative bound comes from the same quantities used in the cutoff estimate. For $F=G_j(\sigma t)f$, the bounds on $F$, on $\mathcal A F=r\partial_rF+\frac12F$, and on $B_\nu(\Theta F)$ give
\begin{align}
    \|v'\|_{L^2(I)}\le C_\nu(E+1)\|f\|_2,
    \qquad v=\omega_L S F
\end{align}
Therefore $L=O(\text{poly}_\nu(E,\delta^{-1}))$ suffices for the spectral aliasing and sampling estimates.  Applying the same grid reconstruction estimate to pairs of retained Laguerre states gives the sampled Gram bound used in the transform.

\begin{theorem}[Effective FFTLog control of radial evolution]
\label{thm:fftlog-control}
Fix an order $\nu\ge0$. For every $E\ge1$ and $0<\delta<1/2$, there are effective endpoints $x_-<x_+$, a $C^2$ window, and a grid $L=2^{q_L}$ satisfying
\begin{align}
    x_+-x_-=O_\nu\!\left(\log\frac{E+1}{\delta}\right),
    \qquad
    L=\operatorname{poly}_\nu(E,\delta^{-1})
\end{align}
such that the finite FFTLog block $W_{\nu,L}$ approximates the exact Hankel transform at every intermediate stage of the radial Hankel-chirp factorization. In particular, for all $|t|\le\pi/2$ and all $f\in\operatorname{ran} P_E$ with $\|f\|=1$,
\begin{align}
    \|P_L(t)\mathsf S_LSf-\mathsf S_LSe^{-itH_\nu}f\|_2\le 4\delta
\end{align}
where $P_L(t)$ is the finite radial evolution block obtained from Proposition~\ref{prop:radial-factorization}.
The same grid then certifies
\begin{align}
    \|A_E^*A_E-I\|\le\delta
\end{align}
\end{theorem}

\begin{proof}
Choose radii $0<r_-<r_+$ so that the near-origin estimate \eqref{eq:hard-edge-tail} and the outer-tail estimate \eqref{eq:large-radius-tail} are each at most a fixed constant times $\delta$. Then choose an interval $I=[x_-,x_+]$ that contains $[\log r_-,\log r_+]$ with room for the window transition region. The bounds in \eqref{eq:hard-edge-tail} and \eqref{eq:large-radius-tail} give
\begin{align}
    x_+-x_-=O_\nu\!\left(\log\frac{E+1}{\delta}\right),
\end{align} 
because the radial quantities in those estimates grow at most polynomially in $E+1$ and $\delta^{-1}$. 

With this interval fixed, choose its length $T=x_+-x_-$ so that the series in \eqref{eq:right-cell-tail}, the corresponding Bessel estimate on the intervals $I-\ell T$, and the three terms in \eqref{eq:periodization-decomposition} are each bounded by fixed constants times $\delta$. The estimate \eqref{eq:cutoff-Bnu-bound} controls the intervals to the right of $I$ through $r^2\mathcal H_\nu=\mathcal H_\nu B_\nu$, and the Bessel estimate controls the intervals to the left of $I$. The identity \eqref{eq:exact-periodization} then gives the periodized Hankel action on the finite interval.

After $I$ and the window have been fixed, choose $L$ so that the two spectral errors in \eqref{eq:spectral-alias-bounds}, the grid reconstruction error in \eqref{eq:H1-sampling}, and the low-ringing origin shift are each bounded by some constants times $\delta$. For $F=G_j(\sigma t)f$ and $v=\omega_LSF$, the bounds on $F$, on $\mathcal A F=r\partial_rF+\frac12F$, and on $B_\nu(\Theta F)$ used in \eqref{eq:cutoff-Bnu-bound} give
\begin{align}
    \|v'\|_{L^2(I)}\le C_\nu(E+1)\|f\|_2
\end{align}
Together with \eqref{eq:exact-spectral-aliasing}, this derivative bound gives $L=O(\text{poly}_\nu(E,\delta^{-1}))$ for the spectral aliasing and sampling estimates.

Only the four Hankel replacements contribute to approximation error in the finite radial evolution block. The chirps are diagonal on the same logarithmic grid, with phase $e^{-i\gamma e^{2x_j}}$ at the grid point $x_j$, so sampling commutes exactly with each radial chirp. The interval estimates, the spectral aliasing bounds, and the grid reconstruction bound control the replacement of $S\mathcal H_\nu S^{-1}$ by $W_{\nu,L}$ at each intermediate Hankel stage. Summing the resulting errors gives the displayed $4\delta$ bound for $P_L(t)$. Applying the grid reconstruction estimate uniformly to every unit vector in $\operatorname{ran}P_E$ shows that sampling preserves its squared norm up to $\delta$. Taking the supremum of the resulting quadratic-form bound gives $\|A_E^*A_E-I\|\le\delta$.
\end{proof}

\begin{corollary}[Radial evolution on the orthonormal sampled family]
\label{cor:laguerre-polar-evolution}
Let $G=A_E^*A_E$, $V_E=A_EG^{-1/2}$, and
\begin{align}
    D_t\ket k=e^{-it(2k+\nu+1)}\ket k,
    \qquad 0\le k<E.
\end{align}
If $\|G-I\|\le g<1$ and $\|P_L(t)A_E-A_ED_t\|\le\delta_t$, then
\begin{align}
    \|P_L(t)V_E-V_ED_t\|
    \le\frac{\delta_t}{\sqrt{1-g}}
    +2\sqrt{1+g}\left((1-g)^{-1/2}-1\right).
    \label{eq:polar-evolution-error}
\end{align}
\end{corollary}

\begin{proof}
Using $V_E=A_EG^{-1/2}$, add and subtract $A_ED_tG^{-1/2}$ to obtain
\begin{align}
    P_L(t)V_E-V_ED_t
    ={}&(P_L(t)A_E-A_ED_t)G^{-1/2}
    +A_E(D_tG^{-1/2}-G^{-1/2}D_t).
\end{align}
The first term is at most $\delta_t/\sqrt{1-g}$. For the second, use $\|A_E\|\le\sqrt{1+g}$ and
\begin{align}
    \|D_tG^{-1/2}-G^{-1/2}D_t\|
    \le2\|G^{-1/2}-I\|
    \le2\left((1-g)^{-1/2}-1\right).
\end{align}
This proves \eqref{eq:polar-evolution-error}.
\end{proof}

\subsection{Sampled states and spectral filtering}

The finite grid does not preserve the orthogonality of the Laguerre functions exactly. We introduce an orthonormal comparison family so that the output of a unitary circuit can be compared with the normalized sampled functions requested in Theorem~\ref{thm:fixed-order-laguerre}.

\begin{proposition}[Sampled trial-state preparation]
\label{prop:laguerre-sampled-preparation}
Let $E>N$ and suppose that the grid satisfies
\begin{align}
    g:=\|A_E^*A_E-I\|\le\frac14.
\end{align}
Set
\begin{align}
    V_E=A_E(A_E^*A_E)^{-1/2},
    \qquad
    \ket{v_k}=V_E\ket k,
\end{align}
and let $\mathsf B_E\ket k=\ket{\Phi_{k,\nu}^{(L)}}$. Then $V_E$ is an isometry and
\begin{align}
    \|\mathsf B_E-V_E\|\le2g.
    \label{eq:sampled-polar-comparison}
\end{align}
The grid can be refined so that there is a clean unitary which, controlled by $0\le n<N$, prepares a sampled WKB state $\ket{\mu_n}$ satisfying
\begin{align}
    \langle v_n|\mu_n\rangle&\ge c_{\rm trial}(\nu),\label{eq:sampled-trial-overlap}\\
    \|(I-V_EV_E^*)\ket{\mu_n}\|_2^2
    &\le C_\nu\left(\frac{n+1}{E+1}+\Delta x(E+1)\right).
    \label{eq:sampled-trial-leakage}
\end{align}
The preparation uses a number of gates and work qubits polynomial in $q_L$, $\log E$, and the requested number of precision bits, and it returns its arithmetic work registers to zero.
\end{proposition}

\begin{proof}
Put $G=A_E^*A_E$. The Gram assumption places the spectrum of $G$ in $[1-g,1+g]$, so $G^{-1/2}$ exists and $V_E^*V_E=I$. Functional calculus bounds $\|A_E-V_E\|$ by $1-\sqrt{1-g}$. Individual normalization of the columns of $A_E$ changes them by at most another constant multiple of $g$, which proves \eqref{eq:sampled-polar-comparison}.

Appendix~\ref{app:laguerre-wkb-analysis} constructs a continuum trial state with overlap at least $2c_{\rm trial}(\nu)$ with $\phi_{n,\nu}$ and squared spectral tail at most $C_\nu(n+1)/(E+1)$. The sampling estimate \eqref{eq:H1-sampling} preserves the overlap after a sufficiently fine grid is chosen. Applied to the high-energy tail, the same estimate contributes the second term in \eqref{eq:sampled-trial-leakage}. Equation~\eqref{eq:sampled-polar-comparison} then transfers the overlap from the normalized sampled Laguerre function to $\ket{v_n}$.

It remains to explain why the trial state can be prepared without loading an $L$-entry table. Its WKB expression is supported on an interval of fixed width in the log coordinate. Controlled by $n$, the circuit computes the first grid point in that interval, prepares a uniform superposition over the $R=\Theta_\nu(\Delta x^{-1})$ possible offsets, evaluates the signed WKB amplitude at each selected point, and applies one controlled rotation. Choose a constant $M_{\rm amp}$ larger than the absolute value of every trial amplitude. The probability of the accepted rotation branch is
\begin{align}
    p_n=\frac{1}{RM_{\rm amp}^2}\sum_j|\mathcal G_{n,\nu}(x_j)|^2.
\end{align}
The rectangle estimate gives $\Delta x\sum_j|\mathcal G_{n,\nu}(x_j)|^2=\Theta_\nu(1)$, while $R\Delta x=\Theta_\nu(1)$. Hence $p_n$ is bounded below by a positive constant depending only on $\nu$. Fixed-point amplitude amplification therefore prepares $\ket{\mu_n}$ with a number of repetitions logarithmic in the target error. Reversing the pointwise calculations clears every work register. Each evaluated function is smooth on the fixed WKB interval, so fixed-precision arithmetic has cost polynomial in the stated register lengths.
\end{proof}

The next proposition explains why the trial state can be converted into the requested degree-$n$ state without measuring $n$.

\begin{proposition}[Spectral filtering and coherent label erasure]
\label{prop:laguerre-spectral-filter}
Let $N<E<M=2^m$, and suppose that a phase-estimation labeler approximates
\begin{align}
    \ket{v_k}\ket0_{\rm qpe}\longmapsto\ket{v_k}\ket k_{\rm qpe},
    \qquad 0\le k<E,
\end{align}
to operator-norm error at most $\lambda$. Let
\begin{align}
    d=\max_{0\le n<N}\|(I-V_EV_E^*)\ket{\mu_n}\|_2.
\end{align}
If $\lambda+d+2g\le c_{\rm trial}(\nu)/4$, then equality testing against the input degree followed by one common fixed-point amplification sequence prepares every $\ket{v_n}$ coherently. For fixed-point error $\rho$, the preparation error is
\begin{align}
    O_\nu(\rho+\lambda+d+g).
    \label{eq:spectral-filter-error}
\end{align}
The input degree and the phase-estimation register can then be erased coherently by labeling $\ket{v_n}$, subtracting $n$ modulo $M$, and reversing the labeler.
\end{proposition}

\begin{proof}
Expand the retained part of $\ket{\mu_n}$ in the orthonormal family $\{\ket{v_k}\}_{k<E}$. Since $n<N<E<M$, equality modulo $M$ selects only the term $k=n$. Equations~\eqref{eq:sampled-polar-comparison} and \eqref{eq:sampled-trial-overlap} show that its amplitude is positive and bounded below by a constant depending only on $\nu$. The discarded part contributes at most $d$, and the imperfect labeler contributes at most $\lambda$. These bounds prove that the equality-marked component remains separated from zero.

A fixed-point amplification polynomial can therefore be chosen for the common interval of possible marked amplitudes. The same polynomial acts in every degree block, so the error for a superposition of degrees is the largest block error rather than the sum of the block errors. This proves \eqref{eq:spectral-filter-error} and prevents an $n$-dependent output phase.

For label erasure, an ideal labeler maps $\ket{v_n}\ket0$ to $\ket{v_n}\ket n$. Modular subtraction clears the original input label, and the inverse labeler clears its own output. The subtraction is an exact reversible permutation. The two uses of the approximate labeler add at most $2\lambda$ to the error.
\end{proof}

\subsection{Resource estimates}

The FFTLog block needs the phase $\theta_\nu(\omega_k)$ on the frequency register, the chirp factors need the angles $\gamma e^{2x_j}$ on the grid register, and the trial-state loader needs rotations derived from its amplitude. Shifted phase estimation also applies the controlled scalar phase $e^{i\pi2^j(\nu+1)/M}$ and compares the recovered degree with the input degree. For the FFTLog multiplier, write
$m_\nu(\omega)=e^{i\theta_\nu(\omega)}$ with
\begin{align}
    \theta_\nu(\omega)
    =-\omega\log2
    -2\operatorname{Im}\operatorname{Log}\Gamma
      \left(\frac{\nu+1}{2}+\frac{i\omega}{2}\right)
\end{align} 
To evaluate this phase to $p$ bits, set
\begin{align}
    z=\frac{\nu+1}{2}+\frac{i\omega}{2}
\end{align} 
For fixed $\nu$, choose $q=O_\nu(p)$ so that $z+q$ lies in a region where the Stirling expansion has a certified remainder. The Gamma recurrence gives
\begin{align}
\operatorname{Log}\Gamma(z)
=\operatorname{Log}\Gamma(z+q)-\sum_{j=0}^{q-1}\operatorname{Log}(z+j)
\end{align} 
The term $\operatorname{Log}\Gamma(z+q)$ is then evaluated with $O(p)$ Bernoulli terms, with remainder at most $2^{-p}$. The number of recurrence factors depends on the target precision and on the fixed order $\nu$ and is independent of $|\omega|$.

\begin{proposition}[Resource bounds]
\label{prop:resource-ledger}
Fix an order $\nu$. Let
\begin{align}
    W=O_\nu(q_L+p+\log(M+2)+\log(\epsilon^{-1}+2)).
\end{align}
The phase function $\theta_\nu(\omega_k)$, the chirp angles $\gamma e^{2x_j}$, the controlled offset phases, and the trial-state rotation angles can each be evaluated reversibly to operator error at most $2^{-p}$. The filter comparison uses the map
\begin{align}
    \ket{k}\ket n\ket0\longmapsto
    \ket{k}\ket n\ket{k-n\bmod M}
\end{align}
followed by a zero test on the last register and by uncomputation of that register. Each arithmetic or rotation subroutine uses $O_\nu(W^{12})$ gates and $O_\nu(W^{12})$ qubits. Each grid Fourier transform uses $O(q_L^2)$ gates and $O(q_L)$ qubits and is therefore dominated by these bounds. In Theorem~\ref{thm:fixed-order-laguerre}, choosing $p=O_\nu(B)$ so that every operator error is at most $\epsilon/(100C_\nu(B+1)^8)$ gives a full circuit with $O_\nu(B^{15})$ gates and $O_\nu(B^{12})$ qubits.
\end{proposition}

\begin{proof}
Because the order is fixed, the constants used by these subroutines are fixed before the reversible subroutines are generated. For each of $\nu$, $\pi$, $\log2$, $x_-$, $x_+$, and $\Delta x$, for each endpoint of a window transition region, and for each Stirling coefficient used in the Gamma expansion, choose rational lower and upper bounds that differ by at most $2^{-p-O_\nu(1)}$. The phase and rotation evaluations use fixed-precision reversible arithmetic with enough extra bits that the final angle differs from the target value by at most $2^{-p}$. The resulting angles are applied by kickback. The comparison register for $k-n\bmod M$ is computed, tested for zero, and uncomputed. Standard reversible arithmetic and single-qubit rotation approximations give the local $O_\nu(W^{12})$ gate and qubit bound.

The phase-estimation register has $m=O_\nu(B)$ qubits and therefore uses $O_\nu(B)$ controlled radial propagators. Both fixed-point amplification procedures have degree $O_\nu(B)$ because their initial success amplitudes are bounded below by constants and their target errors are inverse exponential in $B$. Expanding the labelers and trial-state reflections inside these amplification sequences gives at most $C_\nu(B+1)^3$ primitive calls. Set the operator error of each call to at most $\epsilon/(100C_\nu(B+1)^8)$. The triangle inequality for products of approximate unitaries keeps the total contribution from rounding and rotation approximation errors inside the error budget. With this choice of local precision one has $W=O_\nu(B)$, which gives the stated $O_\nu(B^{15})$ total gate bound. The same work registers are reused across calls, so the maximum number of logical qubits is $O_\nu(B^{12})$.
\end{proof}

\section{Circular beam splitters } \label{App:A}

\cardToRad*
\begin{proof}
The phase correction needed below can also be written using Cartesian gates. Indeed,
\begin{align}
    N_- = \frac{1}{2}\left(a_1^{\dagger}a_1+a_2^{\dagger}a_2+i(a_1^{\dagger}a_2-a_2^{\dagger}a_1)\right),
\end{align}
so
\begin{align}
    \exp(i\theta N_-)=R_1(\theta/2)R_2(\theta/2)\textup{BS}_{Y}(\theta/2),
\end{align}
where $R_j(\theta)=\exp(i\theta a_j^{\dagger}a_j)$ is the phase rotation on the $j$-th Cartesian mode.
Let $G=a_1a_2^{\dagger}+a_2a_1^{\dagger}$, so $U(\theta)=\exp(i\theta G)$. The commutators
\begin{align}
    [G,a_1^{\dagger}]=a_2^{\dagger}, \qquad [G,a_2^{\dagger}]=a_1^{\dagger}
\end{align}
imply
\begin{align}
    U(\theta)a_1^{\dagger}U(\theta)^{\dagger}=\cos(\theta)a_1^{\dagger}+i\sin(\theta)a_2^{\dagger}, \qquad
    U(\theta)a_2^{\dagger}U(\theta)^{\dagger}=i\sin(\theta)a_1^{\dagger}+\cos(\theta)a_2^{\dagger}.
\end{align}
Setting $U=U(\pi/4)$ gives $U a_1^{\dagger}U^{\dagger}=a_+^{\dagger}$ and $U a_2^{\dagger}U^{\dagger}=i a_-^{\dagger}$. By the Cartesian decomposition above, $P:=R_1(-\pi/4)R_2(-\pi/4)\textup{BS}_{Y}(-\pi/4)=\exp(-i\pi N_-/2)$. This gate leaves $a_+^{\dagger}$ unchanged and satisfies $P a_-^{\dagger}P^{\dagger}=-i a_-^{\dagger}$. Therefore the combined operation $\widetilde{U}=PU$ satisfies
\begin{align}
    \widetilde{U} a_1^{\dagger}\widetilde{U}^{\dagger}=a_+^{\dagger}, \qquad \widetilde{U} a_2^{\dagger}\widetilde{U}^{\dagger}=a_-^{\dagger}.
\end{align}

Since $\widetilde{U}\ket{0,0}=\ket{0,0}$, we have
\begin{align}
    \widetilde{U}\ket{n_1,n_2}
    &=\widetilde{U}\frac{(a_1^{\dagger})^{n_1}(a_2^{\dagger})^{n_2}}{\sqrt{n_1!n_2!}}\ket{0,0} \notagbreak
    &=\frac{(\widetilde{U} a_1^{\dagger}\widetilde{U}^{\dagger})^{n_1}(\widetilde{U} a_2^{\dagger}\widetilde{U}^{\dagger})^{n_2}}{\sqrt{n_1!n_2!}}\ket{0,0} \notagbreak
    &=\frac{(a_+^{\dagger})^{n_1}(a_-^{\dagger})^{n_2}}{\sqrt{n_1!n_2!}}\ket{0,0}
    =\ket{n_1,n_2}_{+}.
\end{align}
By the circular-mode discussion in the radial oscillator notation section, $\ket{n_1,n_2}_{+}$ has radial quantum number $n_r=\min(n_1,n_2)$ and angular momentum $m=n_1-n_2$. This proves the claimed radial-mode identification.
\end{proof}

\section{Displacement operator and Charlier transform} \label{App:B}
The Charlier family admits an additional structure not shared by Meixner or
Krawtchouk: its polynomials are generated from the constant function by powers of
a single backward shift. This is special to the Poisson weight and arises from the bare
factor $e^{t}$ in the generating function \eqref{gfc}.

\charlierShift*

\begin{proof}
Let $A:=1-\frac n\mu T_-$, so the claim $C_k(n;\mu)=A^k\mathbf 1$ is equivalent,
after packaging all $k$ into one exponential generating function, to
\all{
e^{tA}\mathbf 1=\sum_{k\ge0}\frac{t^k}{k!}A^k\mathbf 1
=\sum_{k\ge0}C_k(n;\mu)\frac{t^k}{k!}.
}{opgen}
Since $A=I-\frac n\mu T_-$, we have $e^{tA}=e^t e^{-(t/\mu)nT_-}$. Evaluating the
powers of $nT_-$ on the constant function gives
$(nT_-)^r\mathbf 1=n^{\underline r}$, the falling factorial, so
$$
e^{-(t/\mu)nT_-}\mathbf 1=\sum_{r\ge0}\frac{(-t/\mu)^r}{r!}\,n^{\underline r}
=\left(1-\frac t\mu\right)^n,
$$
the last equality being the binomial expansion
$(1-x)^n=\sum_{r\ge0}\frac{n^{\underline r}}{r!}(-x)^r$ with $x=t/\mu$. Hence
$e^{tA}\mathbf 1=e^t\left(1-\frac t\mu\right)^n$, which is exactly \eqref{gfc},
establishing \eqref{opgen}.
\end{proof}

\charlierDisp*
\begin{proof}
Applying the intertwining identity
$D(\alpha)\,a^\dagger\,D(\alpha)^\dagger=a^\dagger-\bar\alpha$ to
$D(\alpha)\ket k=D(\alpha)\,\frac{(a^\dagger)^k}{\sqrt{k!}}\ket0$ gives
\al{
D(\alpha)\ket k
&=\frac{(a^\dagger-\bar\alpha)^k}{\sqrt{k!}}\,D(\alpha)\ket0 \notagbreak
&=\frac{(a^\dagger-\bar\alpha)^k}{\sqrt{k!}}\,
   e^{-\mu/2}\sum_{n\ge0}\frac{\alpha^n}{\sqrt{n!}}\ket n,
}
where the second equality uses the coherent-state expansion
$D(\alpha)\ket0=e^{-\mu/2}\sum_{n\ge0}\frac{\alpha^n}{\sqrt{n!}}\ket n$.
Define
$$
\phi_k(n)=\frac{\langle n|D(\alpha)|k\rangle}{e^{-\mu/2}\alpha^{n}/\sqrt{n!}},
$$
so that $\phi_0(n)=1$ and the theorem is equivalent to
$\phi_k(n)=\frac{(-\bar\alpha)^k}{\sqrt{k!}}\,C_k(n;\mu)$.

It is convenient to pass to discrete variables and regard states as functions
$f:\Z_{\ge0}\to\C$ supported on the number basis. In this representation the
creation operator acts as
\all{
a^\dagger f
= a^\dagger\sum_{n\ge0} f(n)\ket n
= \sum_{n\ge0} f(n)\sqrt{n+1}\,\ket{n+1}
= \sum_{n\ge1} f(n-1)\sqrt{n}\,\ket n,
} {five}
so that $(a^\dagger f)(n)=\sqrt n\,f(n-1)=\sqrt n\,(T_-f)(n)$, where
$(T_-f)(n)=f(n-1)$ with $f(-1):=0$. Writing $\alpha=\sqrt\mu\,e^{i\theta}$,
define the diagonal operator
$$
(W_\alpha f)(n)
= e^{in\theta}\sqrt{w_\mu(n)}\,f(n),
\qquad
w_\mu(n)=e^{-\mu}\frac{\mu^n}{n!},
$$
which satisfies $W_\alpha 1=D(\alpha)\ket0$. Conjugating $a^\dagger$ by
$W_\alpha$,
\begin{align*}
W_{\alpha}f(n) &= e^{in\theta}\sqrt{w_{\mu}(n)}\,f(n),\\
a^{\dagger}W_{\alpha}f(n) &= \sqrt{n}\,e^{i(n-1)\theta}\sqrt{w_{\mu}(n-1)}\,f(n-1),\\
W_{\alpha}^{-1}a^{\dagger}W_{\alpha}f(n)
  &= \sqrt{n}\,e^{-i\theta}\sqrt{w_{\mu}(n-1)/w_{\mu}(n)}\,f(n-1)
   = \frac{n}{\alpha}\,f(n-1),
\end{align*}
where the last step uses $w_\mu(n-1)/w_\mu(n)=n/\mu$ and
$\tfrac{n}{\sqrt\mu}e^{-i\theta}=n/\alpha$, since $\alpha=\sqrt\mu\,e^{i\theta}$.
Hence
$$
(W_\alpha^{-1}a^\dagger W_\alpha f)(n)=\frac n\alpha f(n-1)=\frac n\alpha (T_-f)(n),
$$
and, using the hinge identity $\alpha\bar\alpha=\mu$,
$$
W_\alpha^{-1}(a^\dagger-\bar\alpha)W_\alpha
=\frac n\alpha T_- -\bar\alpha
=\bar\alpha\!\left(\frac n\mu T_- - I\right).
$$
It is the identity $\alpha\bar\alpha=\mu$ that collapses the coefficient of
$T_-$ to $n/\mu$ and turns the conjugated operator into the Charlier ladder.
Therefore
$$
\phi_k(n)
=\frac{1}{\sqrt{k!}}
 \left[W_\alpha^{-1}(a^\dagger-\bar\alpha)W_\alpha\right]^k 1(n)
=\frac{\bar\alpha^{\,k}}{\sqrt{k!}}
 \left[\left(\frac n\mu T_- - I\right)^k 1\right](n).
$$
By Proposition~\ref{thm:cop},
$C_k(n;\mu)=\left(I-\frac n\mu T_-\right)^k 1(n)$, so
$\left(\frac n\mu T_- - I\right)^k 1=(-1)^k C_k(n;\mu)$, giving
$$
\phi_k(n)=\frac{(-\bar\alpha)^{k}}{\sqrt{k!}}\,C_k(n;\mu),
$$
which is equivalent to the theorem statement.
\end{proof}

\section{Squeeze operator and the Meixner transform} \label{App:D}
\meixnerSqueeze*
\begin{proof}
    We compute the transition amplitude from $\ket{N, k}$ to $\ket{M, k}$ under
    the gate $S(z)$. Starting from annihilation:
    \begin{align}
        K_- \, \ket{n, k} = \sqrt{n(n+2k - 1)}\ket{n-1, k}
    \end{align}
    After $j=N-m$ applications,
    \begin{align}
        (K_-)^j \, \ket{N, k} = \prod_{i=0}^{j-1}\sqrt{(N-i)(N-i+2k - 1)}\ket{N-j, k}
    \end{align}
    Setting $m=N-j$,
    \begin{align}
        e^{-z^*K_-} \, \ket{N, k} = \sum_{m=0}^N \frac{(-z^*)^{N-m}}{(N-m)!} \sqrt{\frac{N!\frac{(2k+N-1)!}{(2k-1)!}}{m! \frac{(2k+m-1)!}{(2k-1)!}}} \ket{m, k}
    \end{align}
    Application of $(1-|z|^2)^{K_0}$:
    \begin{align}
        K_0 \ket{m, k} &= (k+m) \ket{m, k} \\
        (1-|z|^2)^{K_0} \ket{m, k} &= (1-|z|^2)^{k+m}\ket{m,k}
    \end{align}
    The creating part:
    \begin{align}
        (K_+) \, \ket{n, k} = \sqrt{(n+1)(n+2k)}\ket{n+1, k}
    \end{align}
    After $l$ applications from $\ket{m, k}$:
    \begin{align}
        (K_+)^l \, \ket{m, k} = \sqrt{\frac{(m+l)!}{m!} \frac{\frac{(2k+m+l-1)!}{(2k-1)!}}{\frac{(2k+m-1)!}{(2k-1)!}}} \ket{m+l, k}
    \end{align}
    With $l=M-m$ and projecting,
    \begin{align}
        \bra{M, k}{e^{zK_+}}\ket{m,k} = \frac{z^{M-m}}{(M-m)!} \sqrt{\frac{M!\frac{(2k+M-1)!}{(2k-1)!}}{m!\frac{(2k+m-1)!}{(2k-1)!}}}
    \end{align}
    Therefore,
    \begin{align*}
        \bra{M, k} S(z) \ket{N, k} &=
        (1-|z|^2)^{k} \sqrt{N!M! \tfrac{(2k+N-1)!}{(2k-1)!} \tfrac{(2k+M-1)!}{(2k-1)!}} \sum_{m=0}^{\min(N,M)} \frac{(-z^*)^{N-m}}{(N-m)!} \frac{z^{M-m} (1-|z|^2)^m}{(M-m)!\,m!\, \frac{(2k+m-1)!}{(2k-1)!}} \\
        &= (-1)^N z^{*N} z^M (1-|z|^2)^k \sqrt{\tfrac{(2k)_N}{N!}\tfrac{(2k)_M}{M!}} \sum_{m} \frac{(-N)_m(-M)_m}{m!\,(2k)_m} \left(\frac{-(1-|z|^2)}{|z|^2}\right)^m \\
        &= (-1)^N z^{*N} z^M (1-|z|^2)^k \sqrt{\tfrac{(2k)_N}{N!}\tfrac{(2k)_M}{M!}}\; {}_2F_1\!\left(-N, -M; 2k;\; -\frac{1-|z|^2}{|z|^2}\right) \\
        &= (-1)^N z^{*N} z^M (1-|z|^2)^k \sqrt{\binom{N+2k-1}{N} \binom{M+2k-1}{M}} \, M_N(M; 2k, |z|^2)
    \end{align*}
    where the last line uses ${}_2F_1(-N,-M;2k;1-1/|z|^2)=M_N(M;2k,|z|^2)$.
\end{proof}

\section{Beamsplitter operator and the Krawtchouk transform} \label{App:C}

\begin{proposition}
Let $w_p(n) = \binom{N}{n} p^n q^{N-n}$ be the weights for a binomial distribution with parameter $p$, then the Krawtchouk polynomials satisfy the orthogonality relation 
\[ 
\sum_{n=0}^N w_p(n) K_j(n;p,N) K_k(n;p,N) = \frac{1}{\binom{N}{k}} \left(\frac{q}{p}\right)^k \delta_{jk}. 
\]
\end{proposition}

\begin{proof}
    Define,
    \begin{align}
        W(s, t) &= \sum_{j,k=0}^{N} \left[\sum_{n=0}^{N} w_p(n) K_j(n;p,N) K_k(n;p,N)\right] \binom{N}{j} \binom{N}{k} s^j t^k \\
        &= \sum_{n=0}^{N} w_p(n) \left(1 - \frac{q}{p}s\right)^n (1+s)^{N-n} \left(1 - \frac{q}{p}t\right)^n (1+t)^{N-n}
    \end{align}
    We know that, 
    \begin{align}
        \sum_{n=0}^{N} w_p(n) u^n v^{N-n} = (pu + qv)^N
    \end{align}
    for the binomial moment generating function with $u = (1 - \frac{q}{p}s)(1 - \frac{q}{p}t)$ and $v = (1+s)(1+t)$. Then, 
    \begin{align}
        W(s,t) &= \left(p\left(1 - \frac{q}{p}s\right)\left(1 - \frac{q}{p}t\right) + q(1+s)(1+t)\right)^N \\
        &= \left(p - qs - qt + \frac{q^2}{p}st + q + qs + qt + qst\right)^N \\
        &= \left(1 + \frac{q^2 + pq}{p} st\right)^N  = \left(1 + \frac{q}{p} st\right)^N
    \end{align}
    Expanding $\left(1 + \frac{q}{p} st\right)^N = \sum_{k=0}^{N} \binom{N}{k} \left(\frac{q}{p}\right)^k s^k t^k$ and comparing coefficients of $s^j t^k$ yields, 
    \begin{align}
        \sum_{n=0}^{N} w_p(n) K_j(n;p,N) K_k(n;p,N) = \frac{1}{\binom{N}{k}^2}\binom{N}{k} \left(\frac{q}{p}\right)^k \delta_{jk} = \frac{1}{\binom{N}{k}} \left(\frac{q}{p}\right)^k \delta_{jk}
    \end{align}
\end{proof}

\begin{proposition}[Beam splitter intertwining]
Let $B(\theta) = \exp[\theta(a_1^\dagger a_2 - a_1 a_2^\dagger)]$ be the beam splitter operator acting on two bosonic modes. Then the following intertwining identities hold: 
\al{
B(\theta)\, a_1^\dagger\, B(\theta)^\dagger &= \cos\theta\, a_1^\dagger - \sin\theta\, a_2^\dagger, \notagbreak
B(\theta)\, a_2^\dagger\, B(\theta)^\dagger &= \sin\theta\, a_1^\dagger + \cos\theta\, a_2^\dagger.
}
\end{proposition}

\begin{proof}
Let $A = \theta(a_1^\dagger a_2 - a_1 a_2^\dagger)$. By BCH, $B(\theta) a_j^\dagger B(\theta)^\dagger = a_j^\dagger + [A, a_j^\dagger] + \frac{1}{2!}[A, [A, a_j^\dagger]] + \cdots$. Using $[a_i, a_j^\dagger] = \delta_{ij}$ between the two modes, 
\[ 
[A, a_1^\dagger] = -\theta a_2^\dagger, \qquad [A, a_2^\dagger] = \theta a_1^\dagger. 
\] 
 The higher commutators cycle: $[A, [A, a_1^\dagger]] = -\theta^2 a_1^\dagger$, $[A, [A, a_2^\dagger]] = -\theta^2 a_2^\dagger$, etc. Re-summing produces the rotation $\begin{pmatrix} \cos\theta & -\sin\theta \\ \sin\theta & \cos\theta \end{pmatrix}$ acting on $(a_1^\dagger, a_2^\dagger)^T$.  
\end{proof}

\kravchuckBs*
\begin{proof}
    Using the Beam Splitter intertwining identity, $B(\theta)(a_j^\dagger)^m B(\theta)^\dagger = (B(\theta) a_j^\dagger B(\theta)^\dagger)^m$:
    \begin{align}
        B(\theta) |k, N{-}k\rangle = \frac{(\cos\theta\, a_1^\dagger - \sin\theta\, a_2^\dagger)^k (\sin\theta\, a_1^\dagger + \cos\theta\, a_2^\dagger)^{N-k}} {\sqrt{k!(N-k)!}} \ket{0,0}
    \end{align}
    Writing $c = \cos\theta$, $s = \sin\theta$ and summing over $k$
    \begin{align}
        \sum_{k=0}^N \frac{t^k}{\sqrt{k!(N-k)!}} B(\theta)|k,N{-}k\rangle = \frac{1}{N!} \left[(tc+s)\,a_1^\dagger + (c - ts)\, a_2^\dagger\right]^N  \ket{0,0}
    \end{align}
    Projecting onto $\bra{ n, N{-}n}$,
    \begin{align}
        \Phi_n(t) &= \sum_{k=0}^N \frac{t^k}{\sqrt{k!(N-k)!}} \bra{n, N{-}n}B(\theta) \ket{ k, N{-}k } \\
        &= \frac{(tc+s)^n (c-ts)^{N-n}}{\sqrt{n!(N-n)!}}
    \end{align}
    With $\tau = -ts/c$, so that $tc/s = -\tau c^2/s^2 = -\tau q/p$.  Then $tc+s = s(1+tc/s) = s(1-q\tau/p)$ and $c-ts = c(1+\tau)$, yields:
    \begin{align}
        \Phi_n(t) = \frac{s^n c^{N-n}}{\sqrt{n!(N-n)!}} \left(1-\frac{q}{p}\tau\right)^n (1+\tau)^{N-n}
    \end{align}

    By the generating function, $(1-q\tau/p)^n(1+\tau)^{N-n} = \sum_{k=0}^N \binom{N}{k} K_k(n;p,N) \tau^k$. Since $\tau^k = (-s/c)^k t^k = (-1)^k (p/q)^{k/2} t^k$ (using $s^2=p$, $c^2=q$)
    \begin{align}
        \Phi_n(t) = \frac{s^n c^{N-n}}{\sqrt{n!(N-n)!}} \sum_{k=0}^N \binom{N}{k} K_k(n;p,N)\, (-1)^k (p/q)^{k/2}\, t^k
    \end{align}
    Comparing coefficients of $t^k$ and multiplying by $\sqrt{k!(N-k)!}$
    \begin{align}
        \bra{ n, N{-}n} B(\theta) \ket{ k, N{-}k } = (-1)^k \binom{N}{k} \frac{\sqrt{k!(N-k)!}}{\sqrt{n!(N-n)!}}\, p^{(n+k)/2} q^{(N-n-k)/2}\, K_k(n;p,N)
    \end{align}
    Finally, 
    \begin{align}
        \binom{N}{k} \frac{\sqrt{k!(N-k)!}}{\sqrt{n!(N-n)!}} &= \frac{N!}{\sqrt{k!(N-k)!\, n!(N-n)!}} \\
        &= \sqrt{\binom{N}{n}\binom{N}{k}}
    \end{align}
\end{proof}

\section{Hahn polynomials and Clebsch-Gordan coefficients} \label{app:hahn-cg}
Throughout, $J_z\ket{j,m}=m\ket{j,m}$ and
$J_\pm\ket{j,m}=\sqrt{(j\mp m)(j\pm m+1)}\ket{j,m\pm1}$, the coefficient being
read as zero when $m\pm1\notin\{-j,\ldots,j\}$, and
$\mathbf J^2:=J_z^2+\tfrac12(J_+J_-+J_-J_+)$ acts on $V_j$ as $j(j+1)$.  On
$V_{j_1}\otimes V_{j_2}$ a subscript $i\in\{1,2\}$ indicates the tensor factor
acted on, and $\mathbf J^2:=(\mathbf J_1+\mathbf J_2)^2$ with
$J_z=J_{1z}+J_{2z}$, $J_\pm=J_{1\pm}+J_{2\pm}$; the space
$\mathcal V_M^{j_1,j_2}$, the parametrisation $m_1=M-m_--x$, $m_2=m_-+x$, the
weight $\rho$ and the constants $h_\ell$ are as in
Section~\ref{subsec:hahn_transform}.
\lemHahnCG*
\begin{proof}
Set
\begin{align*}
    s:=\frac{\alpha+\beta}{2}.
\end{align*}
The endpoints of the allowed $m_2$ interval are determined by
$-j_2\leq m_2\leq j_2$ and $-j_1\leq M-m_2\leq j_1$.  For the stated
parameters,
\begin{align*}
    M-j_1
    &=\frac{\beta-\alpha}{2}
      -\frac{N+\alpha+\beta}{2}
      =-\frac N2-\alpha,\\
    M+j_1
    &=\frac{\beta-\alpha}{2}
      +\frac{N+\alpha+\beta}{2}
      =\frac N2+\beta.
\end{align*}
Since $\alpha,\beta\geq0$, it follows that
\begin{align*}
    m_-&=\max\left(-j_2,M-j_1\right)=-\frac N2,\\
    m_+&=\min\left(j_2,M+j_1\right)=\frac N2.
\end{align*}
Moreover,
\begin{align*}
    j_1-j_2=s,
    \qquad
    |M|=\frac{|\beta-\alpha|}{2}\leq s,
\end{align*}
and hence the smallest total spin occurring on the fixed-$M$ shell, $\mathcal{V}_{M}^{j_1, j_2}$, is
$J_{\min}=s$.  For $x\in\{0,\ldots,N\}$, write
\begin{align}
\label{eq:hahn-cg-product-basis}
    \ket{x}
    &:={}
    \ket{j_1,m_1(x)}\ket{j_2,m_2(x)},\\
    m_1(x)&=M+\frac N2-x
      =\frac{N-\alpha+\beta}{2}-x,
    \qquad
    m_2(x)=-\frac N2+x.
\end{align}

We briefly recall the operators used below.  On the standard basis of $V_j$,
define
\begin{align*}
    J_z\ket{j,m}
    &=m\ket{j,m},\\
    J_+\ket{j,m}
    &=\sqrt{(j-m)(j+m+1)}\ket{j,m+1},\\
    J_-\ket{j,m}
    &=\sqrt{(j+m)(j-m+1)}\ket{j,m-1}.
\end{align*}
Thus $J_+$ and $J_-$ increase and decrease the label $m$ by one, respectively;
the coefficient vanishes when the resulting label would lie outside
$\{-j,\ldots,j\}$.  The quadratic operator
\begin{align*}
    \mathbf J^2
    :=J_z^2+\frac12(J_+J_-+J_-J_+)
\end{align*}
satisfies $\mathbf J^2\ket{j,m}=j(j+1)\ket{j,m}$.  On
$V_{j_1}\otimes V_{j_2}$, a subscript indicates the tensor factor on which an
operator acts; for example,
$J_{1-}J_{2+}=J_-\otimes J_+$.  The combined operators are
$J_z=J_{1z}+J_{2z}$ and $J_\pm=J_{1\pm}+J_{2\pm}$.

We next compute the matrix of the corresponding total quadratic operator
$\mathbf J^2=(\mathbf J_1+\mathbf J_2)^2$ in the basis
\eqref{eq:hahn-cg-product-basis}.  Expanding it gives
\begin{align*}
    \mathbf J^2
    =\mathbf J_1^2+\mathbf J_2^2
     +2J_{1z}J_{2z}+J_{1+}J_{2-}+J_{1-}J_{2+},
\end{align*}
the diagonal entry is
\begin{align}
\label{eq:hahn-cg-diagonal}
    d_x
    :=j_1(j_1+1)+j_2(j_2+1)+2m_1(x)m_2(x).
\end{align}
The operator $J_{1-}J_{2+}$ sends $\ket{x}$ to $\ket{x+1}$ with coefficient
\begin{align*}
    u_x
    &:=\sqrt{
       (j_1+m_1(x))(j_1-m_1(x)+1)
       (j_2-m_2(x))(j_2+m_2(x)+1)}\\
    &=\sqrt{
       (N+\beta-x)(\alpha+x+1)(N-x)(x+1)}.
\end{align*}
Similarly, $J_{1+}J_{2-}$ sends $\ket{x}$ to $\ket{x-1}$ with coefficient
$u_{x-1}$.  We set $u_{-1}=u_N=0$ to include the endpoints.

Let
\begin{align*}
    J=s+n,
    \qquad
    c_n(x):=
    C_{j_1,m_1(x);j_2,m_2(x)}^{s+n,M}.
\end{align*}
The equation $\mathbf J^2\ket{J,M}=J(J+1)\ket{J,M}$ therefore gives
\begin{align}
\label{eq:cg-casimir-recurrence}
    u_xc_n(x+1)+u_{x-1}c_n(x-1)+d_xc_n(x)
    =(s+n)(s+n+1)c_n(x).
\end{align}

We now compare this with the Hahn difference equation.  Define
\begin{align*}
    p_x&:=(x+\alpha+1)(N-x),\\
    q_x&:=x(N+\beta+1-x),\\
    \lambda_n&:=n(n+\alpha+\beta+1).
\end{align*}
Expanding the forward and backward differences in the Hahn equation yields
\begin{align}
\label{eq:hahn-three-term-x}
    p_xQ_n(x+1)+q_xQ_n(x-1)
    +\bigl(\lambda_n-p_x-q_x\bigr)Q_n(x)=0.
\end{align}
Indeed, the coefficient of $Q_n(x+1)$ is
\begin{align*}
&x(N+\beta+1-x)
 +N(\alpha+1)-(\alpha+\beta+2)x\\
&\hspace{3cm}=(x+\alpha+1)(N-x)=p_x,
\end{align*}
whereas the coefficient of $Q_n(x-1)$ is $q_x$; the remaining coefficient is
$-p_x-q_x+\lambda_n$.

The Hahn weight in \eqref{eq:hahn_weight} satisfies the detailed-balance
identity
\begin{align}
\label{eq:hahn-detailed-balance}
    \frac{\rho(x+1)}{\rho(x)}
    &=\frac{\alpha+x+1}{x+1}
      \frac{N-x}{N+\beta-x}
      =\frac{p_x}{q_{x+1}}.
\end{align}
Set
\begin{align*}
    \phi_n(x):=\widehat Q_n(x;\alpha,\beta,N)
    =\frac{\sqrt{\rho(x)}}{\sqrt{h_n}}Q_n(x;\alpha,\beta,N).
\end{align*}
Multiplying \eqref{eq:hahn-three-term-x} by
$\sqrt{\rho(x)/h_n}$ and using \eqref{eq:hahn-detailed-balance} gives
\begin{align}
\label{eq:normalized-hahn-recurrence}
    &\sqrt{p_xq_{x+1}}\,\phi_n(x+1)
    +\sqrt{q_xp_{x-1}}\,\phi_n(x-1)\notag\\
    &\hspace{2.5cm}
    +\bigl(\lambda_n-p_x-q_x\bigr)\phi_n(x)=0.
\end{align}
The off-diagonal coefficients are exactly those of the Casimir recurrence:
\begin{align*}
    p_xq_{x+1}
    &=(x+\alpha+1)(N-x)(x+1)(N+\beta-x)
      =u_x^2,\\
    q_xp_{x-1}&=u_{x-1}^2.
\end{align*}

It remains to compare the diagonal terms.  First,
\begin{align}
\label{eq:px-qx-expanded}
    p_x+q_x
    =N(\alpha+1)+2Nx+(\beta-\alpha)x-2x^2.
\end{align}
Since $j_1=N/2+s$, $j_2=N/2$, $m_1(x)=M+N/2-x$, and
$m_2(x)=x-N/2$, direct substitution into \eqref{eq:hahn-cg-diagonal} gives
\begin{align*}
    d_x
    &={}
    \left(\frac N2+s\right)
    \left(\frac N2+s+1\right)
    +\frac N2\left(\frac N2+1\right)
    +2\left(M+\frac N2-x\right)\left(x-\frac N2\right)\\
    &={}
    s(s+1)+N(\alpha+1)+2Nx+(\beta-\alpha)x-2x^2\\
    &=s(s+1)+p_x+q_x.
\end{align*}
Also, because $\alpha+\beta=2s$,
\begin{align*}
    (s+n)(s+n+1)
    &=s(s+1)+n(n+2s+1)\\
    &=s(s+1)+\lambda_n.
\end{align*}
Consequently,
\begin{align}
\label{eq:hahn-cg-diagonal-match}
    (s+n)(s+n+1)-d_x=\lambda_n-p_x-q_x.
\end{align}

Define $g_n(x):=(-1)^xc_n(x)$.  Substituting
$c_n(x)=(-1)^xg_n(x)$ into \eqref{eq:cg-casimir-recurrence}, multiplying by
$-1$, and using \eqref{eq:hahn-cg-diagonal-match} yields
\begin{align*}
    u_xg_n(x+1)+u_{x-1}g_n(x-1)
    +\bigl(\lambda_n-p_x-q_x\bigr)g_n(x)=0.
\end{align*}
This is precisely \eqref{eq:normalized-hahn-recurrence}.  For $N\geq1$, all
$u_x$ with $0\leq x<N$ are positive, so this irreducible tridiagonal
eigenvalue problem has a one-dimensional eigenspace for each eigenvalue
$\lambda_n$.  Both vectors are normalized:
\begin{align*}
    \sum_{x=0}^{N}|g_n(x)|^2&=1,\\
    \sum_{x=0}^{N}|\phi_n(x)|^2
    &=\frac{1}{h_n}\sum_{x=0}^{N}\rho(x)Q_n(x)^2=1.
\end{align*}
It follows that $g_n=\pm\phi_n$.

Finally, the sign is fixed at $x=0$.  Since $Q_n(0)=1$ and
$\rho(0),h_n>0$, we have $\phi_n(0)>0$.  We use the following form of the
Condon--Shortley factorial formula:
\begin{align*}
&C_{j_1,m_1;j_2,m_2}^{J,M}\\
&={}
\left[
\frac{
 (2J+1)(J+j_1-j_2)!(J-j_1+j_2)!(j_1+j_2-J)!
}{
 (j_1+j_2+J+1)!
}
\right]^{1/2}\\
&\quad\times
\left[
 (J+M)!(J-M)!
 (j_1-m_1)!(j_1+m_1)!
 (j_2-m_2)!(j_2+m_2)!
\right]^{1/2}\\
&\quad\times
\sum_z
\frac{(-1)^z}{
 z!(j_1+j_2-J-z)!(j_1-m_1-z)!(j_2+m_2-z)!}\\
&\hspace{3.8cm}\times
\frac{1}{
 (J-j_2+m_1+z)!(J-j_1-m_2+z)!},
\end{align*}
where the sum is over those integers for which every factorial is
nonnegative.  At $x=0$, one has $m_2=-j_2$, and the factor
$(j_2+m_2-z)!=(-z)!$ forces $z=0$.  The remaining factorial arguments are
\begin{align*}
    J+j_1-j_2&=\alpha+\beta+n,
    &J-j_1+j_2&=n,\\
    j_1+j_2-J&=N-n,
    &j_1+j_2+J+1&=N+\alpha+\beta+n+1,\\
    J+M&=\beta+n,
    &J-M&=\alpha+n,\\
    j_1-m_1(0)&=\alpha,
    &j_1+m_1(0)&=N+\beta,\\
    j_2-m_2(0)&=N,
    &j_2+m_2(0)&=0,\\
    J-j_2+m_1(0)&=\beta+n,
    &J-j_1-m_2(0)&=n.
\end{align*}
Substitution and cancellation therefore give
\begin{align*}
    c_n(0)
    =\left[
    \frac{
      (\alpha+\beta+2n+1)
      (\alpha+\beta+n)!(\alpha+n)!(\beta+N)!N!
    }{
      (\alpha+\beta+N+n+1)!(N-n)!\alpha!(\beta+n)!n!
    }
    \right]^{1/2}>0.
\end{align*}
Hence $g_n(0)=c_n(0)>0$, so $g_n=\phi_n$.  Recalling the definitions of
$g_n$, $\phi_n$, $m_1(x)$, and $m_2(x)$ proves
\begin{align*}
    \widehat Q_n(x;\alpha,\beta,N)
    =(-1)^x
    C_{j_1,M-m_--x;j_2,m_-+x}^{J_{\min}+n,M}.
\end{align*}
For $N=0$, both sides are the unique normalized one-dimensional vector and the
same conclusion is immediate.
\end{proof}
\section{Computation of the Racah Rotation Amplitudes}
\label{sec:racah_computation}

For a general Racah transformation, the change-of-basis coefficient is
\begin{align}
\label{eq:racah-6j-hahn}
    F_{ed}^{abcf}
    &=\braket{(a(bc)_d)_f}{((ab)_e c)_f}\\
    &=(-1)^{a+b+c+f}\sqrt{(2e+1)(2d+1)}
      \begin{Bmatrix}
          a&b&e\\
          c&f&d
      \end{Bmatrix}.
\end{align}
We now evaluate this formula for $c=\frac12$.  For
\begin{align*}
    ((a b)_e\,\tfrac12)_f
    \longleftrightarrow
    (a(b\,\tfrac12)_d)_f,
\end{align*}
the admissible intermediate labels are
\begin{align*}
    e=f\pm\frac12,
    \qquad
    d=b\pm\frac12,
\end{align*}
subject to the remaining triangle constraints.  Write
\begin{align*}
        e_-&:=f-\frac12,
        &e_+&:=f+\frac12,\\
        d_+&:=b+\frac12,
        &d_-&:=b-\frac12,
\end{align*}
and define
\begin{align*}
        \varepsilon&:=(-1)^{a+b+f+\frac12},\\
        R_+&:=(a+b+f+\tfrac32)(-a+b+f+\tfrac12),\\
        R_-&:=(a-b+f+\tfrac12)(a+b-f+\tfrac12),\\
        D&:=(2b+1)(2f+1).
\end{align*}
The needed evaluations follow from the two spin-$\frac12$ special cases in
the NIST Digital Library of Mathematical Functions, together with the standard
symmetries of the $6j$ symbol.\footnote{See the NIST Digital Library of
Mathematical Functions, Section~34.5(i), Eqs.~(34.5.2)--(34.5.3), for the two
spin-$\frac12$ evaluations, and Section~34.5(ii), Eq.~(34.5.8), for the
symmetries used here.  See also D.~A.~Varshalovich, A.~N.~Moskalev, and
V.~K.~Khersonskii, \emph{Quantum Theory of Angular Momentum}, World
Scientific, 1988, Section~9.6.}
For the $e_+$ row, Eq.~(34.5.2) with
\begin{align*}
    (j_1,j_2,j_3)=\left(a,b,f+\frac12\right)
\end{align*}
gives the $(e_+,d_+)$ symbol directly, while Eq.~(34.5.3) with the same
substitution gives the $(e_+,d_-)$ symbol.  The $e_-$ row follows by first
using the symmetries
\begin{align*}
    \begin{Bmatrix}
        a&b&e_-\\
        \frac12&f&d_+
    \end{Bmatrix}
    &=
    \begin{Bmatrix}
        a&f&d_+\\
        \frac12&b&e_-
    \end{Bmatrix},\\
    \begin{Bmatrix}
        a&b&e_-\\
        \frac12&f&d_-
    \end{Bmatrix}
    &=
    \begin{Bmatrix}
        a&e_-&b\\
        \frac12&d_-&f
    \end{Bmatrix}.
\end{align*}
The first expression is Eq.~(34.5.3) with
$(j_1,j_2,j_3)=(a,f,b+\frac12)$, and the second is Eq.~(34.5.2) with
$(j_1,j_2,j_3)=(a,f-\frac12,b)$.  In the notation of those formulas, their
phase exponent is $j_1+j_2+j_3$.  It equals
$a+b+f+\frac12$ in the first three cases and
$a+b+f-\frac12$ in the $(e_-,d_-)$ case, accounting for the sole minus sign
below.  After these substitutions, the four values are
\begin{align}
\label{eq:half-6j-values-hahn}
\begin{Bmatrix}
        a&b&e_-\\
        \frac12&f&d_+
\end{Bmatrix}
&=\varepsilon
    \sqrt{\frac{R_+}{D(2f)(2b+2)}},\\
\begin{Bmatrix}
        a&b&e_-\\
        \frac12&f&d_-
\end{Bmatrix}
&=-\varepsilon
    \sqrt{\frac{R_-}{D(2f)(2b)}},\\
\begin{Bmatrix}
        a&b&e_+\\
        \frac12&f&d_+
\end{Bmatrix}
&=\varepsilon
    \sqrt{\frac{R_-}{D(2f+2)(2b+2)}},\\
\begin{Bmatrix}
        a&b&e_+\\
        \frac12&f&d_-
\end{Bmatrix}
&=\varepsilon
    \sqrt{\frac{R_+}{D(2f+2)(2b)}}.
\end{align}
Since $c=\frac12$, the phase in \eqref{eq:racah-6j-hahn} is precisely
$\varepsilon$.  Also,
\begin{align*}
        2e_-+1&=2f,
        &2e_++1&=2f+2,\\
        2d_++1&=2b+2,
        &2d_-+1&=2b.
\end{align*}
Substituting these identities and \eqref{eq:half-6j-values-hahn} into
\eqref{eq:racah-6j-hahn} gives, entry by entry,
\begin{align*}
        F_{e_-,d_+}^{ab,\frac12,f}
        &=\varepsilon\sqrt{(2f)(2b+2)}
            \left(\varepsilon
            \sqrt{\frac{R_+}{D(2f)(2b+2)}}\right)
            =\sqrt{\frac{R_+}{D}},\\
        F_{e_-,d_-}^{ab,\frac12,f}
        &=\varepsilon\sqrt{(2f)(2b)}
            \left(-\varepsilon
            \sqrt{\frac{R_-}{D(2f)(2b)}}\right)
            =-\sqrt{\frac{R_-}{D}},\\
        F_{e_+,d_+}^{ab,\frac12,f}
        &=\varepsilon\sqrt{(2f+2)(2b+2)}
            \left(\varepsilon
            \sqrt{\frac{R_-}{D(2f+2)(2b+2)}}\right)
            =\sqrt{\frac{R_-}{D}},\\
        F_{e_+,d_-}^{ab,\frac12,f}
        &=\varepsilon\sqrt{(2f+2)(2b)}
            \left(\varepsilon
            \sqrt{\frac{R_+}{D(2f+2)(2b)}}\right)
            =\sqrt{\frac{R_+}{D}},
\end{align*}
where $\varepsilon^2=1$.  Hence, in the ordered old basis
$\{e_-,e_+\}$ and ordered new basis $\{d_+,d_-\}$, the Racah matrix is
\begin{align}
\label{eq:spin-half-racah-matrix-hahn}
    F(a,b,f)
    =\begin{pmatrix}
        \cos\vartheta&-\sin\vartheta\\
        \sin\vartheta&\cos\vartheta
      \end{pmatrix},
\end{align}
where
\begin{align}
\label{eq:spin-half-racah-angle-hahn}
    \cos\vartheta
    &=\sqrt{\frac{
       (a+b+f+\frac32)(-a+b+f+\frac12)}
       {(2b+1)(2f+1)}},\\
    \sin\vartheta
    &=\sqrt{\frac{
       (a-b+f+\frac12)(a+b-f+\frac12)}
       {(2b+1)(2f+1)}}.
\end{align}
Here the row is indexed by $e$ and the column by $d$.  The two numerators add
to the denominator:
\begin{align*}
&(a+b+f+\tfrac32)(-a+b+f+\tfrac12)\\
&\quad +(a-b+f+\tfrac12)(a+b-f+\tfrac12)
=(2b+1)(2f+1),
\end{align*}
so the displayed matrix is orthogonal.

For the $t$-th layer, the correspondence with
the general formula is
\begin{align*}
        a&=j_1,
        &b&=S_{t-1},
        &c&=\frac12,\\
        e&=K_{t-1},
        &f&=K_t,
        &d&=S_t.
\end{align*}
Therefore the ordered old and new labels are, respectively,
\begin{align*}
        K_{t-1}&\in
        \left\{K_t-\frac12,K_t+\frac12\right\},\\
        S_t&\in
        \left\{S_{t-1}+\frac12,S_{t-1}-\frac12\right\}.
\end{align*}
Substituting $a=j_1$, $b=S_{t-1}$, and $f=K_t$ into
\eqref{eq:spin-half-racah-angle-hahn} gives
\begin{align*}
        \cos\vartheta_t
        &=\sqrt{\frac{
            (j_1+S_{t-1}+K_t+\frac32)
            (-j_1+S_{t-1}+K_t+\frac12)}
            {(2S_{t-1}+1)(2K_t+1)}},\\
        \sin\vartheta_t
        &=\sqrt{\frac{
            (j_1-S_{t-1}+K_t+\frac12)
            (j_1+S_{t-1}-K_t+\frac12)}
            {(2S_{t-1}+1)(2K_t+1)}}.
\end{align*}
In particular, the coefficient stated above is the upper-left entry,
\begin{align*}
        F_{K_t-\frac12,S_{t-1}+\frac12}^{
            j_1,S_{t-1},\frac12,K_t}
        =\cos\vartheta_t.
\end{align*}
The other three entries are
\begin{align*}
        F_{K_t-\frac12,S_{t-1}-\frac12}^{
            j_1,S_{t-1},\frac12,K_t}
        &=-\sin\vartheta_t,\\
        F_{K_t+\frac12,S_{t-1}+\frac12}^{
            j_1,S_{t-1},\frac12,K_t}
        &=\sin\vartheta_t,\\
        F_{K_t+\frac12,S_{t-1}-\frac12}^{
            j_1,S_{t-1},\frac12,K_t}
        &=\cos\vartheta_t.
\end{align*}
At a triangle boundary one channel is absent, and the corresponding
one-dimensional transformation is used.  Thus the circuit computes
$\vartheta_t$ from $j_1,S_{t-1},K_t$, applies the phase-consistent
single-qubit rotation, and uncomputes its arithmetic workspace.

\section{WKB analysis for the radial Laguerre trial state}
\label{app:laguerre-wkb-analysis}

The WKB trial state is localized inside the oscillatory region of the degree-$n$ Laguerre function. Its natural radial scale is $\Lambda_{n,\nu}=4n+2\nu+2$, so we use the dimensionless position
\begin{align}
    F_{n,\nu}(x)&=(S\phi_{n,\nu})(x),
    &\Lambda_{n,\nu}&=4n+2\nu+2,
    &z_{n,\nu}(x)&=\frac{e^{2x}}{\Lambda_{n,\nu}}.
\end{align}
The phase and slowly varying amplitude in the fixed-order Laguerre asymptotics are
\begin{align}
    \xi(z)&=\frac12\left(\sqrt{z-z^2}+\arcsin\sqrt z\right),
    &q_{\rm WKB}(z)&=\left(\frac{z}{1-z}\right)^{1/4},\\
    \Theta_{n,\nu}(z)&=\Lambda_{n,\nu}\xi(z)-\frac{\pi\nu}{2}-\frac\pi4.
\end{align}
Fix a real $C^2$ function $\chi$ supported on $[1/4,3/4]$, equal to one on $[1/3,2/3]$, and satisfying $0\le\chi\le1$. This cutoff stays a fixed distance from both turning points. We form the localized oscillatory approximation
\begin{align}
    U_{n,\nu}(x)&=\chi(z_{n,\nu}(x))q_{\rm WKB}(z_{n,\nu}(x))\cos\Theta_{n,\nu}(z_{n,\nu}(x)),\\
    Z_{n,\nu}&=\|U_{n,\nu}\|_2^2,
    &G_{n,\nu}&=Z_{n,\nu}^{-1/2}U_{n,\nu}.
\end{align}

\begin{Lemma}[WKB trial state]
\label{lem:laguerre-wkb-trial}
For each fixed $\nu\ge0$, there exist a threshold $n_\nu$, one fixed normalized nonnegative function $G_\nu^{\rm lo}\in C_c^2(\R)$, and constants $c_{\rm trial}(\nu)>0$ and $C_{\rm trial}(\nu)\ge1$ with the following properties. Set $\mathcal G_{n,\nu}=G_\nu^{\rm lo}$ for $0\le n<n_\nu$ and $\mathcal G_{n,\nu}=G_{n,\nu}$ for $n\ge n_\nu$, and define $\psi_n=S^{-1}\mathcal G_{n,\nu}$. Then
\begin{align}
    \langle\phi_{n,\nu},\psi_n\rangle&\ge2c_{\rm trial}(\nu),\\
    \|(I-P_K)\psi_n\|_2^2&\le C_{\rm trial}(\nu)\frac{n+1}{K+1},
    \qquad K\ge n+1.
\end{align}
\end{Lemma}

\begin{proof}
The fixed-order Frenzen--Wong expansion \cite{frenzen1988uniform}, specialized to $[a,b]=[1/4,3/4]$, gives uniformly in $z\in[a,b]$ 
\footnote{see also NIST DLMF
Section 18.15(iv), equations 18.15.17--18.15.20.}
\begin{align}
 L_n^{(\nu)}(\Lambda_{n,\nu}z)
 ={}&\frac{e^{\Lambda_{n,\nu}z/2}}
 {2^\nu z^{\nu/2+1/4}(1-z)^{1/4}}
 \sqrt{\frac{2}{\pi\Lambda_{n,\nu}}}\\
 &\quad\times
 \left(\cos\Theta_{n,\nu}(z)
 +O_\nu(\Lambda_{n,\nu}^{-1})\right).
 \label{eq:laguerre-wkb-uniform-input}
\end{align}
The cited remainder majorant is explicit in the fixed order, compact interval,
and truncation index. It therefore supplies a constant and threshold for fixed $\nu$. Substitution
of \eqref{eq:laguerre-wkb-uniform-input} into the normalized function
$F_{n,\nu}=S\phi_{n,\nu}$ cancels the exponential factors and gives
\begin{align}
 F_{n,\nu}\!\left(\frac12\log(\Lambda_{n,\nu}z)\right)
 =c_{n,\nu}q_{\rm WKB}(z)
 \left(\cos\Theta_{n,\nu}(z)+R_{n,\nu}(z)\right),
 \qquad
 \|R_{n,\nu}\|_\infty\le\frac{C_\nu}{\Lambda_{n,\nu}},
\end{align}
where
\begin{align}
 c_{n,\nu}=\frac{2^{1-\nu}}{\sqrt\pi}
 \sqrt{\frac{\Gamma(n+1)}{\Gamma(n+\nu+1)}}
 \Lambda_{n,\nu}^{\nu/2}\longrightarrow\frac2{\sqrt\pi}.
\end{align}

Direct differentiation gives
\begin{align}
 \xi'(z)=\frac12\sqrt{\frac{1-z}{z}},
 \qquad
 \frac d{dz}\bigl(2\Theta_{n,\nu}(z)\bigr)
 =\Lambda_{n,\nu}\sqrt{\frac{1-z}{z}}.
\end{align}
Since $dx=dz/(2z)$, one integration by parts gives, for $r=1,2$,
\begin{align}
 \left|\int_a^b\chi(z)^r q_{\rm WKB}(z)^2
 e^{2i\Theta_{n,\nu}(z)}\frac{dz}{2z}\right|
 \le\frac{C_r}{\Lambda_{n,\nu}},
 \qquad
 C_r=\frac12\int_a^b
 \left|\frac d{dz}\frac{\chi(z)^r}{1-z}\right|dz.
\end{align}
There is no boundary term because $\chi(a)=\chi(b)=0$.  Writing
$\cos^2\Theta=(1+\cos2\Theta)/2$ shows that
\begin{align}
 Z_{n,\nu}
 =\frac12\int_a^b\chi(z)^2q_{\rm WKB}(z)^2\frac{dz}{2z}
 +O_\nu(\Lambda_{n,\nu}^{-1})
\end{align}
is bounded above and below by positive constants.  Pairing the bulk formula
with $G_{n,\nu}$ similarly gives
\begin{align}
 \langle F_{n,\nu},G_{n,\nu}\rangle
 =\frac{c_{n,\nu}}{2\sqrt{Z_{n,\nu}}}
 \int_a^b\chi(z)q_{\rm WKB}(z)^2\frac{dz}{2z}
 +O_\nu(\Lambda_{n,\nu}^{-1}).
\end{align}
The main term is strictly positive.  An effective $n_\nu$ therefore gives a
uniform positive lower bound for all $n\ge n_\nu$.  For the finitely many
lower labels, $L_m^{(\nu)}(0)>0$.  Continuity supplies one common compact
interval on which every $F_{m,\nu}$, $m<n_\nu$, is positive.  This interval
can be found effectively: evaluate the finitely many computable Laguerre
polynomials by directed interval arithmetic on successive dyadic intervals
$[-q-1,-q]$ in the log coordinate until positivity is certified.  Choose a
dyadic subinterval of the first certified interval and place on it a fixed
nonnegative piecewise-quintic $C_c^2$ bump.  Its normalization is a computable
constant, so this gives the stated efficiently computable
$G_\nu^{\rm lo}$.  The finitely many overlap integrals are positive and admit
certified lower bounds by interval quadrature.  Choose
$c_{\rm trial}(\nu)$ smaller than one half of the resulting common lower
bound.  Then
\begin{align}
 \langle\phi_{n,\nu},\psi_n\rangle
 =\langle F_{n,\nu},\mathcal G_{n,\nu}\rangle
 \ge2c_{\rm trial}(\nu)
\end{align}
for every $n\ge0$.

It remains to bound spectral leakage.  The compact trial states belong to the
Friedrichs form domain.  In log coordinates, for
$g=\mathcal G_{n,\nu}$,
\begin{align}
 \langle S^{-1}g,H_\nu S^{-1}g\rangle
 =\frac12\int_\R\Bigl[
 &e^{-2x}|g'(x)-g(x)/2|^2
 +(\nu^2-1/4)e^{-2x}|g(x)|^2\\
 &+e^{2x}|g(x)|^2\Bigr]dx.
\end{align}
On the bulk support, $e^{2x}=\Lambda_{n,\nu}z$ and
\begin{align}
 \frac d{dx}\Theta_{n,\nu}(z_{n,\nu}(x))
 =\Lambda_{n,\nu}
 \sqrt{z_{n,\nu}(x)(1-z_{n,\nu}(x))}.
\end{align}
The functions $\chi$, $q_{\rm WKB}$, and their first derivatives are bounded
on $[a,b]$, while the normalization $Z_{n,\nu}$ is bounded away from zero.
Thus, on the support of a bulk trial state,
\begin{align}
 |g(x)|&\le C_\nu,&
 |g'(x)|&\le C_\nu\Lambda_{n,\nu},\\
 e^{-2x}&=\frac{1}{\Lambda_{n,\nu}z}=O(\Lambda_{n,\nu}^{-1}),&
 e^{2x}&=\Lambda_{n,\nu}z=O(\Lambda_{n,\nu}).
\end{align}
Its support has the fixed log width $\frac12\log(b/a)$.  Substitution in the
quadratic form gives
\begin{align}
 \langle\psi_n,H_\nu\psi_n\rangle
 \le C_\nu\Lambda_{n,\nu}
 \le C_\nu'(n+1).
\end{align}
The common low branch has fixed finite energy, so the same bound holds after
enlarging $C_\nu'$.  Expand
$\psi_n=\sum_{k\ge0}c_k\phi_{k,\nu}$.  Every $k\ge K$ has energy at
least $2K+\nu+1$, and hence
\begin{align}
 \|(I-P_K)\psi_n\|_2^2
 &=\sum_{k\ge K}|c_k|^2
 \le\frac{1}{2K+\nu+1}
 \sum_{k\ge K}(2k+\nu+1)|c_k|^2\\
 &\le\frac{\langle\psi_n,H_\nu\psi_n\rangle}{2K+\nu+1}
 \le C_{\rm trial}(\nu)\frac{n+1}{K+1}.
\end{align}
This proves both assertions of the lemma.
\end{proof}

\end{document}